\newcommand{\figref}[1]{Fig.\ \ref{#1}}
\newcommand{\subfigref}[2]{Fig.\ \ref{#1}(#2)}
\newcommand{\subFigref}[2]{Figure \ref{#1}(#2)}
\newcommand{\subfigdref}[3]{Fig.\ \ref{#1}(#2)--(#3)}
\newcommand{\Figref}[1]{Figure \ref{#1}}
\newcommand{\figsref}[2]{Figs.\ \ref{#1} and \ref{#2}}

\newcommand{\secref}[1]{Sec.\ \ref{#1}}
\newcommand{\Secref}[1]{Section \ref{#1}}

\newcommand{\Appref}[1]{Appendix \ref{#1}}

\newcommand{\secsref}[2]{Sec.\ \ref{#1} and \ref{#2}}
\newcommand{\eqnref}[1]{Eq.\ (\ref{#1})}
\newcommand{\Eqnref}[1]{Equation (\ref{#1})}
\newcommand{\Eqnsref}[2]{Equations (\ref{#1}) and (\ref{#2})}
\newcommand{\eqnsref}[2]{Eqs.\ (\ref{#1}) and (\ref{#2})}
\newcommand{\subfigwidth}[0]{1.6in}
\newcommand{\subfigwidthnew}[0]{2.0in}
\newcommand{\subfigwidthtoo}[0]{2.3in}
\newcommand{\Tcross}[0]{T_{\times}}
\newcommand{\Tcrossq}[0]{T_{\times}^{(1/q)}}
\newcommand{\Tcrossp}[0]{T_{\times}^{(p)}}
\newcommand{\Cq}[1]{C_{#1}}
\newcommand{\Cqq}[0]{C_{q}}
\newcommand{\Bq}[1]{B_{#1}}
\newcommand{\Bqq}[0]{B_{q}}
\newcommand{\xpo}[0]{\left(x+1\right)}
\newcommand{\xpoy}[0]{\left(x+y+1\right)}
\newcommand{\xpop}[1]{\left(x+1\right)^{#1}}
\newcommand{\xpoyp}[1]{\left(x+y+1\right)^{#1}}
\newcommand{\trixy}[0]{\left(x^2+x+y\right)}

\newcommand{\tri}[0]{\Delta}
\newcommand{\trip}[0]{\Delta'}

\newcommand{\Kl}[1]{K_{#1}}
\newcommand{\Kld}[2]{K_{#1}^{\left(#2\right)}}
\newcommand{\Klm}[2]{G_{#1}^{\left[#2\right]}}
\newcommand{\Pochhammer}[2]{\frac{\Gamma\left(#1\right)}{\Gamma\left(#2\right)}}
\newcommand{\ie}[0]{\emph{i.e.}}

\pdfoutput=1
\documentclass[aps,pre,twocolumn,showpacs,amsmath,amssymb]{revtex4}
\usepackage{graphicx,dcolumn,bm,verbatim,color,hyperref,subfigure,amsmath,amsthm}

\theoremstyle{plain}  \newtheorem{lemma}{Lemma}

\begin{document}

\title{The stability to instability transition in the structure of large
scale networks}

\begin{abstract}
We examine phase transitions between the ``easy,'' ``hard,'' and the
``unsolvable'' phases when attempting to identify structure in large
complex networks (``community detection'') in the presence of
disorder induced by network ``noise'' (spurious links that obscure
structure), heat bath temperature $T$, and system size $N$. The
partition of a graph into $q$ optimally disjoint subgraphs or
``communities'' inherently requires Potts type variables. In earlier
work [Phil. Mag. {\bf 92}, 406 (2012)] when examining power law and
other networks (and general associated Potts models), we illustrated
transitions in the computational complexity of the community
detection problem typically correspond to spin-glass-type
transitions (and transitions to chaotic dynamics in mechanical
analogs) at both high and low temperatures and/or noise. When
present, transitions at low temperature or low noise correspond to
entropy driven (or ``order by disorder'') annealing effects wherein
stability may initially increase as temperature or noise is
increased before becoming unsolvable at sufficiently high
temperature or noise. Additional transitions between contending
viable solutions (such as those at different natural scales) are
also possible. Identifying community structure via a dynamical
approach
%(some other dynamical methods include \cite{ref:arenassync},\cite{ref:gudkov})
where ``chaotic-type'' transitions were earlier found. The
correspondence between the spin-glass-type complexity transitions
and transitions into chaos in dynamical analogs might extend to
other hard computational problems. In this work, we examine large
networks (with a power law distribution in cluster size) that have a
large number of communities ($q \gg 1$). We infer that large systems
at a constant ratio of $q$ to the number of nodes $N$ asymptotically
tend toward insolvability in the limit of large $N$ for any positive
$T$. The asymptotic behavior of temperatures below which structure
identification might be possible, $T_\times =O[1/\log q]$, decreases
slowly, so for practical system sizes, there remains an accessible,
and generally easy, global solvable phase at low temperature. We
further employ multivariate Tutte polynomials to show that
increasing $q$ emulates increasing $T$ for a general Potts model,
leading to a similar stability region at low $T$. Given the relation
between Tutte and Jones polynomials, our results further suggest a
link between the above complexity transitions and transitions
associated with random knots.
\end{abstract}

\author{Dandan Hu}
\author{Peter Ronhovde}
\author{Zohar Nussinov}
\email{zohar@wuphys.wustl.edu}
\affiliation{Department of Physics, Washington University in St. Louis,
Campus Box 1105, 1 Brookings Drive, St. Louis, MO 63130, USA}

\maketitle{}
%\vskip 0.1in
%-----------------------------------------------------------------------------
%--- end old front matter for pre layout -------------------------------------
%-----------------------------------------------------------------------------

\section{Introduction}

Applications of physics to networks \cite{ref:newmanphystoday} has opened
fascinating doors for enhancing our understanding of these complex systems.
In particular, community detection \cite{ref:fortunatophysrep} endeavors to
identify pertinent structures within such systems.
Applications of the problem are exceptionally broad, and numerous methods have
been proposed to attack the problem
\cite{ref:rosvallmaprw,ref:blondel,ref:hastings,ref:smcd,ref:lanc,ref:rzmultires,
ref:chengshen,ref:gudkov,ref:barberLPA},
some of which have been compared for efficiency and accuracy
\cite{ref:danon,ref:noack,ref:shenchengspectral,ref:lancLFRcompare}.
%Applications include image analysis \cite{ref:huimages}

Computational ``phase transitions'' have been studied in many challenging problems \cite{ref:hoggHW,ref:mezardpz,ref:monassonzecchina,ref:mertensNPPT,ref:gentTSPPT,
ref:weigtVCPT,ref:lacasaprimes,ref:krzakala}.
Practical implications of such studies abound (e.g., Refs.\
\cite{ref:hoggHW,ref:gentTSPPT,ref:baukeMSPPT,ref:mukherjeejamming,ref:arevalojamming}),
and understanding the behavior of algorithmic solutions to these problems is
of interest because the knowledge can be leveraged to understand when a particular
solution is computationally challenging, trustworthy, or perhaps not obtainable either
via an inherent difficulty or required computational effort.
Such knowledge may be used to in certain cases to predict the hard or unsolvable
regimes of the problem \emph{a priori} (e.g., $k$-SAT \cite{ref:mezardpz}) or perhaps,
more practically in general, to dynamically adapt the solver during the onset
of a phase transition \cite{ref:ashokpatraPT}. %\cite{ref:cuestatrafficPT}).

Earlier work related to computational phase transitions with
connections to clustering include
\cite{ref:rosePTC,ref:graepelmapsCPT}, and Ref.\
\cite{ref:dorogovtsevRMP} reviewed some critical phenomena in
complex networks. The complexity of the energy landscape in
community detection was studied for a ``fixed'' Potts model (model
parameters are not set by the network under study)
\cite{ref:rzlocal,ref:huCDPTsgd}, modularity \cite{ref:goodMC}, and
belief propagation on block models \cite{ref:decelleKMZPT}. The
former and latter studies explicitly identified phase transitions in
the respective systems. We extend a previous analysis
\cite{ref:huCDPTsgd} of a Potts model where we studied the
thermodynamic and complexity character resulting in two distinct
transitions: an entropic stabilization transition where added
complexity can result in ``order by disorder'' annealing and a high
temperature disordered unsolvable phase. For extreme complexity
(high noise) at low $T$, the system is again unsolvable. Additional
transitions can appear between unsolvable and difficult solutions or
contending partitions of natural network scales. Here, we seek to
move beyond characterizing the solvable/unsolvable transition to
study the transitions in terms of changes in the energy landscape
and thermodynamic functions as functions of temperature and
``noise'' (intercommunity edges).

We utilize overlap parameters in the form of information theory measures
(see \Appref{app:information}) and a ``computational susceptibility''
$\chi$ (see \Appref{app:chi}).
Using these measures, we monitor increases in the number of local
minima corresponding to (often sharply) increased computational
complexity. We apply our Potts model to solve a random graph with
an embedded ground state, and we identify phase transitions between
``easy'' and ``hard'' \emph{solvable} phases which transition into
\emph{unsolvable} regions. Specifically, the normalized mutual
information (NMI) $I_N$, Shannon entropy $H$, the energy $E$, and
$\chi$ exhibit progressively sharper changes as the system size $N$
increases suggesting the existence of genuine thermodynamic
transitions. Similar analysis can be done for other community
detection approaches. Many community detection methods will agree on
the best solution within the easy phase, but the hard region
presents a substantially more difficult challenge.

The identified transitions may be connected to jamming
\cite{ref:mukherjeejamming,ref:arevalojamming} and avalanche (cascade)
transitions \cite{ref:leeavalanche,ref:morenocascade,ref:wangcascade}
in networks.
Dynamic jamming transitions occur in traffic, computer network,
particulate matter (e.g., sand piles), and the glassy state in amorphous
solids may be caused by similar behavior.
Refs.\ \cite{ref:zhengscascade,ref:wucascade,ref:ikedacascade} showed
relations between clustering and cascades in certain networks,
and Ref.\ \cite{ref:tahbazsalehiauto} relates agent dynamics to the
Kuramoto oscillators model which has been used for community detection
\cite{ref:arenassync}.
The threshold emergence of Giant Connected Components (GCC) is
related to epidemic thresholds \cite{ref:serranoperc,ref:mooreperc},
and by nature of the emerging global connectivity, the GCC is directly
detectable via clustering at large-scale resolutions
[\ie{}, small $\gamma$ in \eqnref{eq:ourpotts}].
Jones polynomials in knot theory are related to Tutte polynomials for the Potts
model, so our results suggest similar transitions in random knots
(see \Appref{app:trefoilknot}).

We will analytically investigate partition functions and free
energies of a several graphs in the high temperature $T$ and large
number of communities $q$ approximations. We illustrate that
increasing $T$ emulates increasing $q$ for a general system, and the
analytical results are consistent with the computational phase
diagrams.

The remarks of the paper is organized as follows: We introduce the
community detection model in \secref{sec:potts} and then the
embedded graph/noise test in \secref{sec:graphdef}.
\Secref{sec:transition} demonstrates the spin-glass-type transitions
that occur in our community detection problem via numerical
simulation using several instability measures. In
\secref{sec:NIcliques}, we derive crossover thresholds for a simple
case and discuss their connections to the numerical simulations, and
\secref{sec:birdexample} demonstrates the effect of the different
solution regions with a specific example. \Secref{sec:f-analytic}
carries out analytic free energy calculations on arbitrary
unweighted graphs using a ferromagnetic Potts model.
\Appref{app:replica} exams the notation of ``trials'' and
``replicas'' which are of paramount importance in our work to
directly probe the phase diagram sans the use of mean-field type or
other approximation concerning complexity. \Appref{app:replica}
defines some terminology used in the paper. \Appref{app:information}
and \Appref{app:chi} describe our information and stability
measures, and \Appref{app:HBA} elaborates on our heat bath community
detection algorithm. We introduce the Tutte polynomial method for
calculating the partition function of a Potts model for unweighted
and weighted graphs in \Appref{app:tuttepoly}, and we show an exact
calculation for a simple connected graph in
\Appref{app:cliquecircle}. Finally, \Appref{app:trefoilknot}
conjectures the existence of a similar transition for knots.

\section{Potts Hamiltonian}
\label{sec:potts}

% --- communities example -------------------------------------------------
\begin{figure}
\centering
\includegraphics[width=0.7\columnwidth]{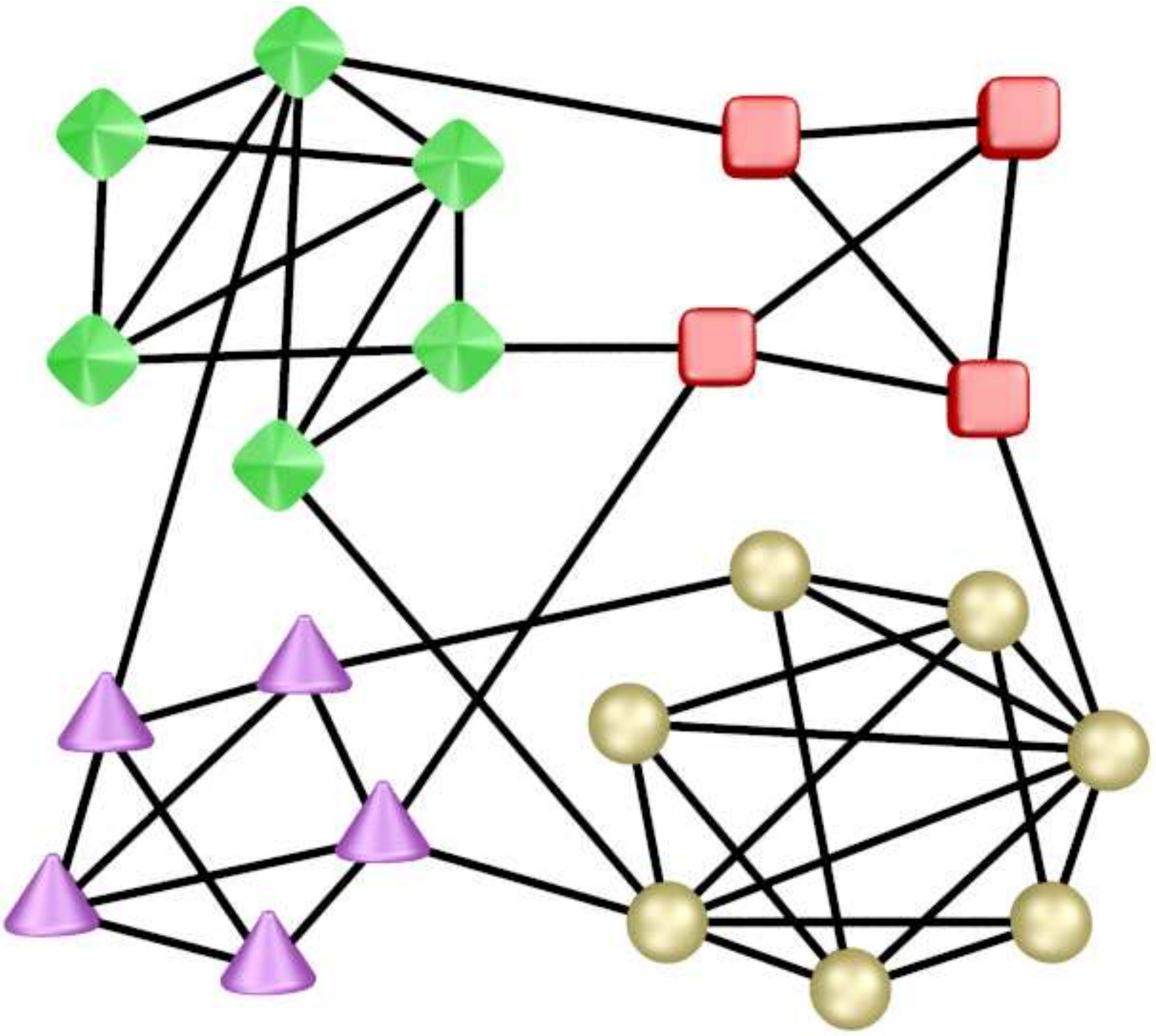}
\caption{(Color online) The figure illustrates a partition where nodes are
separated into distinct communities as indicated by distinct shapes and colors,
thus identifying relevant structure in the graph.
The current work elaborates on computational transitions and disorder
in terms noise (extraneous intercommunity edges) or thermal effects
(high temperature $T$ or large system size $N$) of solving such systems
using a stochastic heat bath solver (see \Appref{app:HBA}).}
\label{fig:communities}
\end{figure}
% --- end communities example ---------------------------------------------

We employ a spin-glass-type Potts model Hamiltonian for solving the
community detection problem
\begin{equation}
  H(\{\sigma\}) = -\frac{1}{2}\sum_{i\neq j}\left[ A_{ij}-\gamma
                   \left(1-A_{ij}\right)\right] \delta(\sigma_i,\sigma_j)
  \label{eq:ourpotts}
\end{equation}
which we refer to as an ``Absolute Potts Model'' (APM).
Given $N$ nodes, $A_{ij}$ denotes the adjacency matrix where $A_{ij}=1$
if nodes $i$ and $j$ are connected and is $0$ otherwise.
In general, $A_{ij}$ may be trivially extended to a weighted adjacency
matrix $w_{ij}$ (perhaps including ``adversarial'' relations) \cite{ref:rzlocal},
but we utilize unweighted graphs in most of the current work
(see \secref{sec:birdexample}).
Each spin $\sigma_{i}$ may assume integer values in the range
$1\leq\sigma_{i}\leq q$ where $q$ is the (dynamic) number of communities
where node $i$ is a member of community $k$ when $\sigma_{i}=k$.
In the current work, we set the resolution parameter \cite{ref:rzmultires}
to $\gamma=1$ which is near an optimal value for communities with high internal
edge densities (see \secref{sec:graphdef}).

Previous work \cite{ref:rzlocal,ref:rzmultires} elaborated on a
``zero-temperature'' ($T=0$) community detection algorithm which we used
to minimize \eqnref{eq:ourpotts}.
A depiction of community structure is shown in \figref{fig:communities}
where different communities are represented by different node shapes
and colors.
Here, we investigate the Hamiltonian at non-zero temperatures ($T>0$)
by applying a heat bath algorithm (HBA, see \Appref{app:HBA}).
%Briefly, we iteratively select each node and test for possible moves
%where probabilities are calculated via a Boltzmann weight
%$e^{-\beta\Delta E}$ ($\beta=1/T$) at a temperature $T$ using the energy
%change ($\Delta E$) as if the node were moved into a connected (or new) cluster.

We further invoke $s$ independent solutions (``trials'', see
Appendix A) by solving copies of the system which differ by a
permutation of the order of the spin indices. This process leads to
states that (perhaps locally) minimize \eqnref{eq:ourpotts}, so we
select the lowest energy trial as the best solution. We vary $s$ in
the range $4 \le s \le 20$ where we employ $s=4$ trials in general
and use $s>4$ trials for calculating the computational
susceptibility in \eqnref{eq:sus}.

In our multi-scale (``multiresolution'') analysis, we solve $r=100$
independent ``replicas'' (see \Appref{app:replica}) and examine
information theory correlations between the replicas and the planted
ground state solutions. We schematically show such a set of
independent solvers in \figref{fig:MRAlandscape} where stronger
agreement among the replicas indicates a more robust solution. We
compute the average inter-replica information correlations among the
ensemble of replicas allowing us to infer a more detailed picture of
the system beyond that of a single optimized solution. Specifically,
\emph{information theory extrema} as a function of $T$ and $\gamma$
(or other scale parameters in general) correspond to most relevant
scale(s) of the system.

% --- MRA landscape figure ------------------------------------------------
\begin{figure}
\centering
\includegraphics[width=0.9\columnwidth]{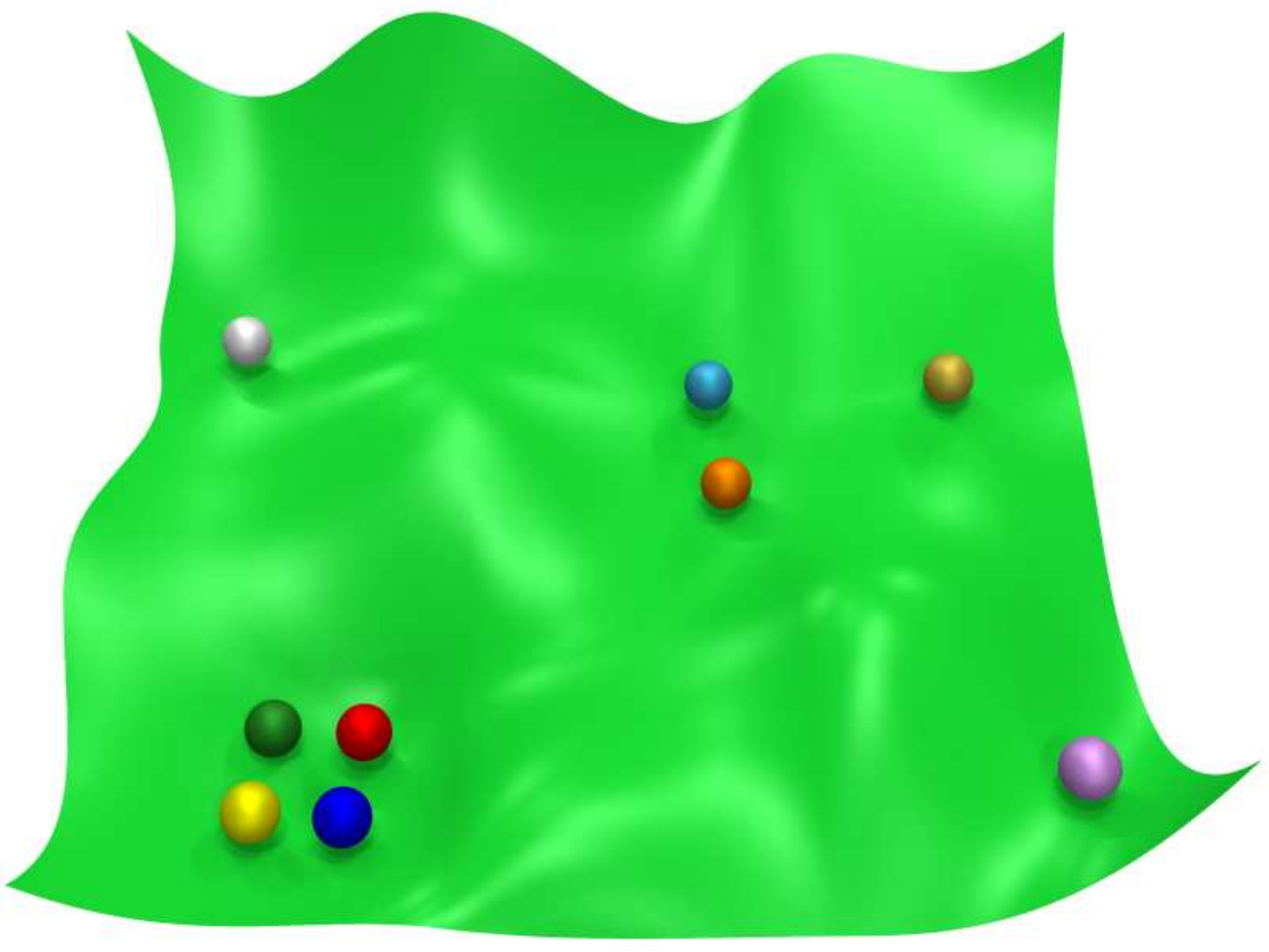}
\caption{(Color online) The schematic illustrates different minimizers
attempting to solve a system.
Colored spheres represent distinct minimizers (``replicas'') that seek
a (perhaps local) minimum of a cost function.
In an easy system, multiple solution attempts will generally reach
a good solution (such as the bottom left region of the landscape),
but hard systems require more effort to solve accurately
(that is, to achieve strong agreement between the replicas).
Unsolvable regions restrict accurate solutions without extreme levels
optimization (such as exhaustive search).}
\label{fig:MRAlandscape}
\end{figure}
% --- end MRA landscape figure --------------------------------------------

% --- phase transition schematics -----------------------------------------
\begin{figure}
\centering
\includegraphics[width=0.9\columnwidth]{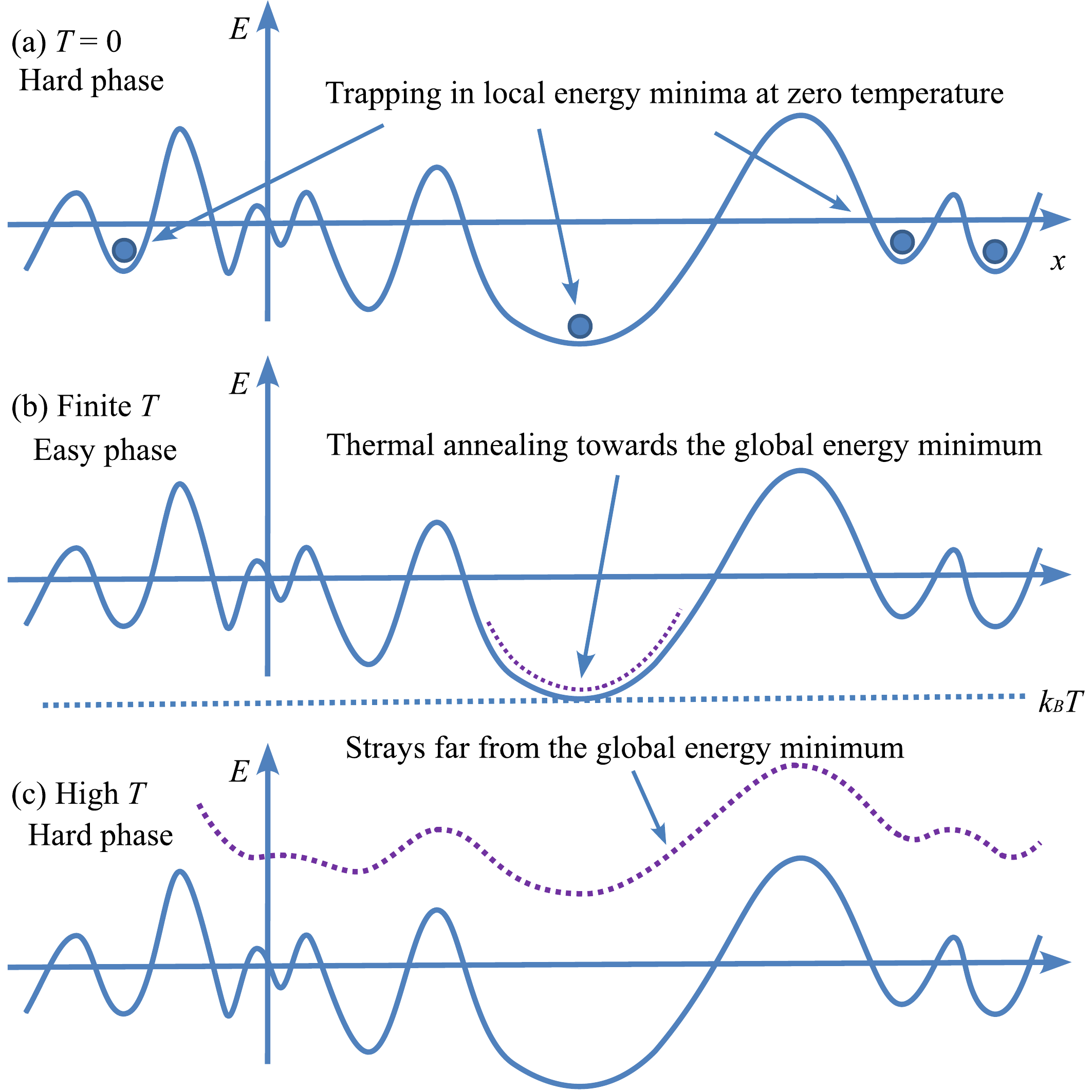}
\caption{(Color online) Panel (a) schematically illustrates in one dimension
the easy and hard phases induced by the level of noise (extraneous
intercommunity edges) encountered by a solver.
Greedy algorithms are easily trapped in local energy minima above a certain
noise threshold.
We previously showed that the model of \eqnref{eq:ourpotts} is robust to noise
\cite{ref:rzlocal} even with a greedy algorithm.
Stochastic solvers such as a heat bath algorithm (see \Appref{app:HBA}) or simulated
annealing enable one to circumvent the effects of some noise, but excessive levels
will still thwart these solvers because meaningful partition information is obscured
by the complexity of the energy landscape.
Panels (b) and (c) schematically depict the easy and hard phases in terms
of the temperature for the stochastic heat bath solver (see \Appref{app:HBA}).
Above a graph-dependent threshold, the solver is less sensitive to local
energy landscape features.}
\label{fig:phaseanalogs}
\end{figure}
% --- end phase transition schematics -------------------------------------

\section{Construction of embedded graphs and the noise test}
\label{sec:graphdef}

Similar to \cite{ref:lancbenchmark}, we construct a ``noise test''
benchmark as a medium in which to study phase transitions in random
graphs with embedded solutions \cite{ref:rzlocal,ref:huCDPTsgd}. We
define the system ``noise'' as intercommunity edges that connect a
given node to communities other than its original or ``best''
community assignment. In general \cite{ref:rzlocal}, it is not
possible at the beginning of an attempted solution to ascertain
which edges contribute to noise and which constitute edges within
communities of the best partition(s).

For each benchmark graph, we divide $N$ nodes into $q$ communities with
a power law distribution of community sizes $\{n_i\}$ given by $n^\beta$
where $\beta=-1$.
We then connect ``intracommunity'' edges at a high average edge density
$p_{in}=0.95$.
Initially, the external edge density is zero, $p_{out}=0$, so that we have
perfectly decoupled clusters.
To this system, we add random intercommunity edges at a density
of $p_{out}<0.5$.
We define $p_{in}$ ($p_{out}$) as the ratio of the number of intracommunity
(intercommunity) edges over the maximum possible intracommunity (intercommunity)
edges.

We define the \emph{average external degree} of each node $Z_{out}$ as the
average number of links that a given node has with nodes in communities other
than its own.
Similarly, the \emph{average internal degree} $Z_{in}$ is defined as the average
number of links to nodes in the \emph{same} community, and $Z_{in} + Z_{out} = Z$
where $Z$ is the average coordination number.
Then we can explicitly write the internal and external edge densities
\begin{equation}
  p_{in}=\frac{NZ_{in}}{\sum_{a=1}^q n_a(n_a-1)},
  \label{eq:pin}
\end{equation}
and
\begin{equation}
  p_{out}=\frac{NZ_{out}}{\sum_{a=1}^q\sum_{b \neq a}^q n_a n_b}.
  \label{eq:pout}
\end{equation}
where $n_a$ denotes the size of community $a$.

The communities in this construction are well defined, on average,
at reasonable levels of noise ($p_{out}\lesssim 0.3$ depending on
the typical community size $n$). As external links are progressively
added to the system ($p_{out}$ increases), the communities become
increasingly difficult to detect. At some stage, enough noise is
added and $p_{out}$ is sufficiently high that the planted partition
cannot be detected despite the fact that the optimal ground state is
still well-defined. This transition often occurs sharply,
particularly for large networks. We investigate the phase transition
from the solvable to unsolvable phases at both low and high
temperatures by means of the heat bath algorithm described in
\Appref{app:HBA}.

\begin{figure}[t!]
\centering
\includegraphics[width=0.9\columnwidth]{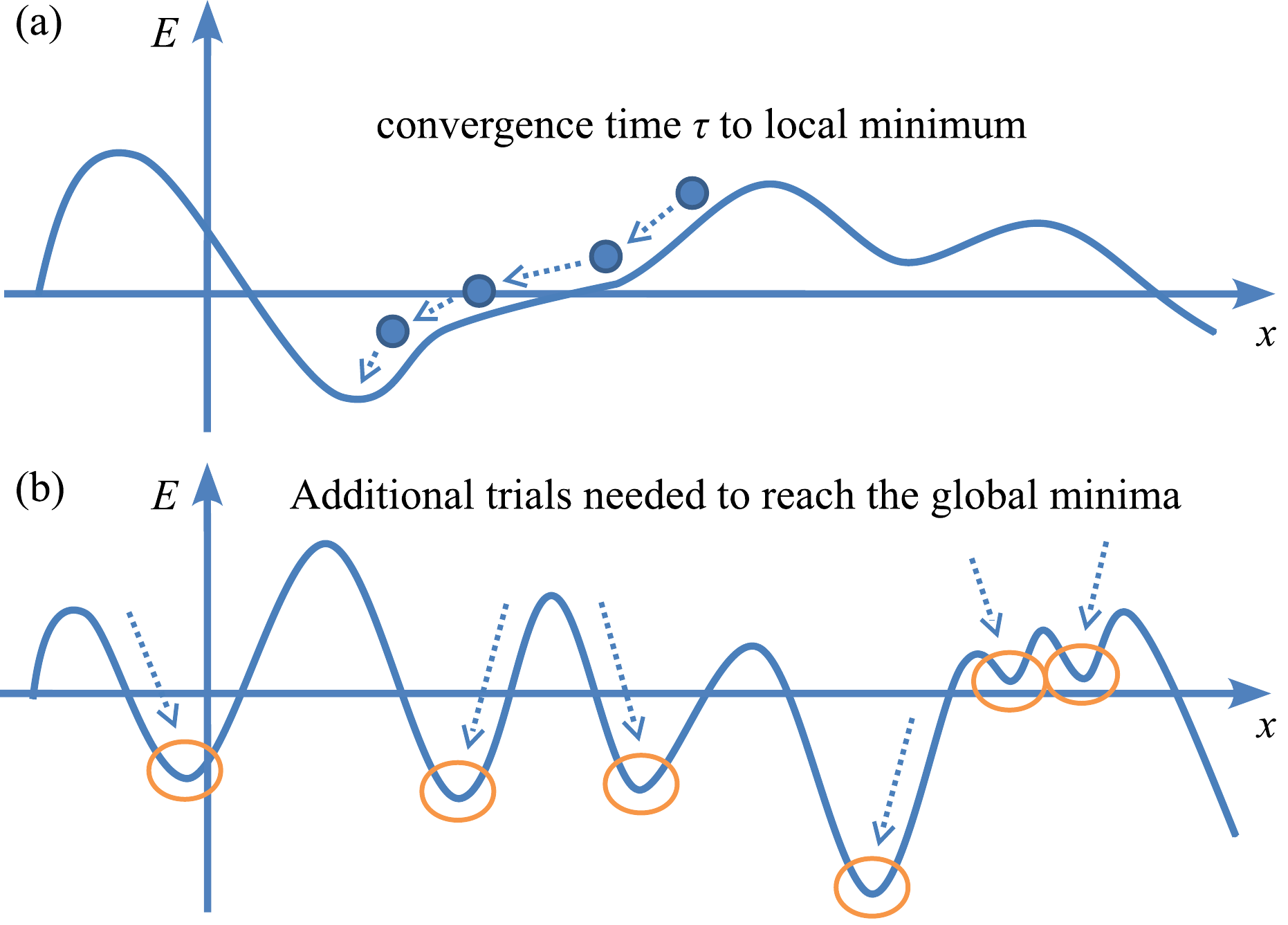}
\caption{(Color online) The figure schematically illustrates the convergence
time of a solver in panel (a) and the effect of additional optimization trials
in panel (b).
Additional optimization trials are utilized in a ``computational susceptibility''
$\chi$ in order to numerically estimate the complexity of the energy landscape
(see \Appref{app:chi}).}
\label{fig:chianalog}
\end{figure}

% --- begin chi fixed alpha plots ----------------------------------
\begin{figure*}[t]
\begin{center}
\subfigure[\ $N=256$, $q=4$, $\alpha
=0.016$]{\includegraphics[width=\subfigwidth]{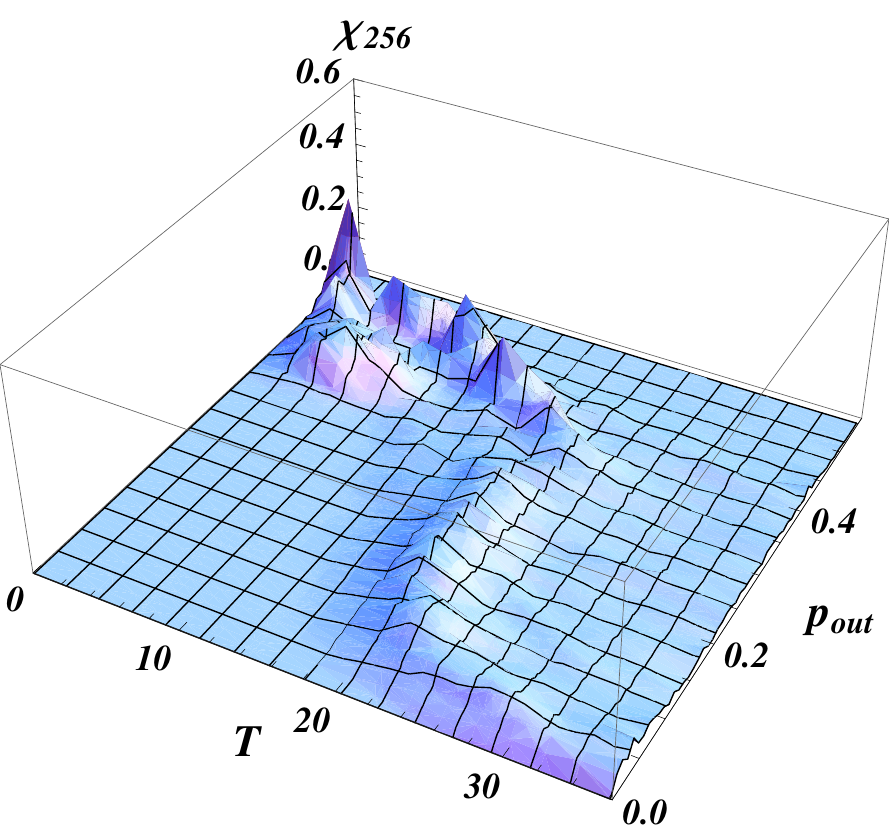}}
\subfigure[\ $N=512$, $q=8$, $\alpha
=0.016$]{\includegraphics[width=\subfigwidth]{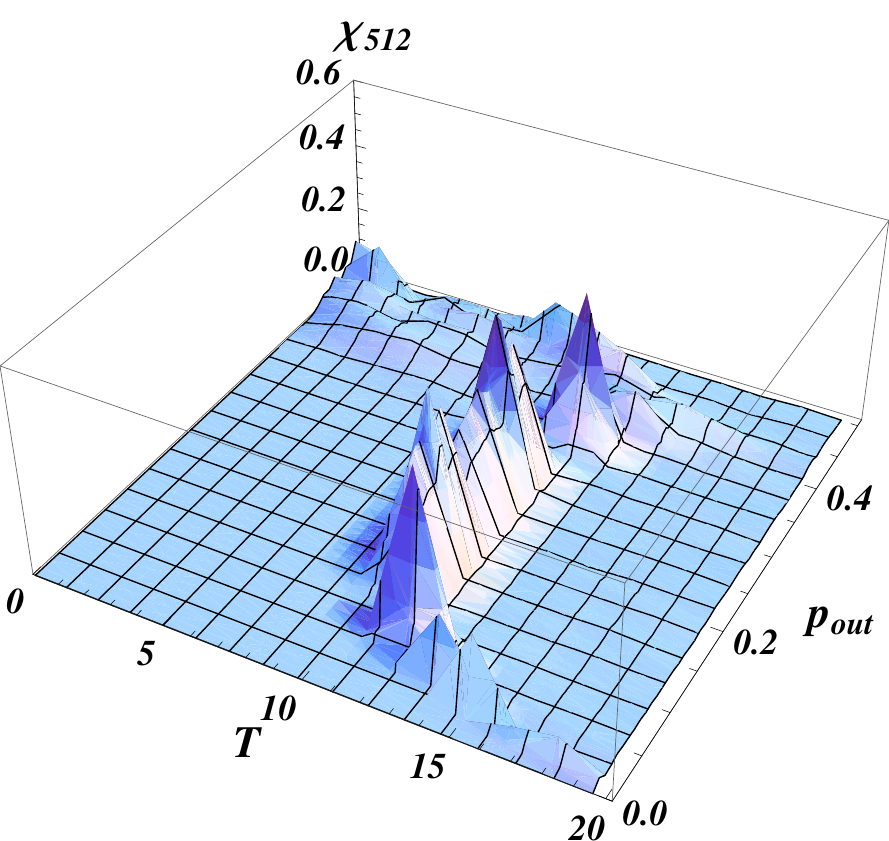}}
\subfigure[\ $N=1024$, $q=16$, $\alpha
=0.016$]{\includegraphics[width=\subfigwidth]{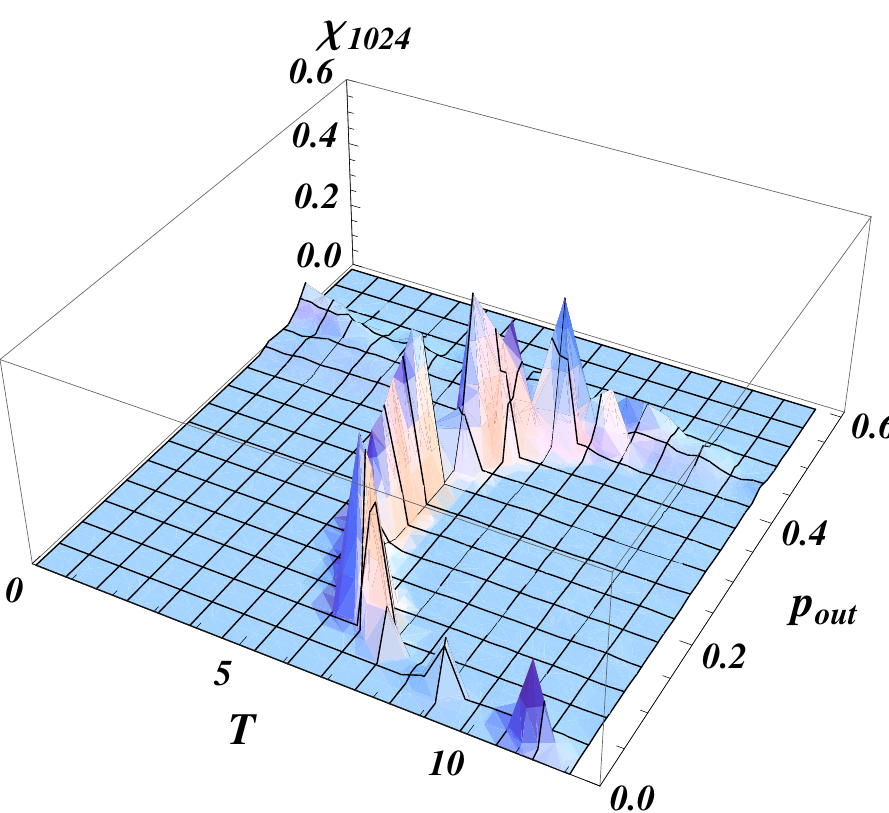}}
\subfigure[\ $N=2048$, $q=32$, $\alpha
=0.016$]{\includegraphics[width=\subfigwidth]{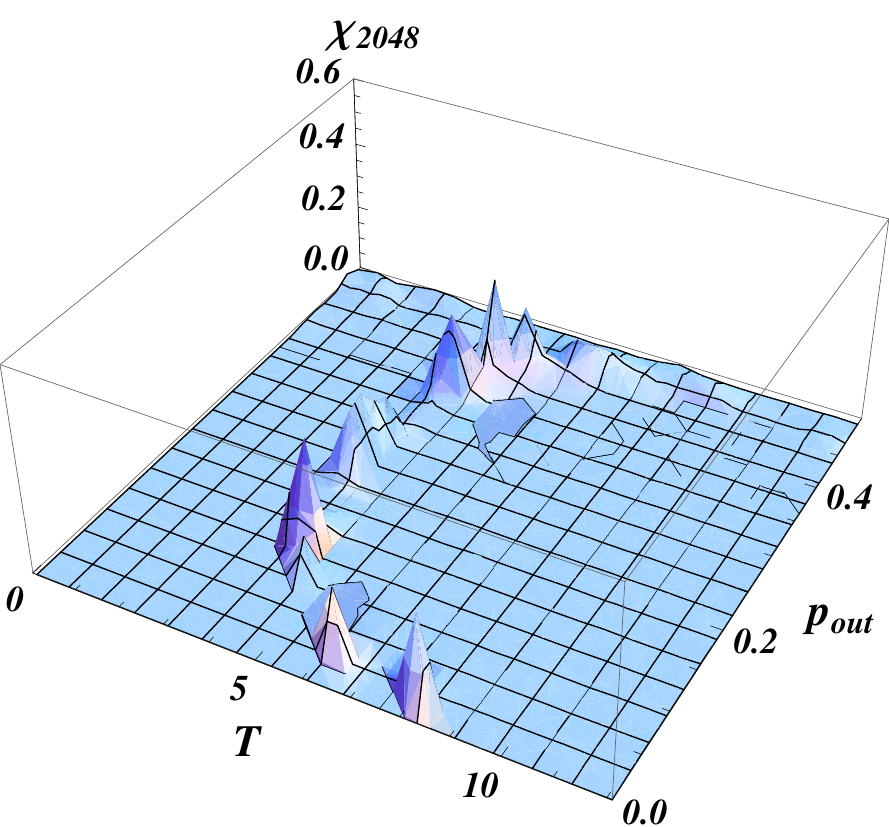}}
\subfigure[\ $N=256$, $q=18$, $\alpha
=0.07$]{\includegraphics[width=\subfigwidth]{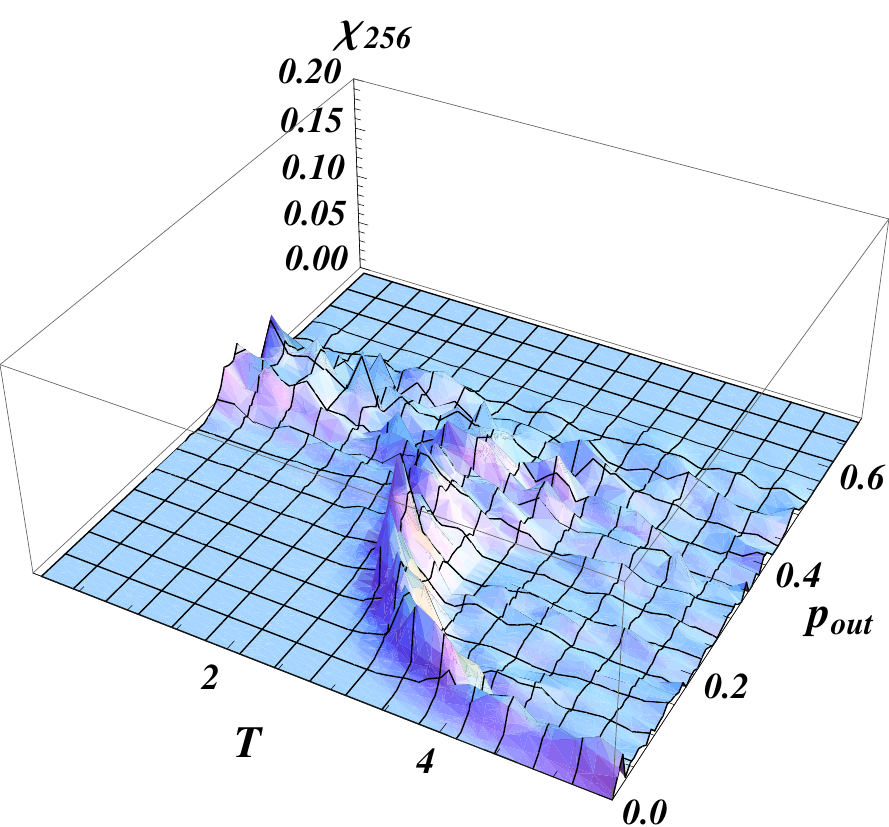}}
\subfigure[\ $N=512$, $q=35$, $\alpha
=0.07$]{\includegraphics[width=\subfigwidth]{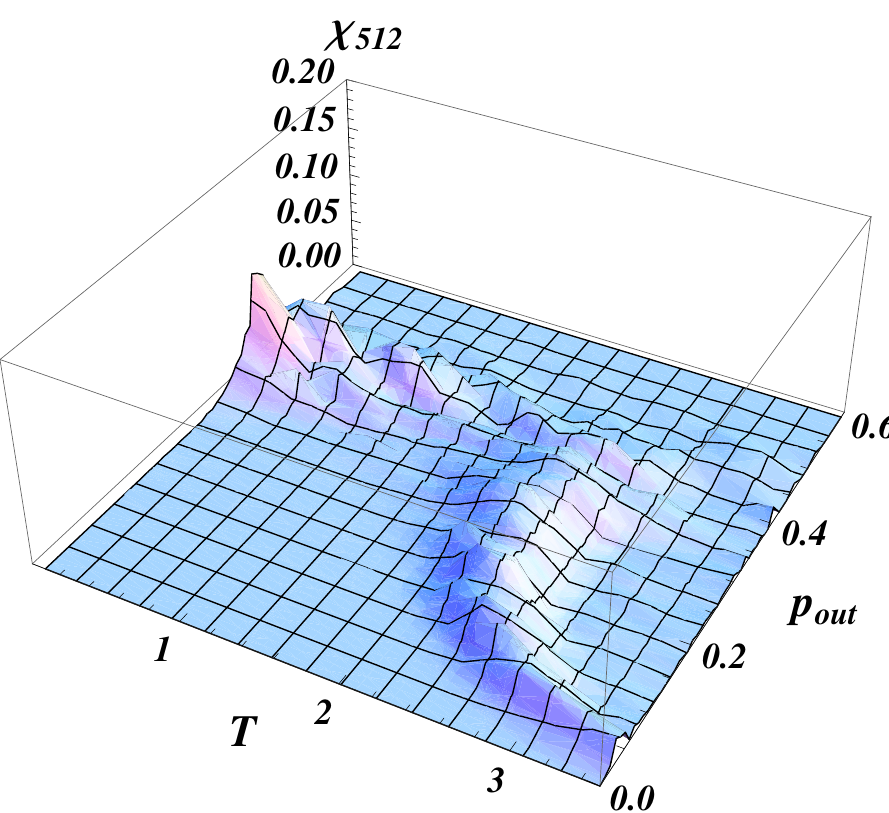}}
\subfigure[\ $N=1024$, $q=70$, $\alpha
=0.07$]{\includegraphics[width=\subfigwidth]{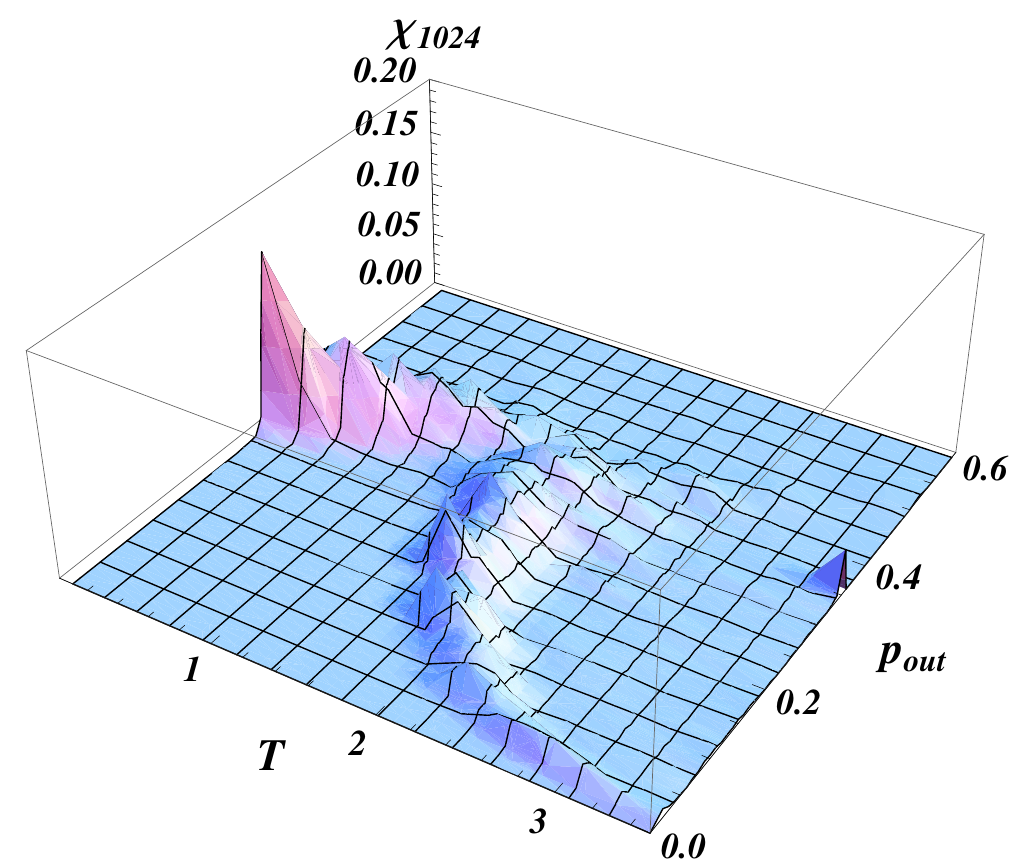}}
\subfigure[\ $N=2048$, $q=140$, $\alpha
=0.07$]{\includegraphics[width=\subfigwidth]{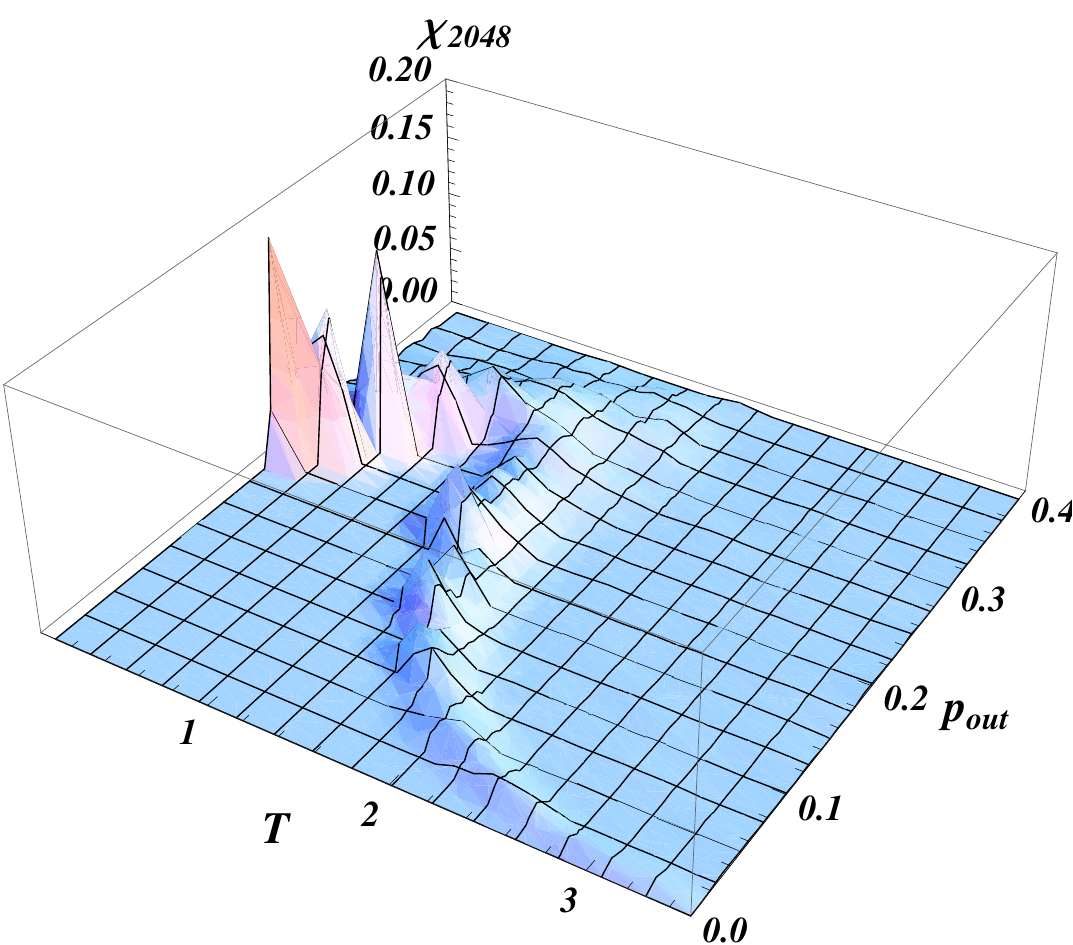}}
\subfigure[\ $N=128$, $q=20$, $\alpha
=0.15$]{\includegraphics[width=\subfigwidth]{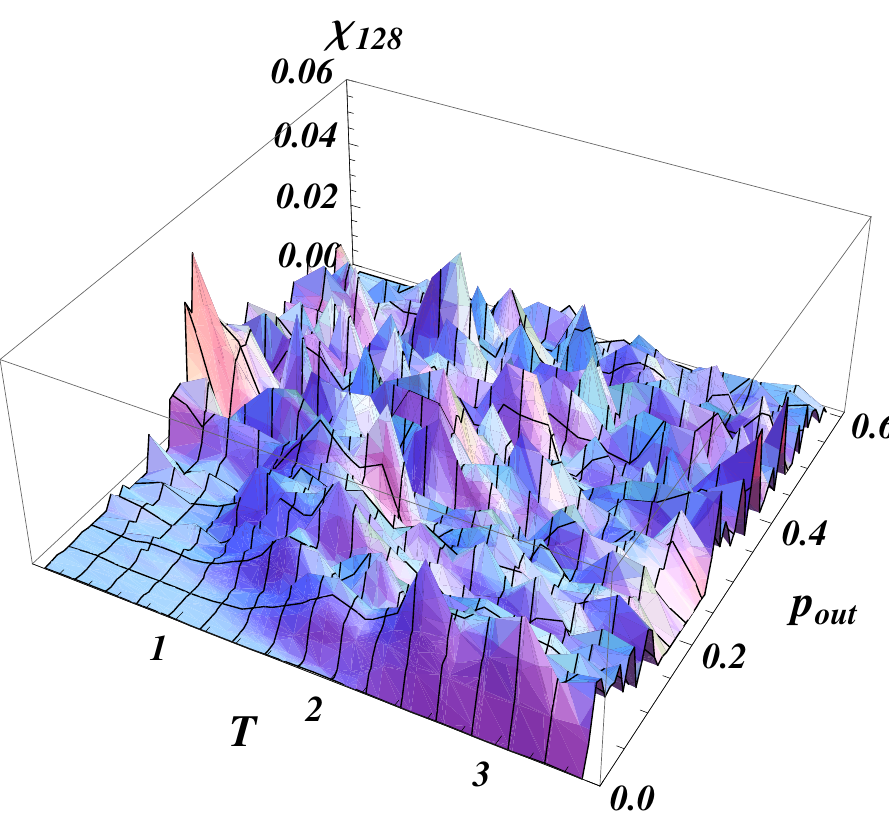}}
\subfigure[\ $N=256$, $q=40$, $\alpha
=0.15$]{\includegraphics[width=\subfigwidth]{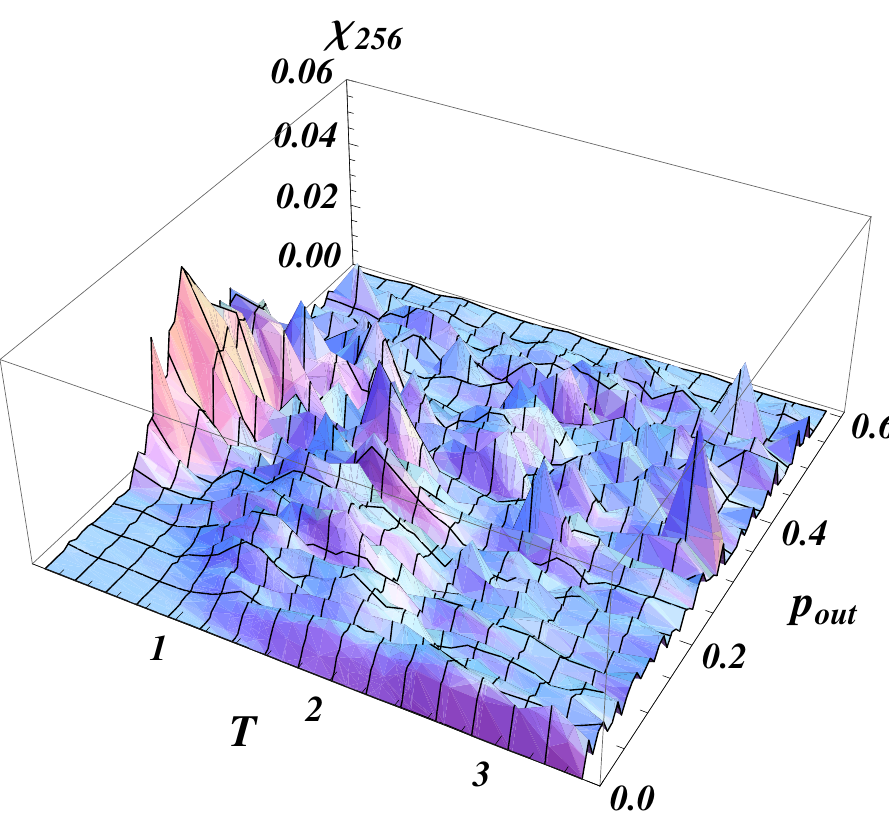}}
\subfigure[\ $N=512$, $q=80$, $\alpha
=0.15$]{\includegraphics[width=\subfigwidth]{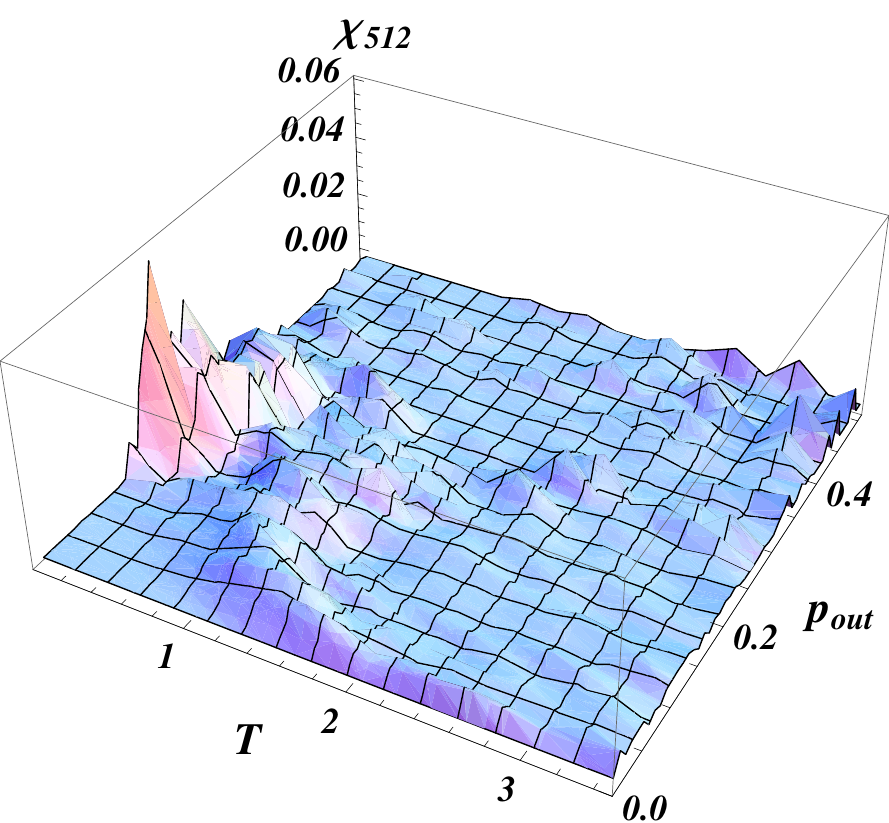}}
\subfigure[\ $N=1024$, $q=160$, $\alpha
=0.15$]{\includegraphics[width=\subfigwidth]{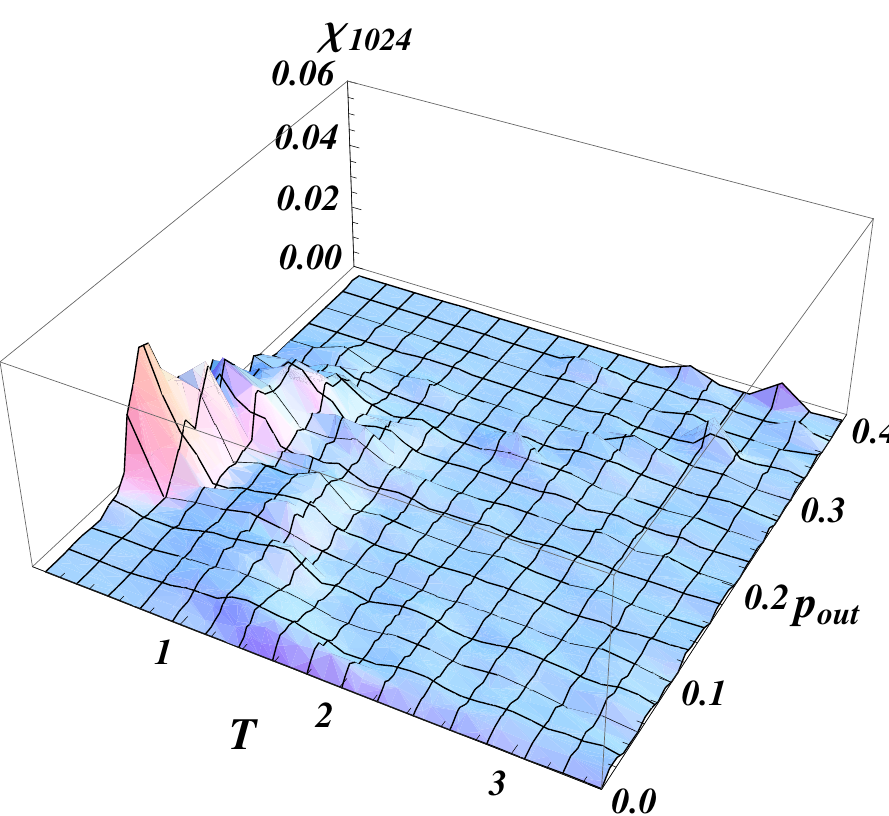}}
\end{center}
\caption{(Color online) Each panel shows a 3D plot $\chi(T,p_{out})$ as a function
of temperature $T$ and noise level $p_{out}$ for systems with the indicated number
of nodes $N$, communities $q$, and $\alpha=q/N$ ratio.
In panels (a--h) for $\alpha=0.016$ and $0.07$, all plots show three clear phases,
and the ``ridges'' at low and high temperatures mark the hard phase.
The hard phase separates the easy phase (the flat region in the lower left corner
with low temperature and low noise) from the unsolvable phase (the flat region
in the upper right corner with high temperature and high noise).
In panels (a)--(h), the ridges in $\chi(T,p_{out})$ become narrower
%(and higher in general)
as $N$ increases.
The area of the easy (hard) regions decreases (increases) from panel (a) to (d)
and (e) to (h), respectively.
In panels (a--d) for $\alpha=0.016$, the hard phase at low temperature becomes
less prominent from panel (a) to (d), but it becomes more prominent at high
temperature.
In panels (e--h) for $\alpha=0.07$, the hard phase at low temperature becomes
more prominent from panel (e) to (h), but it remains constant at high temperature.
In panels (i--l) for $\alpha=0.15$, only the larger systems with $N\ge 512$ show
clear phases.
The smaller systems with $N=128$ in panel (i) and $N=256$ in panel (j) show
very noisy phases where only the easy phase can be readily determined, and
the boundaries for the hard and unsolvable phases are difficult to pinpoint.}
\label{fig:susAllalpha}
\end{figure*}
% --- end chi fixed alpha plots ------------------------------------

\section{Spin glass type transitions}
\label{sec:transition}

We previously reported \cite{ref:huCDPTsgd} on the existence of \emph{two}
spin-glass-type transitions in the constructed graphs mentioned in \secref{sec:graphdef}.
Evidence for the transitions are observed in several measures such the accuracy
of the solution obtained by means of the APM in \eqnref{eq:ourpotts}
(and other models \cite{ref:smcd,ref:rzlocal} in general), the computational
effort required to converge to a solution \cite{ref:rzmultires,ref:rzlocal},
entropy effects, and others.
%In particular, the variation of information $V$ calculated between the test
%system partition and the solution shows a sharp phase transition as the noise
%$p_{out}$ is increased \cite{ref:rzlocal}.
Compared to another Potts-type qualtity function \cite{ref:smcd}
utilizing a ``null model'' (a random graph used to evaluate the
quality of a candidate partition), the APM exhibits a somewhat
sharper transition as $N$ is increased \cite{ref:rzlocal}. As
alluded to above, two transitions are generally encountered as the
noise value (or temperature) is increased. At fixed temperature $T$,
as $p_{out}$ is steadily increased from zero, the first onset of
spin glass behavior first appears for values $p_{1} \le p_{out} \le
p_{2}$.

\subFigref{fig:phaseanalogs}{a} illustrates a one dimension
characterization of the easy and hard phases in terms of the level
of noise (extraneous intercommunity edges) encountered by a greedy
solver. It is in this context that greedy algorithms are, in
general, more easily trapped in local energy minima above a certain
noise threshold. Stochastic solvers such a heat bath algorithm
discussed in \Appref{app:HBA} or simulated annealing (SA) enable one
to circumvent noise to some extent, but excessive levels will even
thwart these more robust solvers because meaningful information is
eventually obscured by the complexity of the energy landscape.
\subfigref{fig:phaseanalogs}{b,c} depict the easy and hard phases at
low and high temperatures $T$, respectively, for our HBA (see
\Appref{app:HBA}). Above a graph-dependent threshold, the solver is
insensitive to local features, and it is unable to find an accurate
solution.

We showed that \eqnref{eq:ourpotts} is robust to noise \cite{ref:rzlocal}
leading to exceptional accuracy even with a greedy algorithm.
Some other methods and cost-functions \cite{ref:blondel,ref:traaglocalscope}
have also proven to be very accurate \cite{ref:lancLFRcompare} with a
greedy-oriented algorithm.
While maximizing modularity \cite{ref:gn} and a closely related cost function
in \cite{ref:smcd} have proven to be accurate and productive,
Refs.\ \cite{ref:fortunato,ref:goodMC,ref:lancfortunatomod} have discussed
problems associated with maximizing modularity in community detection.
We briefly illustrated \cite{ref:rzlocal} a correspondence between the major
transition experienced by \eqnref{eq:ourpotts} and a Potts model in \cite{ref:smcd}.
We conjecture the existence of a related transition for random knots
in \Appref{app:trefoilknot}.

In \secsref{sec:sec1}{sec:sec2}, we elaborate on the transitions using a
computational susceptibility $\chi$ as defined in \Appref{app:chi}.
In analogy with other physical susceptibility parameters, $\chi$ measures
the response of the system to additional optimization effort.
We schematically illustrate the effect in \figref{fig:chianalog}.
A higher $\chi$ indicates a more disordered, but navigable, energy landscape
where a low $\chi$ indicates that additional optimization has less effect
whether due to extreme disorder or a trivially solvable system.
Finally in \secref{sec:sec3}, we illustrate the transitions using additional
stability measures.

\begin{figure*}[t]
\begin{center}
\subfigure[\ ``hard'' phase boundary for
$\alpha=0.016$]{\includegraphics[width=\subfigwidthnew]{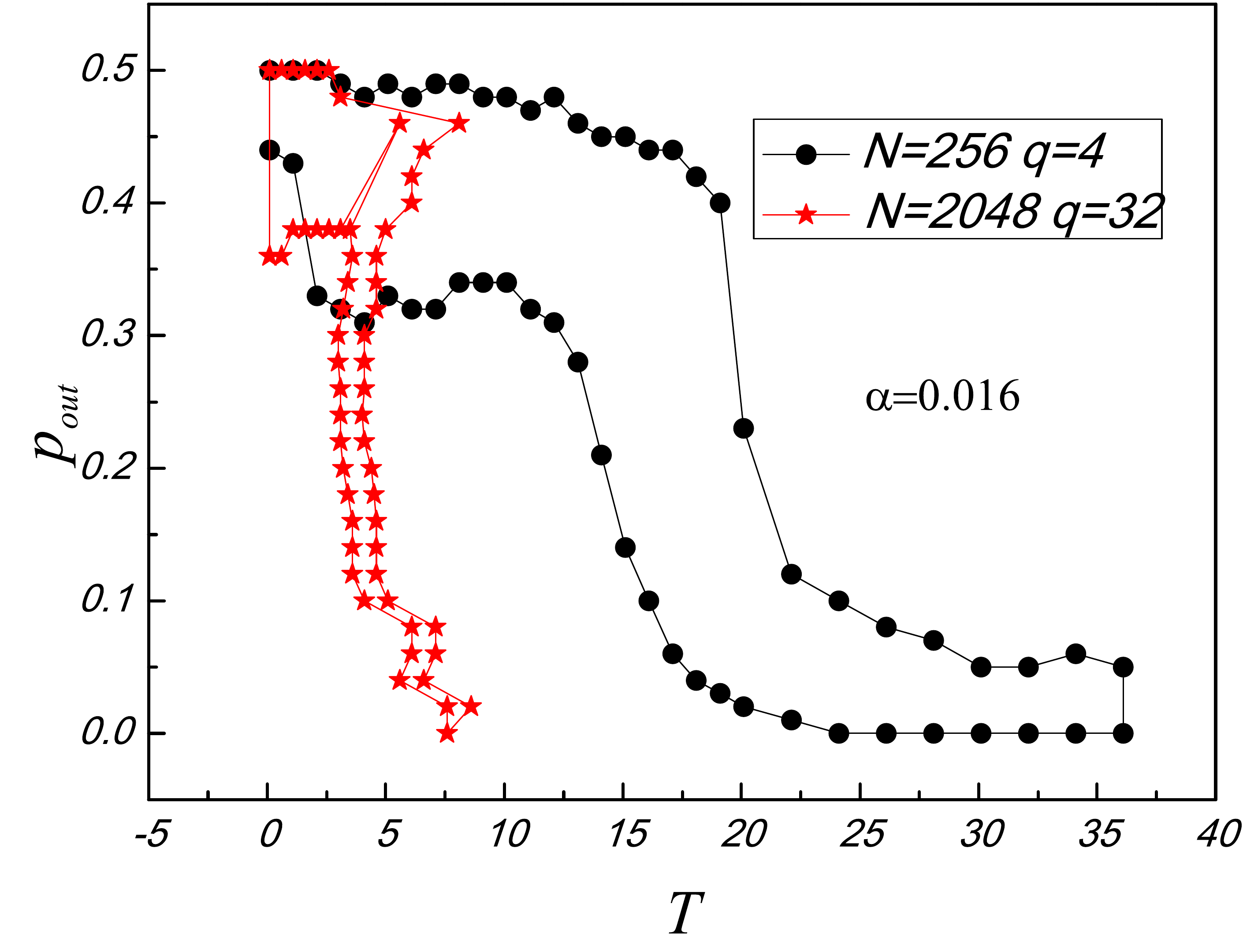}}
\subfigure[\ ``hard'' phase boundary for
$\alpha=0.07$]{\includegraphics[width=\subfigwidthnew]{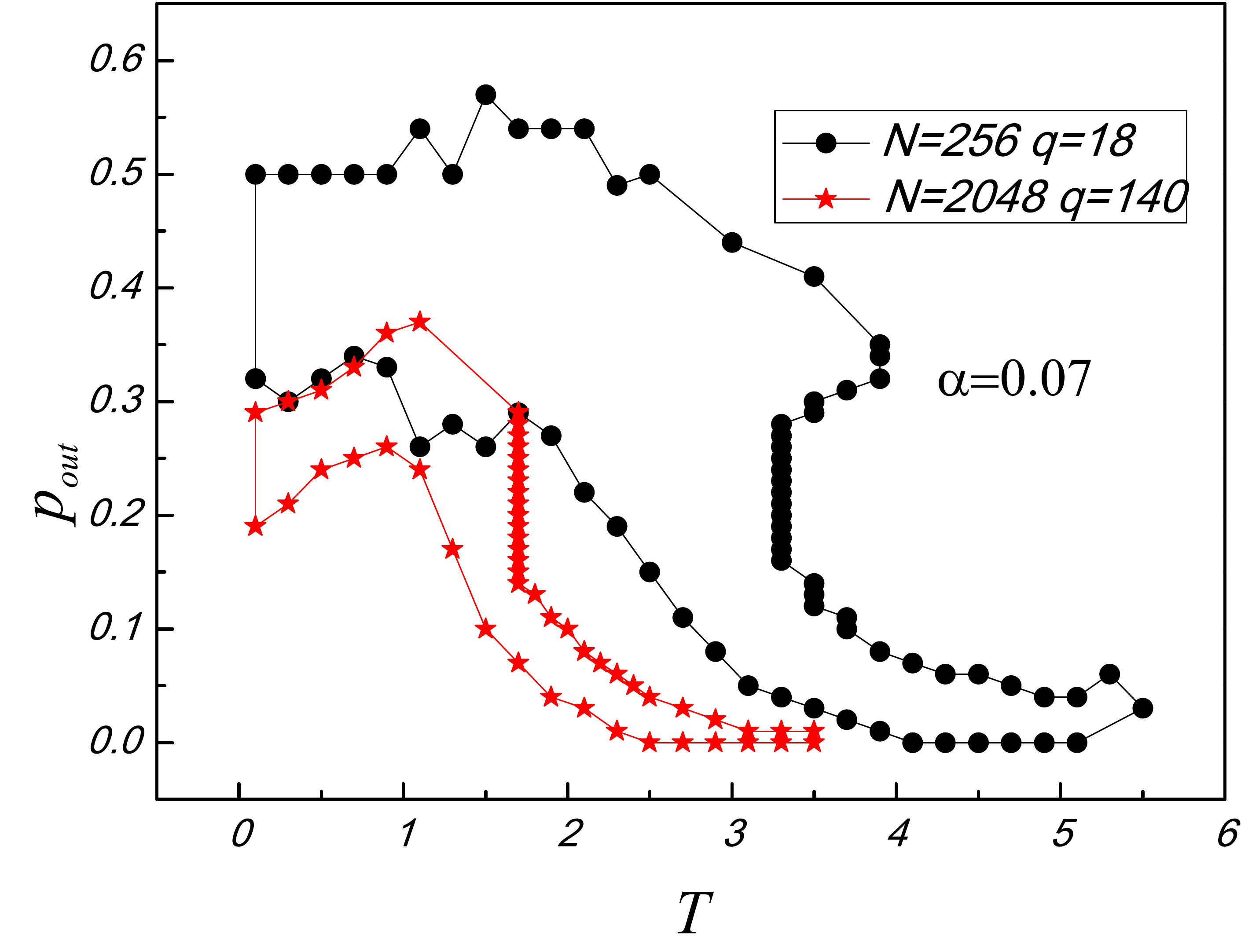}}
\subfigure[\ ``hard'' phase boundary for
$\alpha=0.15$]{\includegraphics[width=\subfigwidthnew]{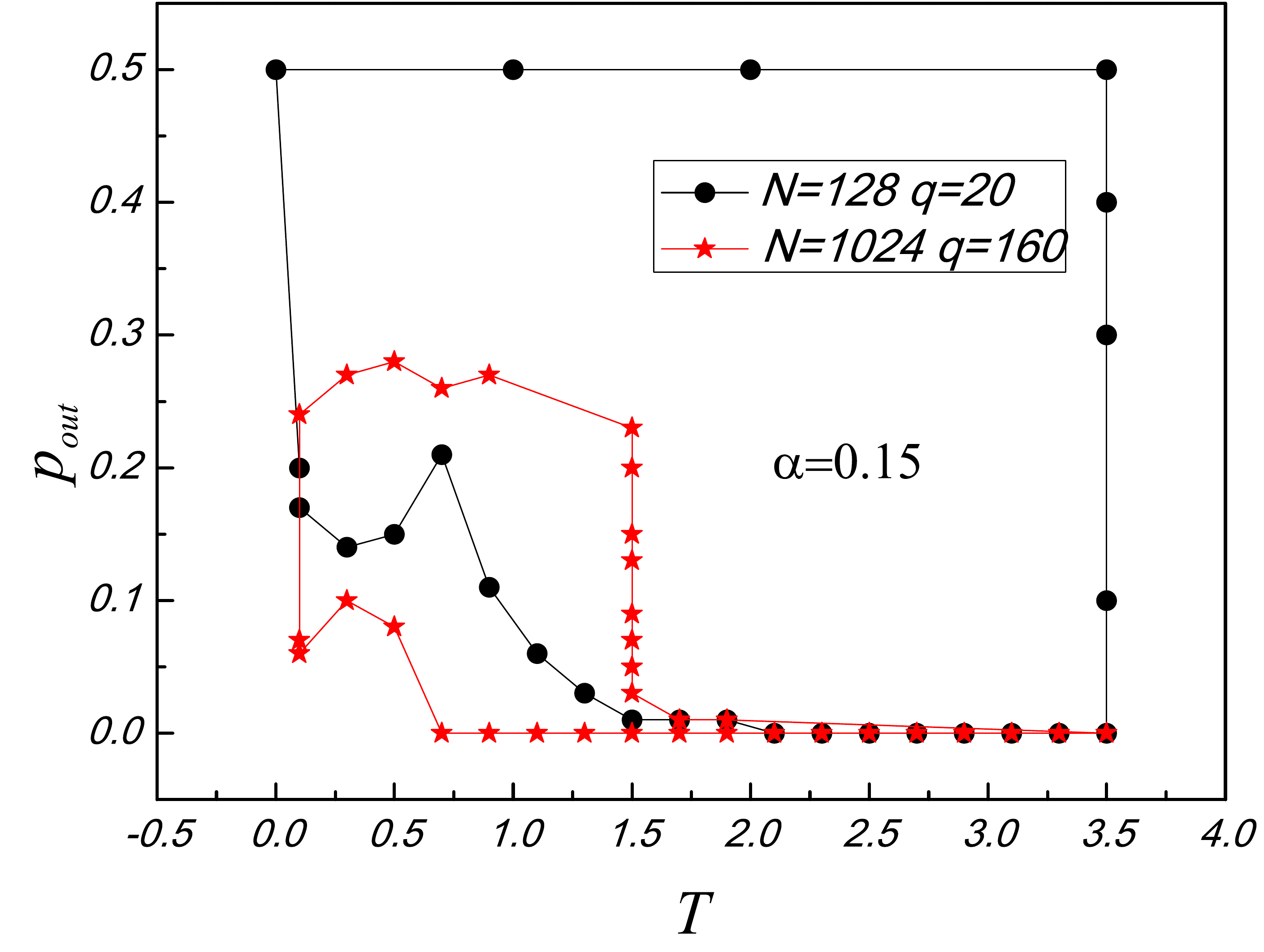}}
\end{center}
\caption{(Color online) Corresponding to \figref{fig:susAllalpha} and
\secref{sec:sec1}, each plot depicts the boundaries of the hard phase
for the system series with a fixed $\alpha=q/N$ ratio.
Panels (a), (b), and (c) show the results for $\alpha=0.016$, $\alpha=0.07$,
and $\alpha=0.15$, respectively.
System sizes range from $N=256$ to $2048$, and $q$ varies from $4$ to $160$
as indicated in each plot.
For each $\alpha$, the area within hard phase boundary becomes progressively
narrower indicating that the transitions from the easy to unsolvable
phases are more clear in the thermodynamic limit.}
\label{fig:sus2dalpha}
\end{figure*}

\begin{figure*}[t]
\begin{center}
\subfigure[\ $p_1(T)$ with $\alpha=0.016$]{\includegraphics[width=\subfigwidthnew]{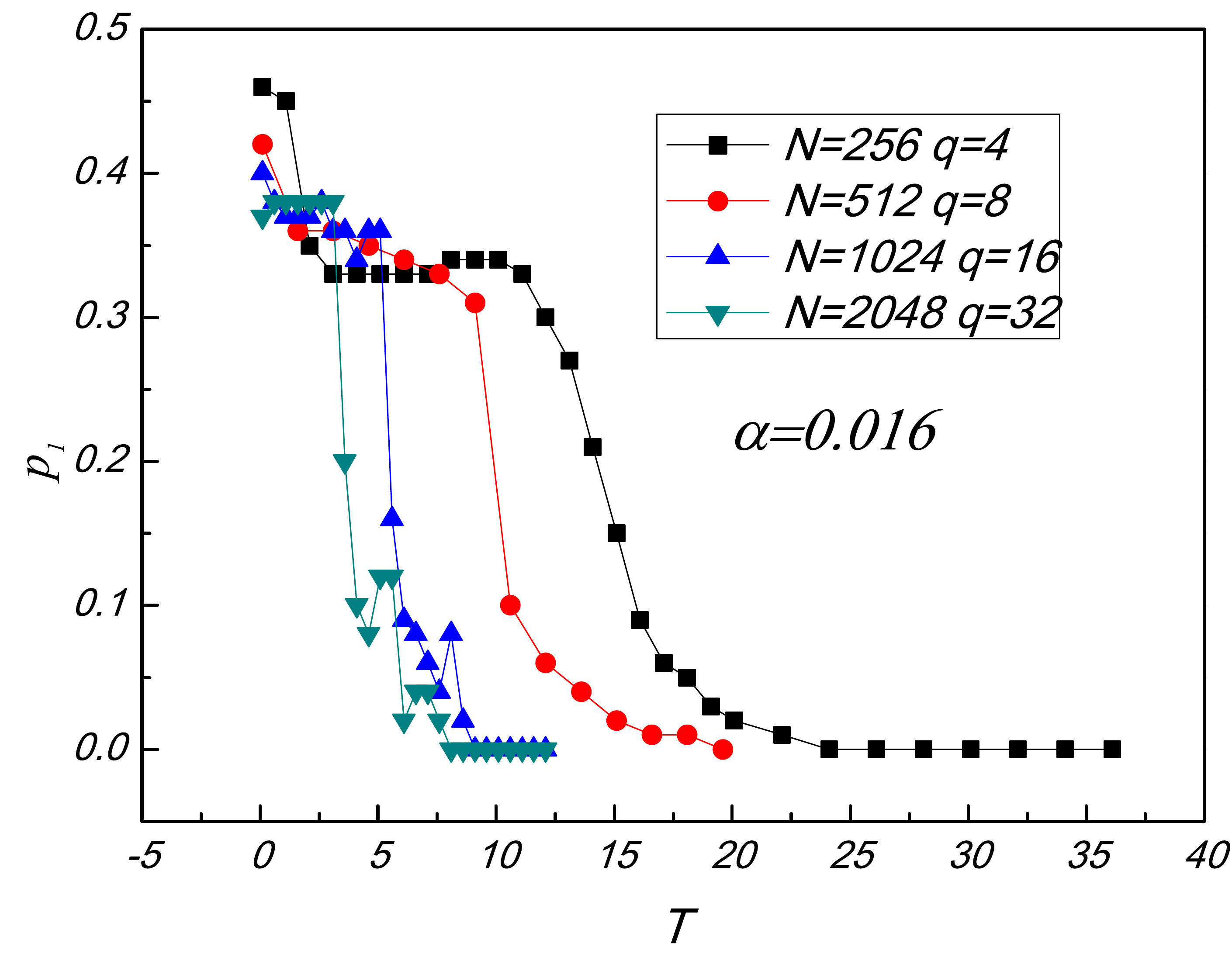}}
\subfigure[\ $p_1(T)$ with $\alpha=0.07$]{\includegraphics[width=\subfigwidthnew]{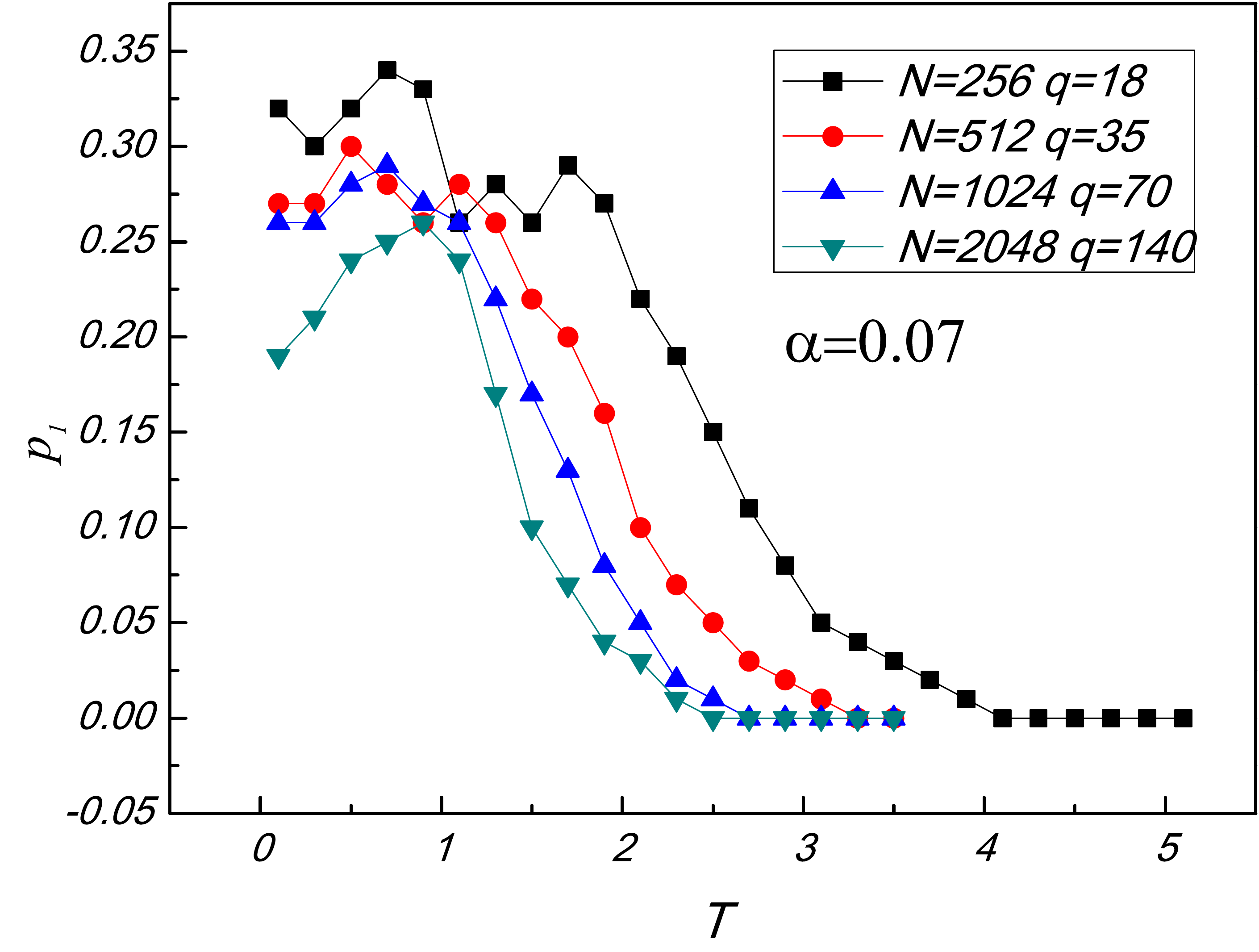}}
\subfigure[\ $p_1(T)$ with $\alpha=0.15$]{\includegraphics[width=\subfigwidthnew]{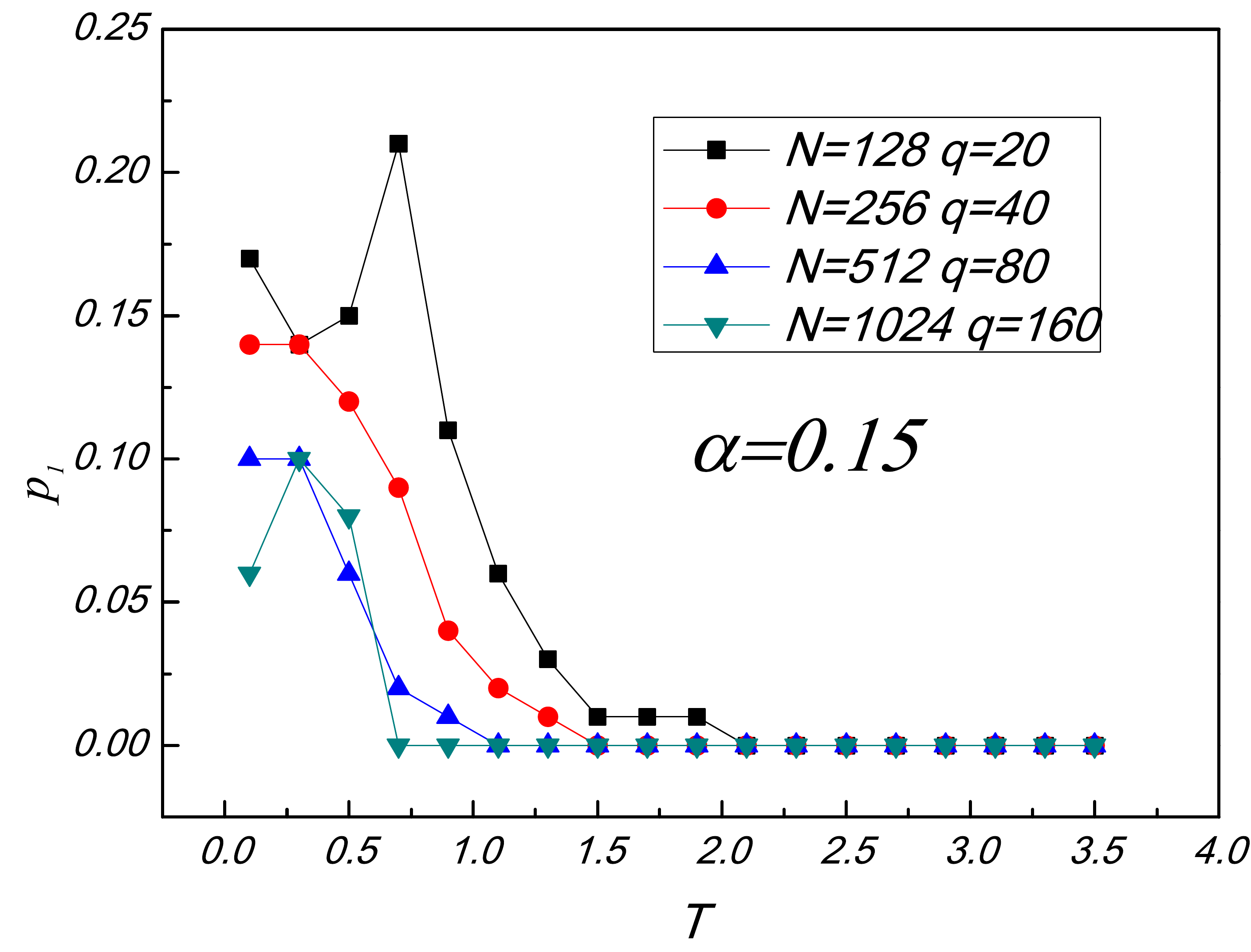}}
\end{center}
\caption{(Color online) Corresponding to \figref{fig:susAllalpha} and
\secref{sec:sec1}, each plot depicts the first phase transition point $p_1$
as a function of the temperature $T$ for systems with a fixed ratio of $\alpha=q/N$.
Panels (a), (b), and (c) show the results for $\alpha=0.016$, $\alpha=0.07$,
and $\alpha=0.15$, respectively.
System sizes range from $N=256$ to $2048$, and $q$ varies from $4$ to $160$
as indicated in each plot.
All panels show that when $\alpha$ is fixed, the value of the first transition
point $p_1$ decreases as the system size increases.
This behavior further indicates that the system becomes more complex to solve
in the thermodynamic limit.}
\label{fig:p1alpha}
\end{figure*}

\begin{figure*}[t]
\begin{center}
\subfigure[\ $\tau$ with
$\alpha=0.016$]{\includegraphics[width=\subfigwidthtoo]{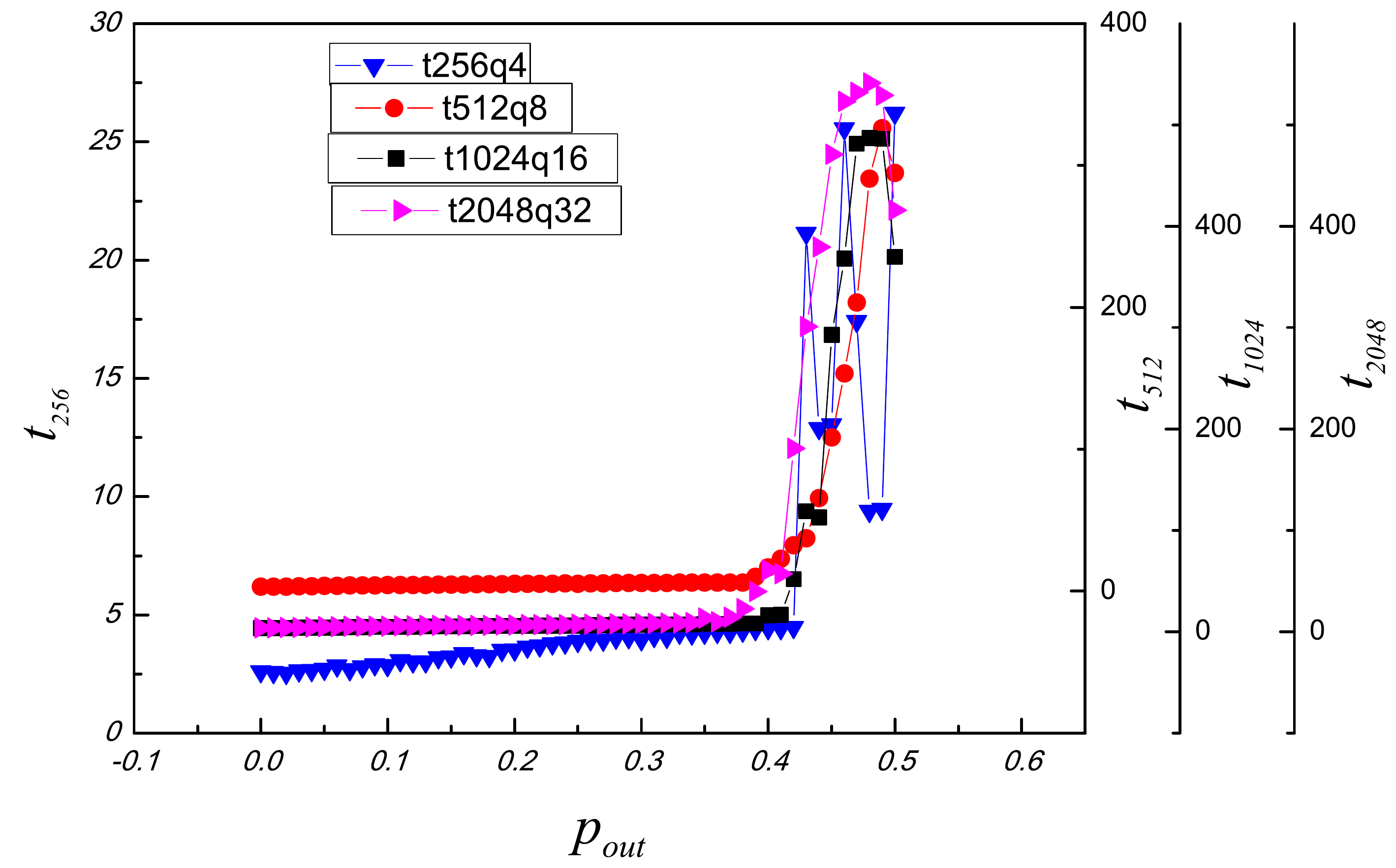}}
\subfigure[\ $\tau$ with
$\alpha=0.07$]{\includegraphics[width=\subfigwidthtoo]{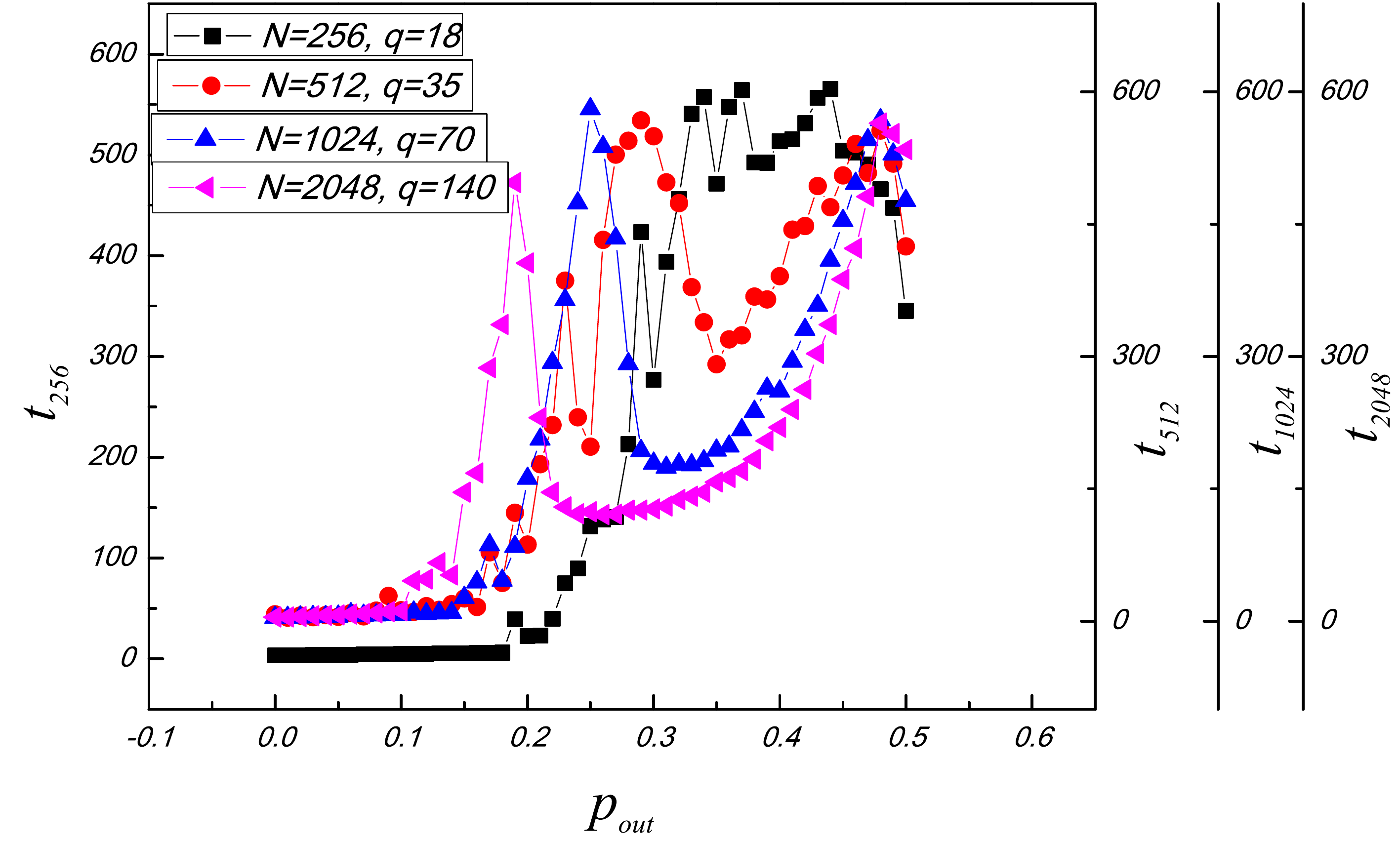}}
\subfigure[\ $\tau$ with
$\alpha=0.15$]{\includegraphics[width=\subfigwidthtoo]{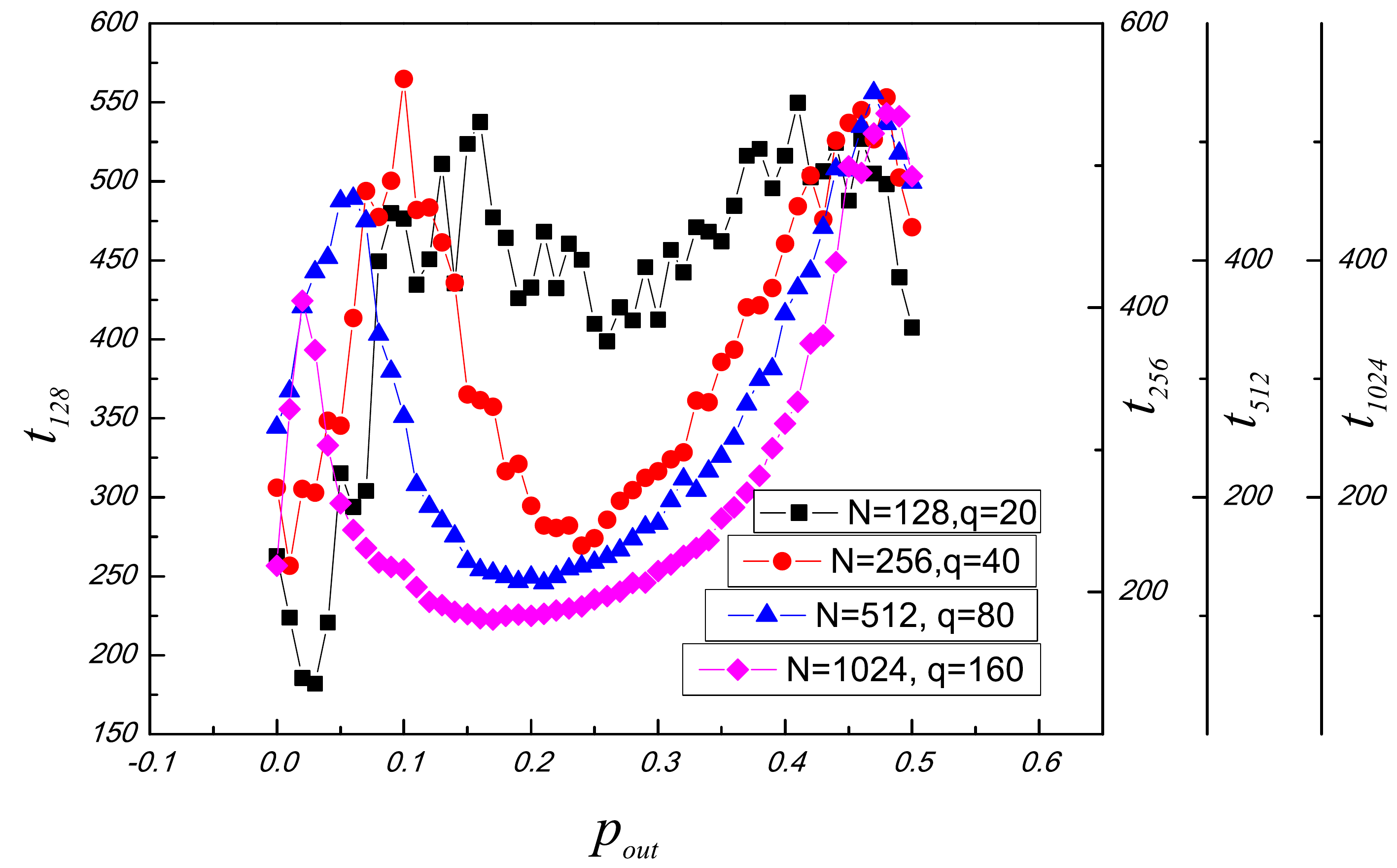}}
\end{center}
\caption{(Color online) Corresponding to \figref{fig:susAllalpha} and
\secref{sec:sec1}, the convergence time $\tau$ [see
\subfigref{fig:chianalog}{a}] as a function of noise $p_{out}$ at
zero temperature for the systems with a fixed ratio of $\alpha=q/N$.
Panels (a), (b), and (c) show the results for $\alpha=0.016$,
$\alpha=0.07$ and $\alpha=0.15$, respectively. System sizes range
from $N=256$ to $N=2048$, and $q$ varies from $4$ to $160$ as
indicates in each plot. The noise level $p_{out}$ at the first peak
of the convergence time corresponds to the first transition point
$p_1$ in \figref{fig:p1alpha} at zero temperature. As the system
size increases, the first peak in the convergence time moves to the
left. They share the same trend as in \figref{fig:sus2dalpha} and
\figref{fig:p1alpha}.} \label{fig:talpha}
\end{figure*}
% --- end chi fixed alpha plots ----------------------------------

% --- begin chi fixed q plots ------------------------------------
\begin{figure*}[t!]
\begin{center}
\subfigure[\ $N=256$,
$q=16$]{\includegraphics[width=\subfigwidth]{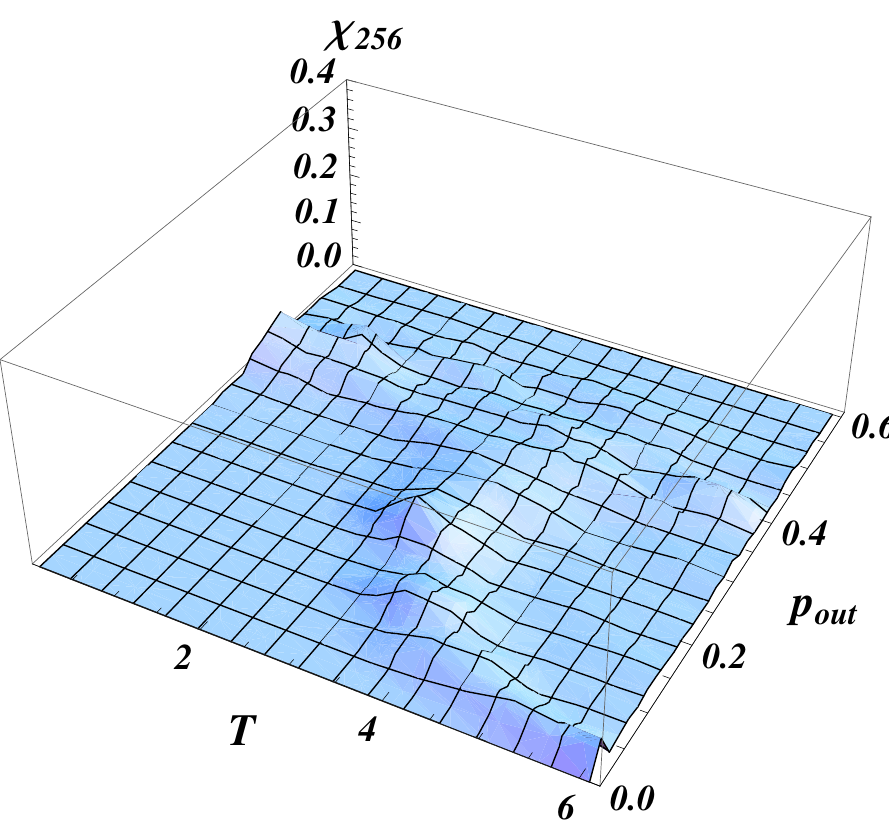}}
\subfigure[\ $N=512$,
$q=16$]{\includegraphics[width=\subfigwidth]{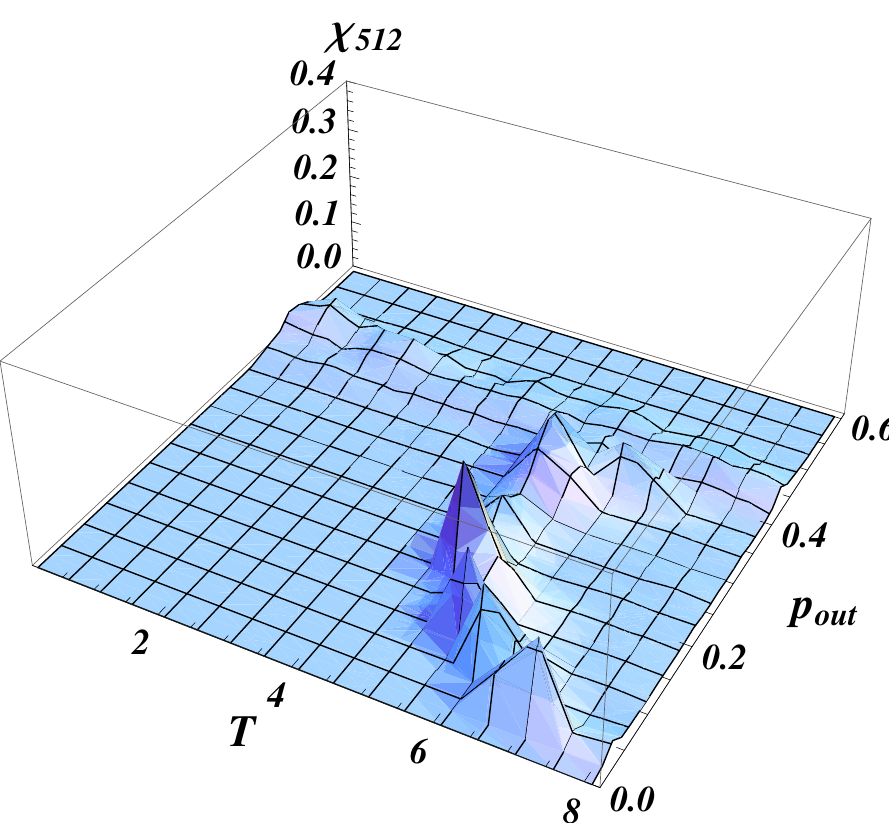}}
\subfigure[\ $N=1024$,
$q=16$]{\includegraphics[width=\subfigwidth]{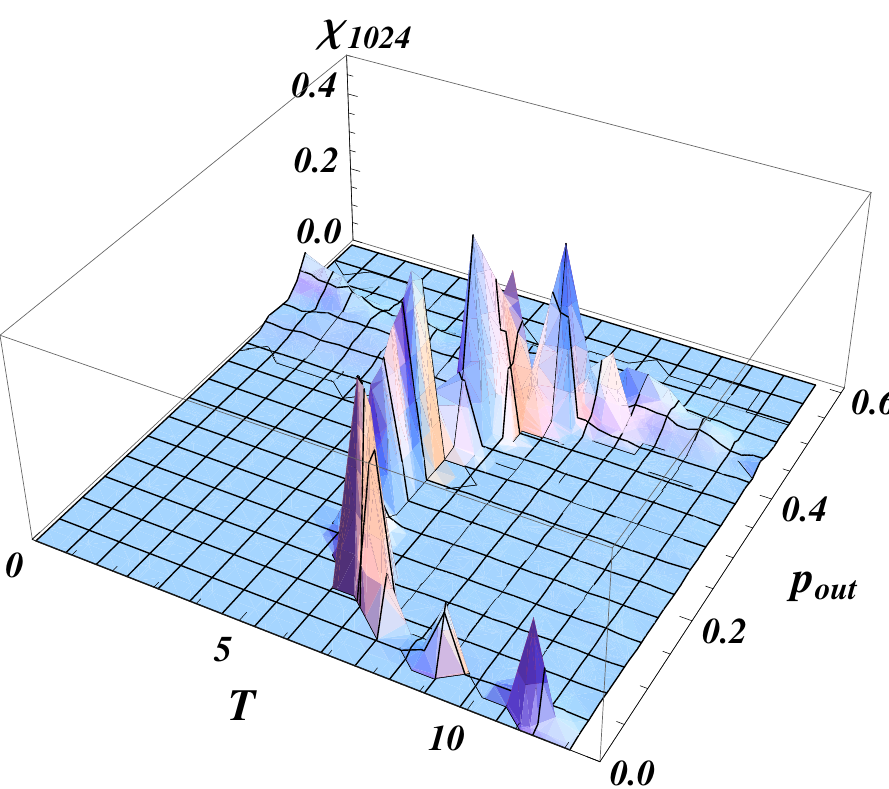}}
\subfigure[\ $N=2048$,
$q=16$]{\includegraphics[width=\subfigwidth]{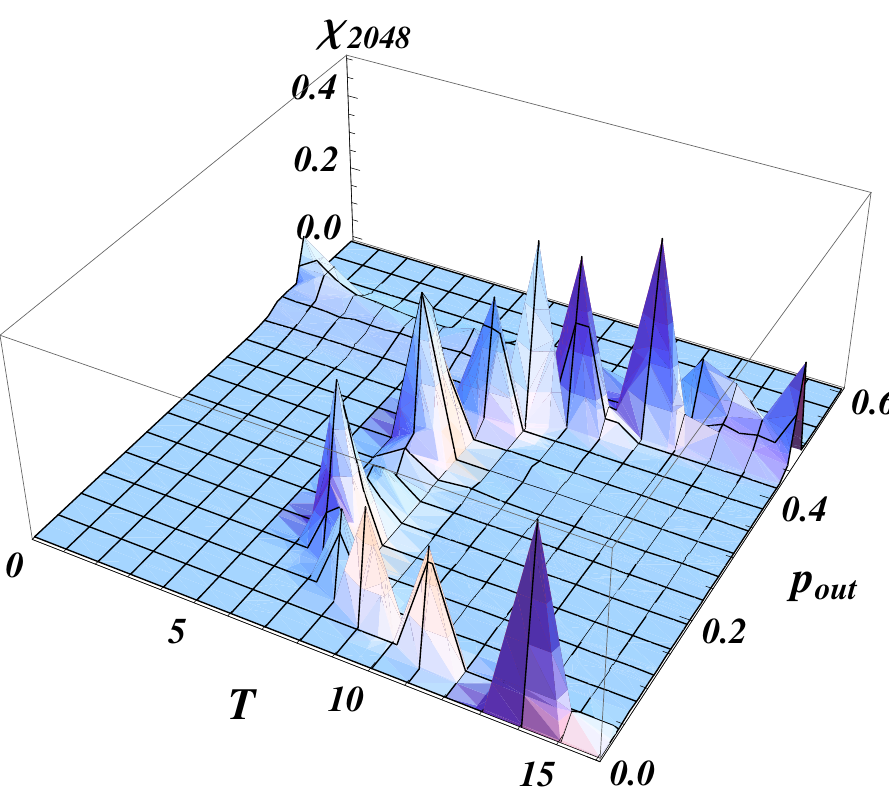}}
\subfigure[\ $N=256$,
$q=40$]{\includegraphics[width=\subfigwidth]{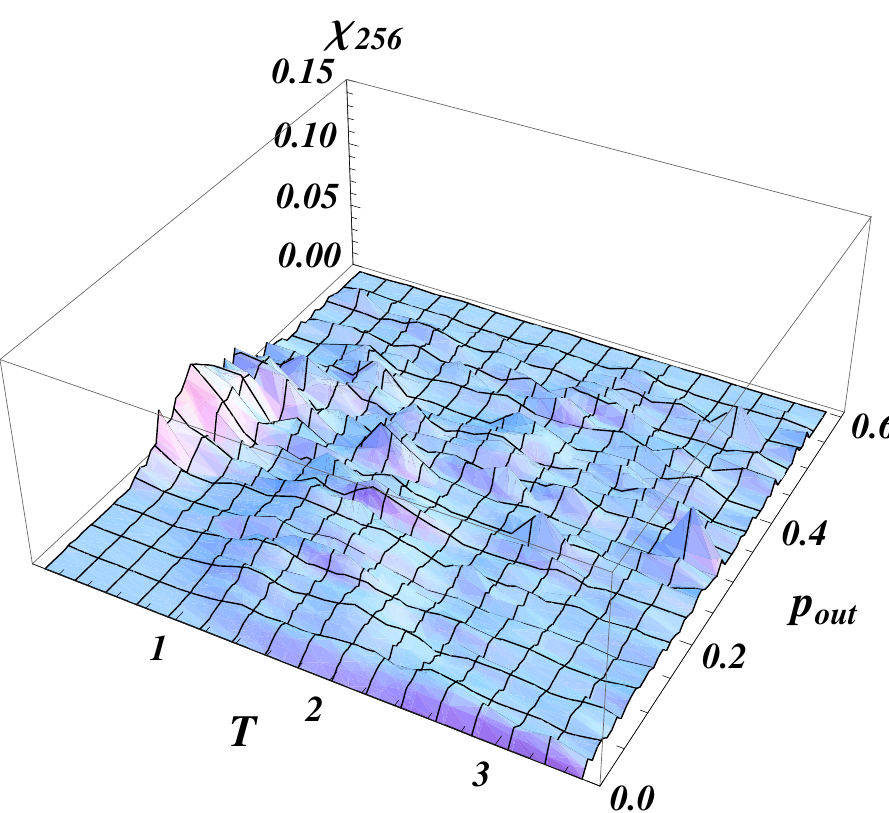}}
\subfigure[\ $N=512$,
$q=40$]{\includegraphics[width=\subfigwidth]{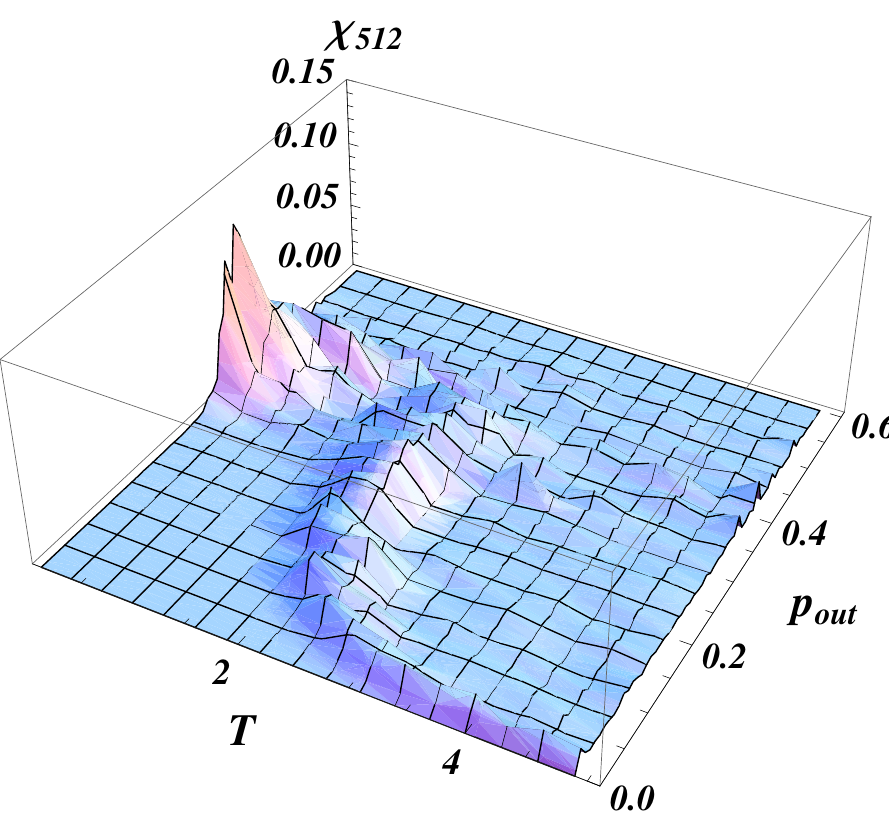}}
\subfigure[\ $N=1024$,
$q=40$.]{\includegraphics[width=\subfigwidth]{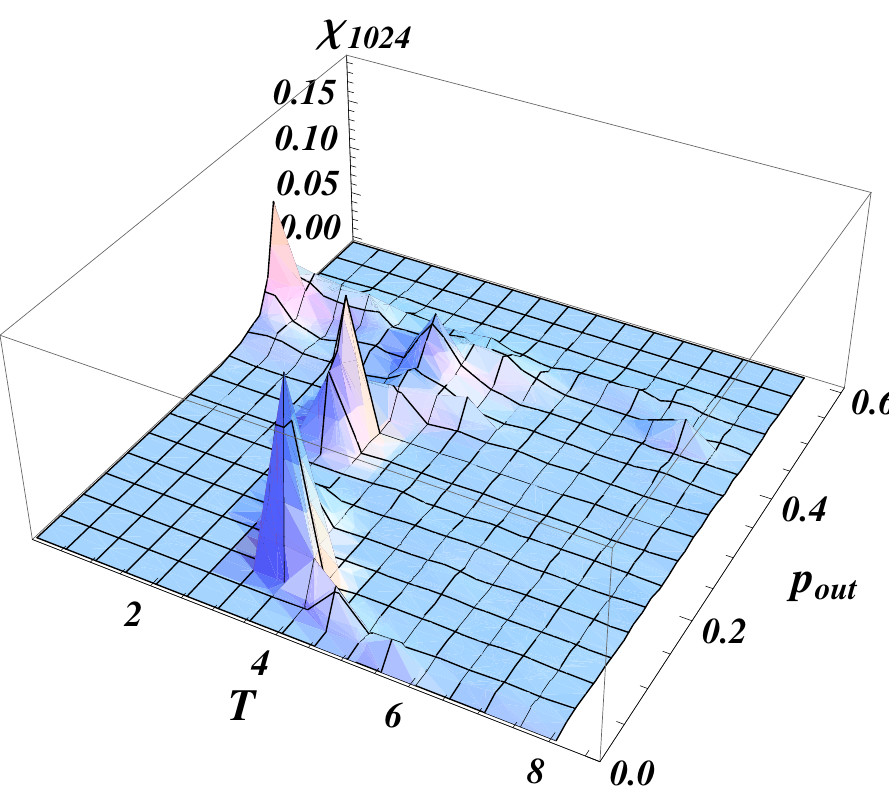}}
\subfigure[\ $N=2048$, $q=40$.]{\includegraphics[width=\subfigwidth]{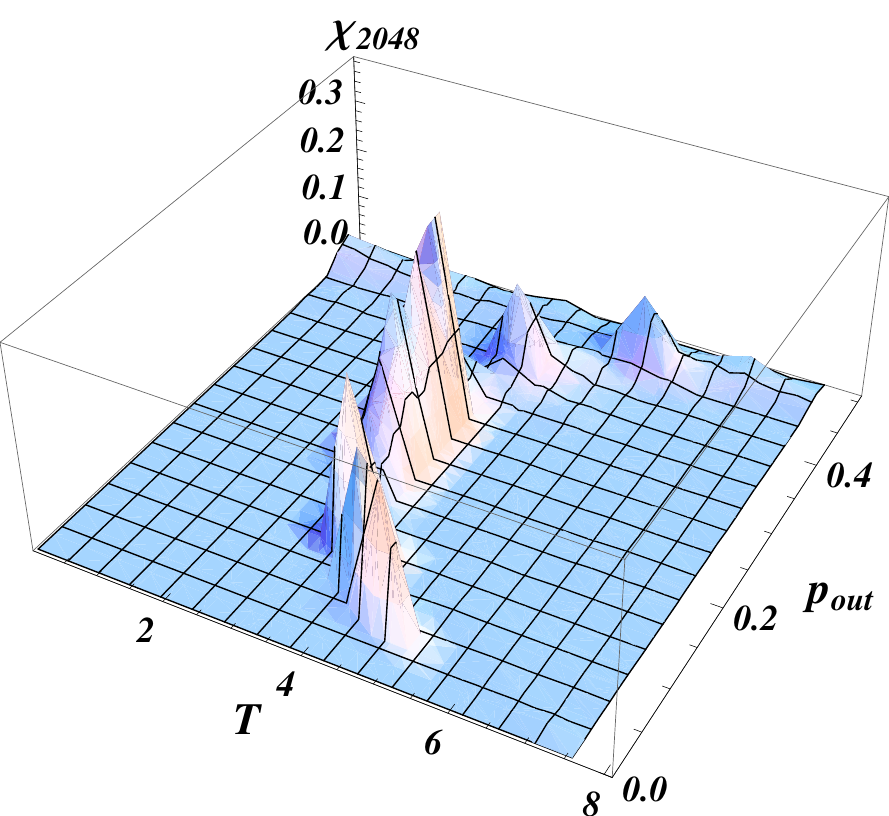}}
\subfigure[\ $N=512$,
$q=70$]{\includegraphics[width=\subfigwidth]{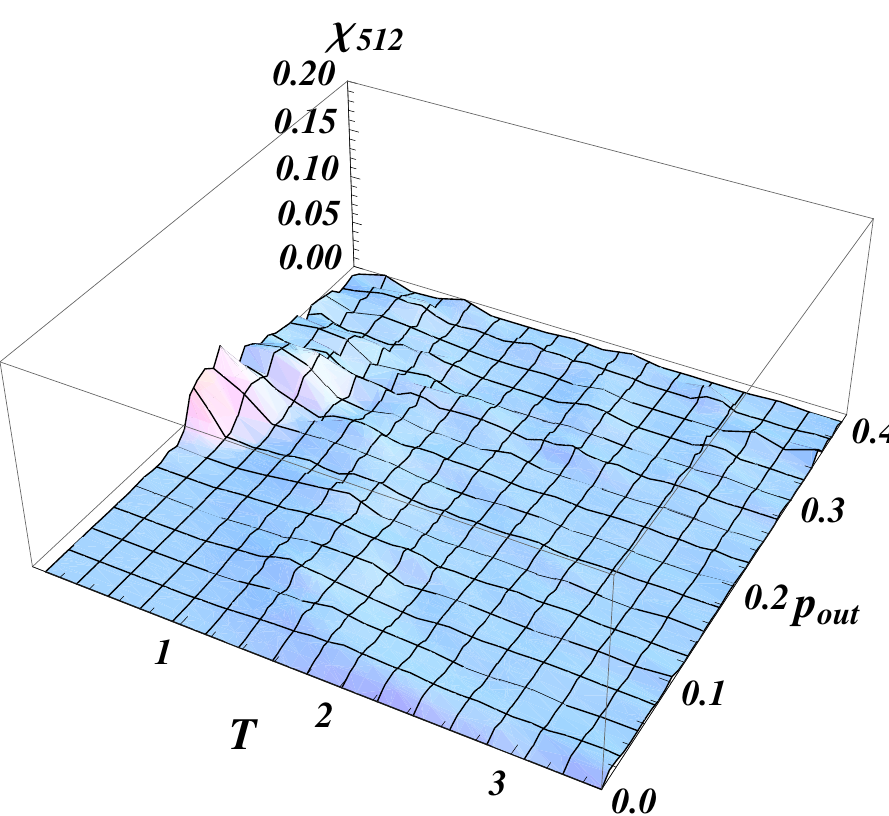}}
\subfigure[\ $N=800$,
$q=70$]{\includegraphics[width=\subfigwidth]{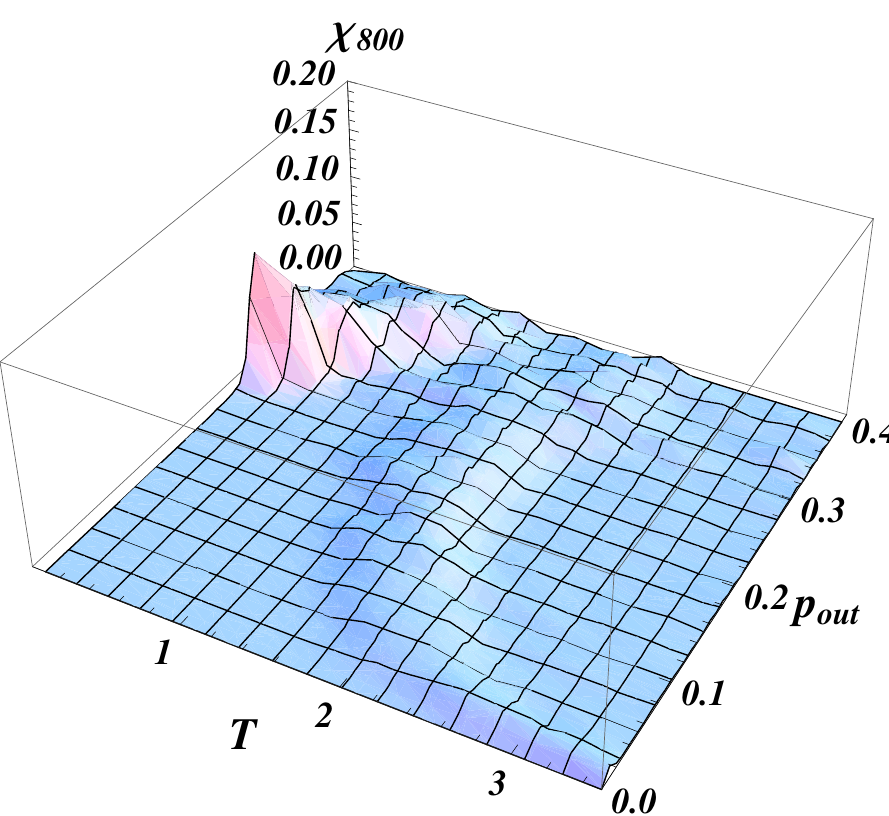}}
\subfigure[\ $N=1024$,
$q=70$]{\includegraphics[width=\subfigwidth]{sus1024q70}}
\subfigure[\ $N=2048$, $q=70$]{\includegraphics[width=\subfigwidth]{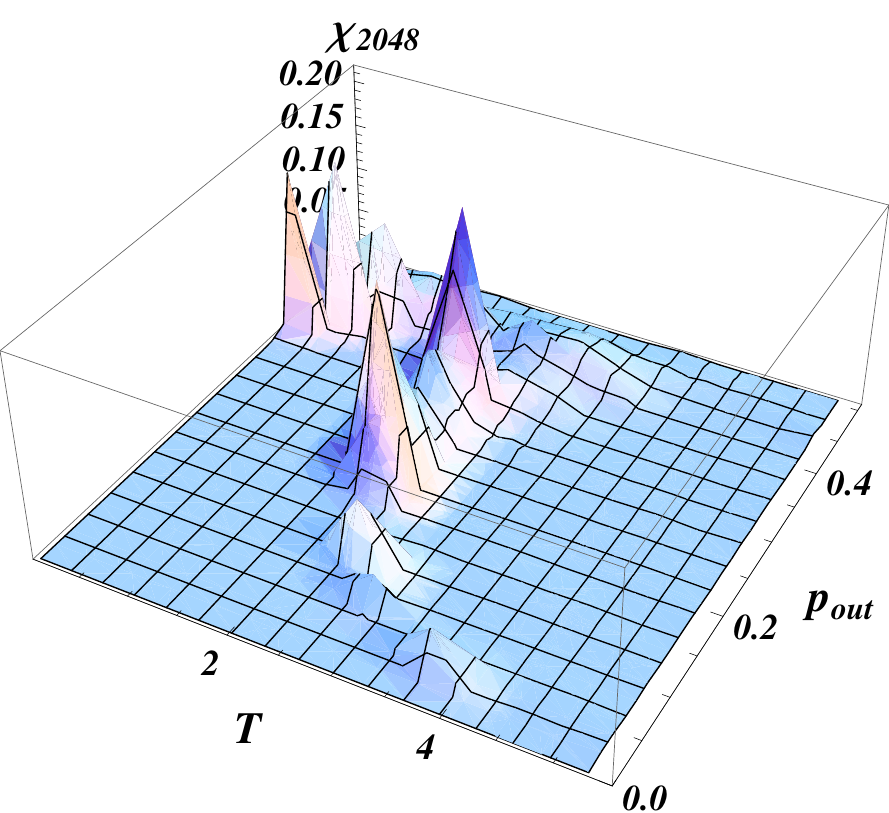}}
\end{center}
\caption{(Color online) Similar to \figref{fig:susAllalpha}, we plot of
$\chi(T,p_{out})$ as a function of temperature $T$ and noise level
$p_{out}$ for systems with the indicated number of nodes $N$,
communities $q$, and $\alpha=q/N$ ratio. Here, $q$ is fixed for each
row series, and we vary $\alpha$ (rows) to examine the behavior as
$N$ increases (columns). The heights of the susceptibility peaks at
higher $T$ increase across each series as $N$ increases whereas the
heights at low $T$ are relatively constant.
The $N=256$ node systems do not show clear hard or unsolvable phases, but
the transitions are strong at high temperature for most panels in the second
and third columns of plots.}
\label{fig:susAllq}
\end{figure*}
% --- end chi fixed q plots --------------------------------------

% --- begin 2D plots ---------------------------------------------
\begin{figure*}[t]
\begin{center}
\subfigure[\ hard phase boundary for $q=16$]{\includegraphics[width=\subfigwidthnew]{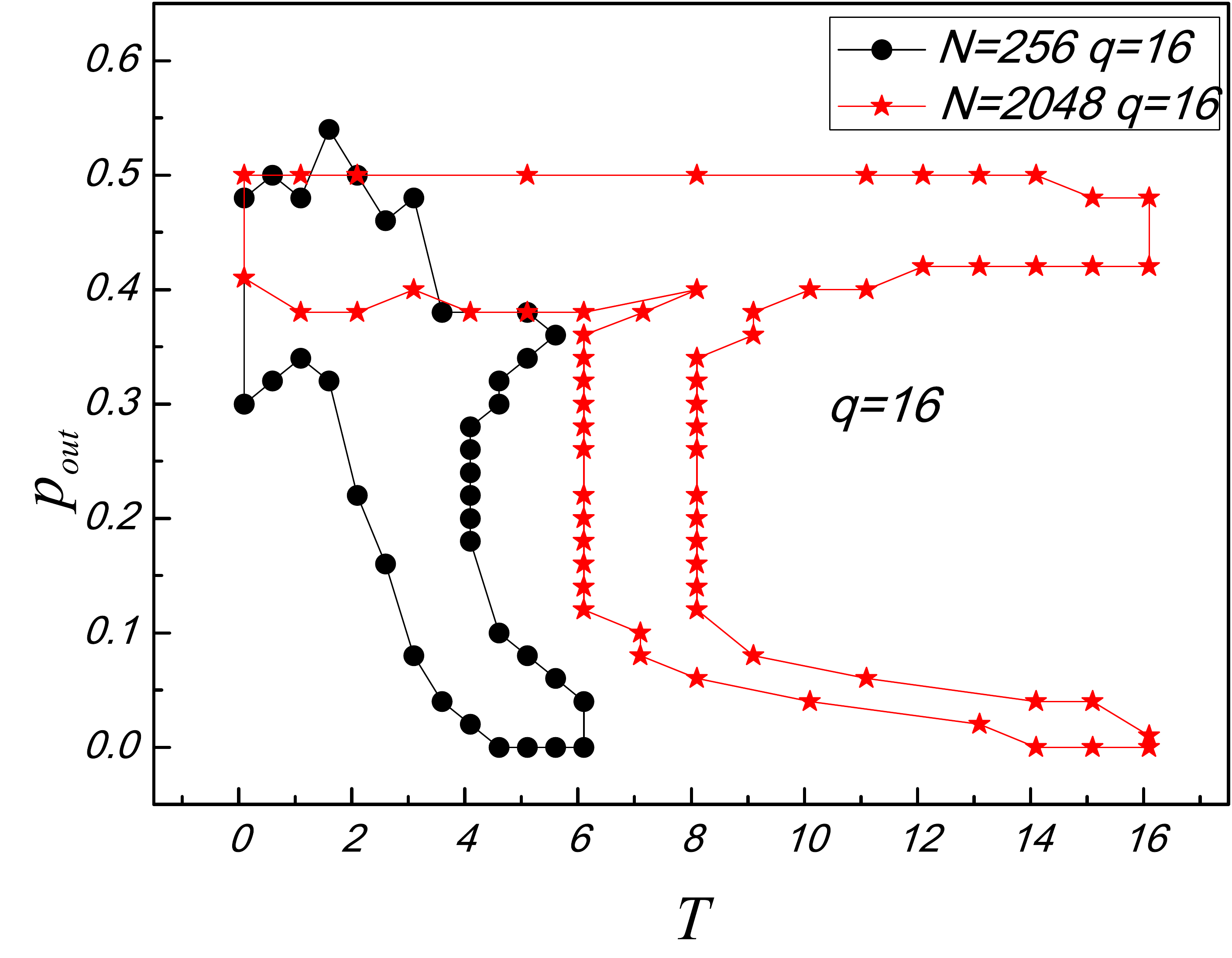}}
\subfigure[\ hard phase boundary for $q=40$]{\includegraphics[width=\subfigwidthnew]{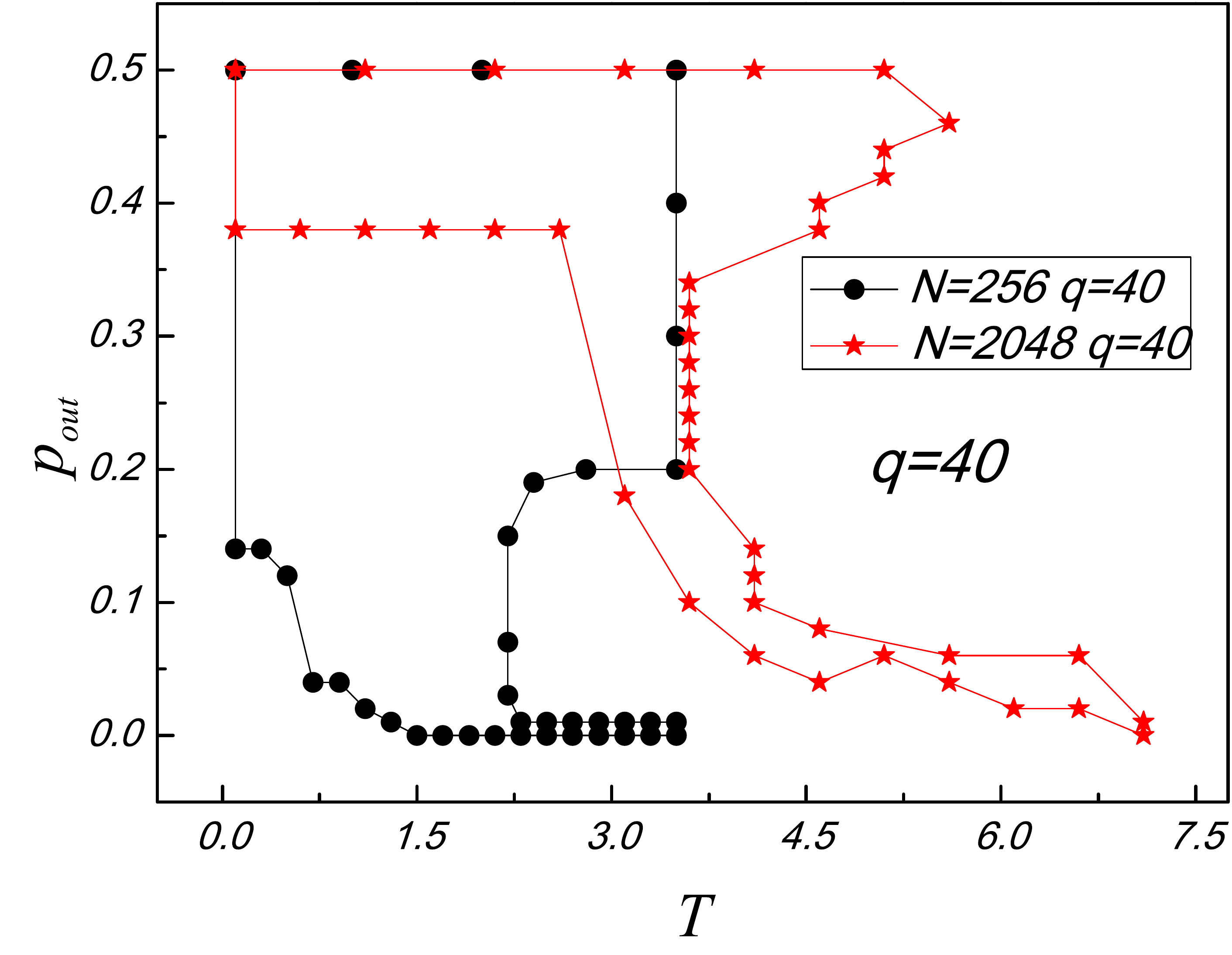}}
\subfigure[\ hard phase boundary for $q=70$]{\includegraphics[width=\subfigwidthnew]{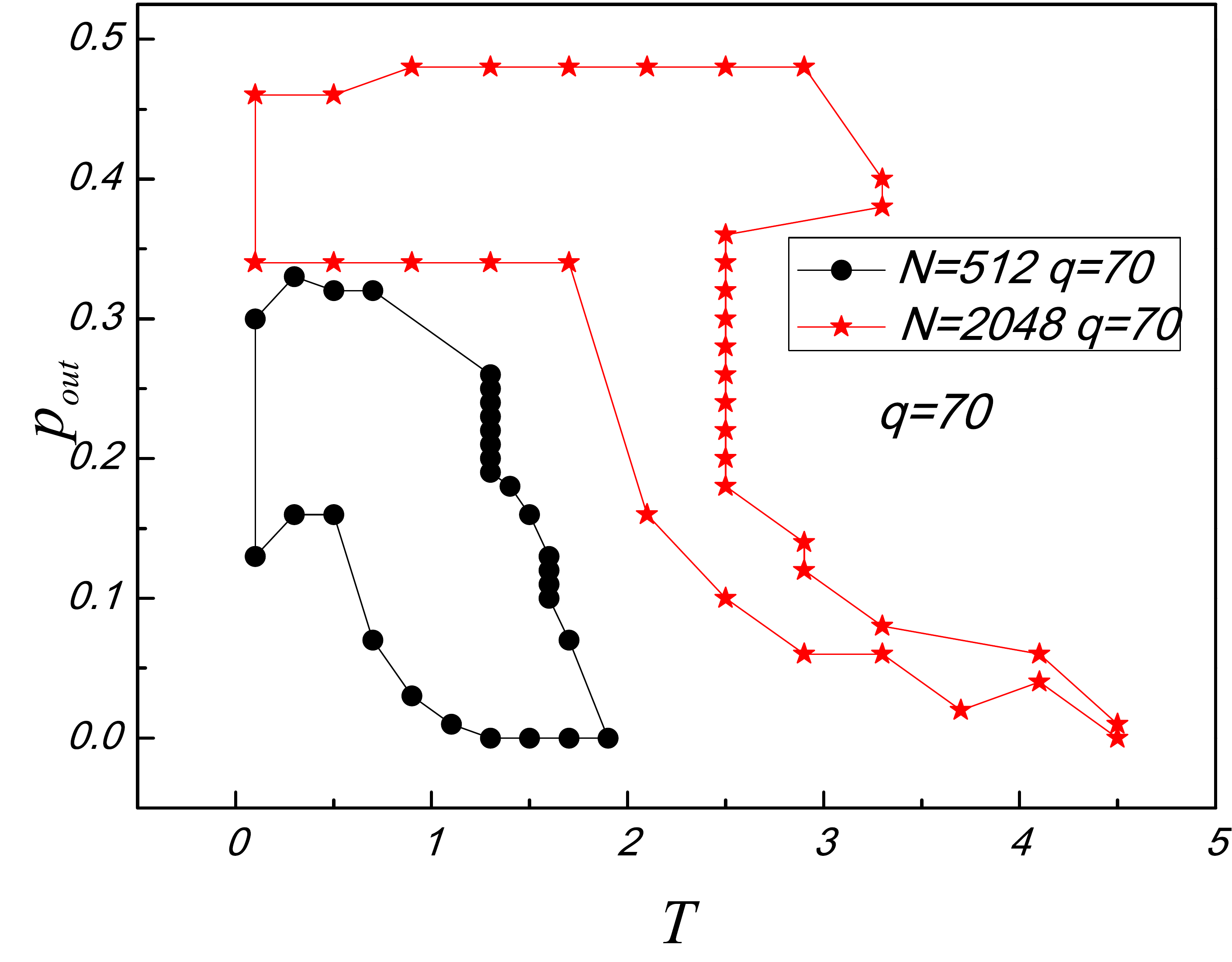}}
\end{center}
\caption{(Color online) Corresponding to \figref{fig:susAllq} and \secref{sec:sec2},
each plot depicts the boundaries of the hard phase
for the system series with a fixed number of communities $q$ where panels (a),
(b), and (c) correspond to $q=16$, $40$, and $70$, respectively.
System sizes range from $N=256$ to $2048$ as indicated.
For each $q$, the area of the hard phase becomes progressively
narrower which indicates clearer transitions from the easy to unsolvable
phases in the thermodynamic limit.}
\label{fig:sus2dq}
\end{figure*}

\begin{figure*}[t]
\begin{center}
\subfigure[\ $p_1(T)$ with $q=16$]{\includegraphics[width=\subfigwidthnew]{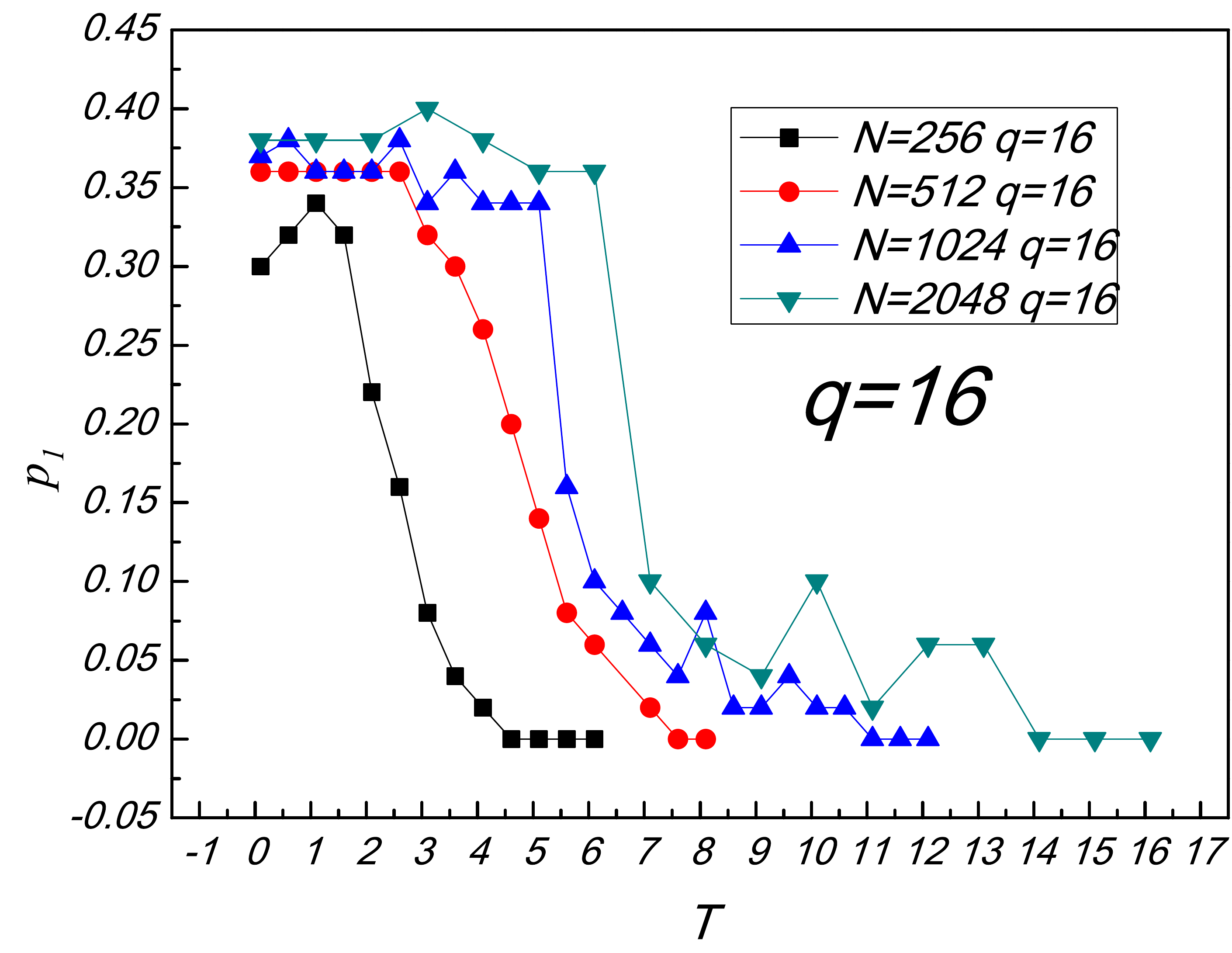}}
\subfigure[\ $p_1(T)$ with $q=40$]{\includegraphics[width=\subfigwidthnew]{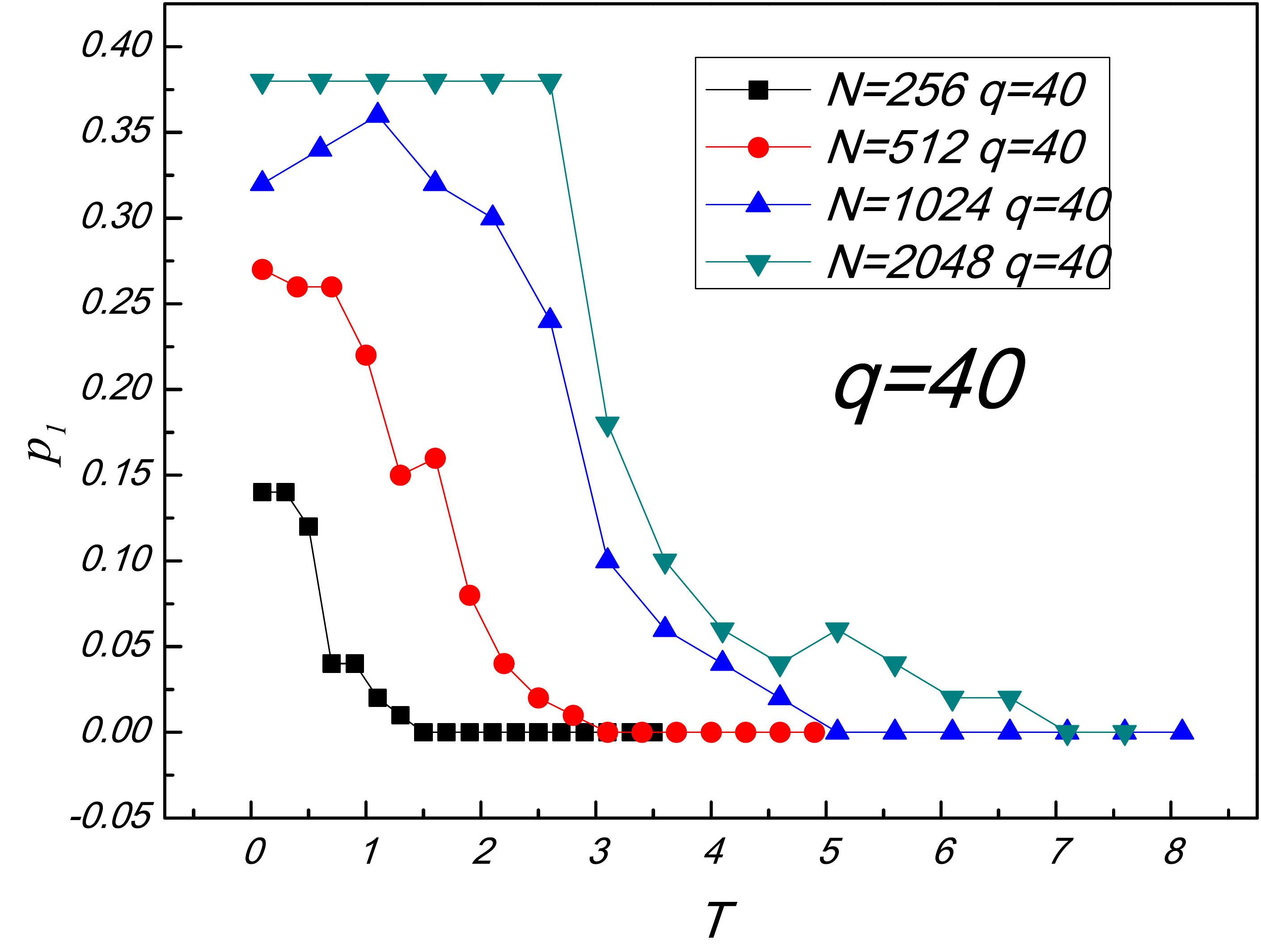}}
\subfigure[\ $p_1(T)$ with $q=70$]{\includegraphics[width=\subfigwidthnew]{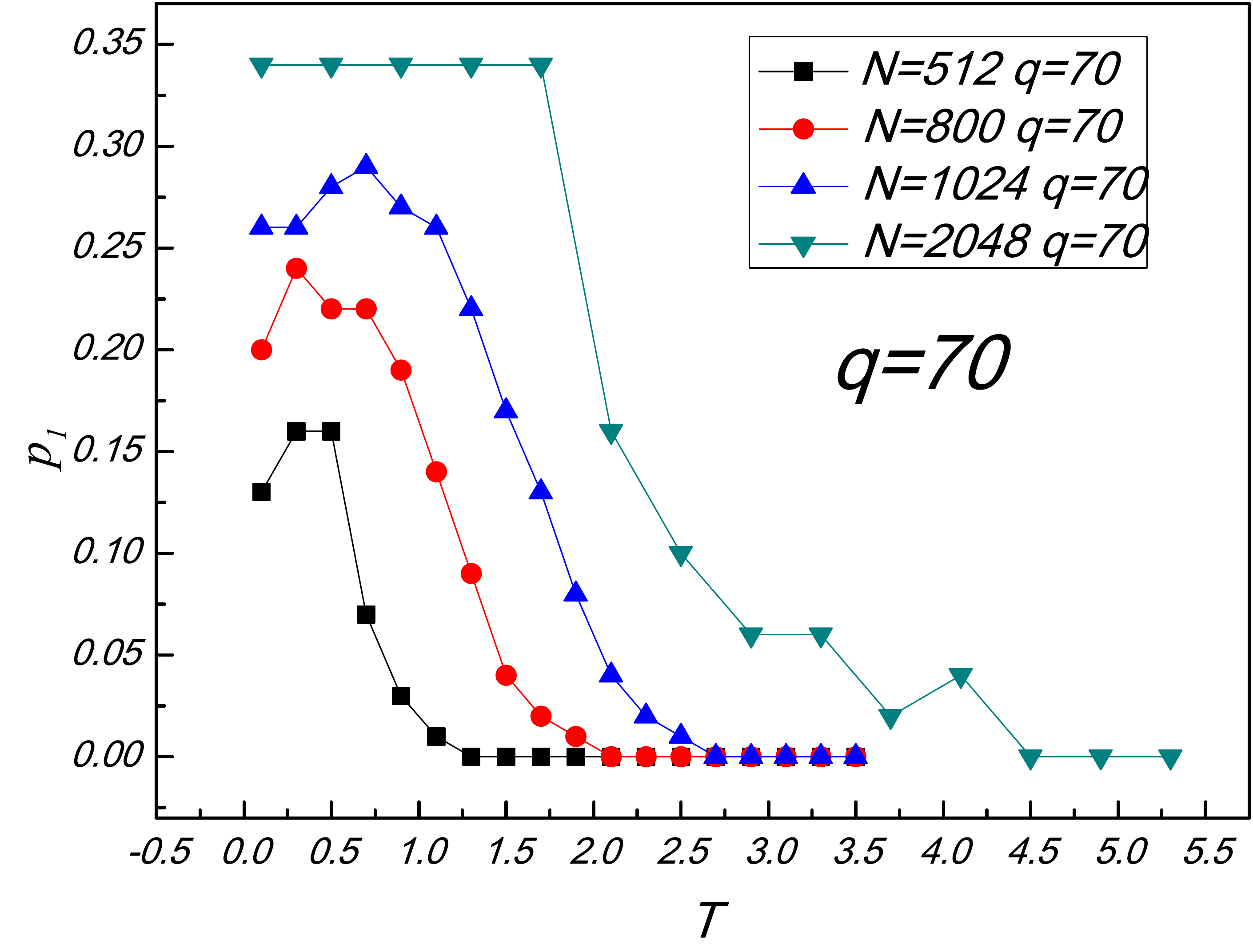}}
\end{center}
\caption{(Color online) Corresponding to \figref{fig:susAllq} and \secref{sec:sec2},
each plot depicts the first phase transition point $p_1$
as a function of temperature $T$ for systems with a fixed $q$.
Panels (a), (b), and (c) show the results for $q=16$, $40$, and $70$, respectively.
System sizes range from $N=256$ to $2048$ as indicated in each plot.
All panels show that the first transition point increases as the system size
increases which is consistent with the complexity trend of the system series.}
\label{fig:p1q}
\end{figure*}

\begin{figure*}[t]
\begin{center}
\subfigure[\ $\tau$ with
$q=16$]{\includegraphics[width=\subfigwidthtoo]{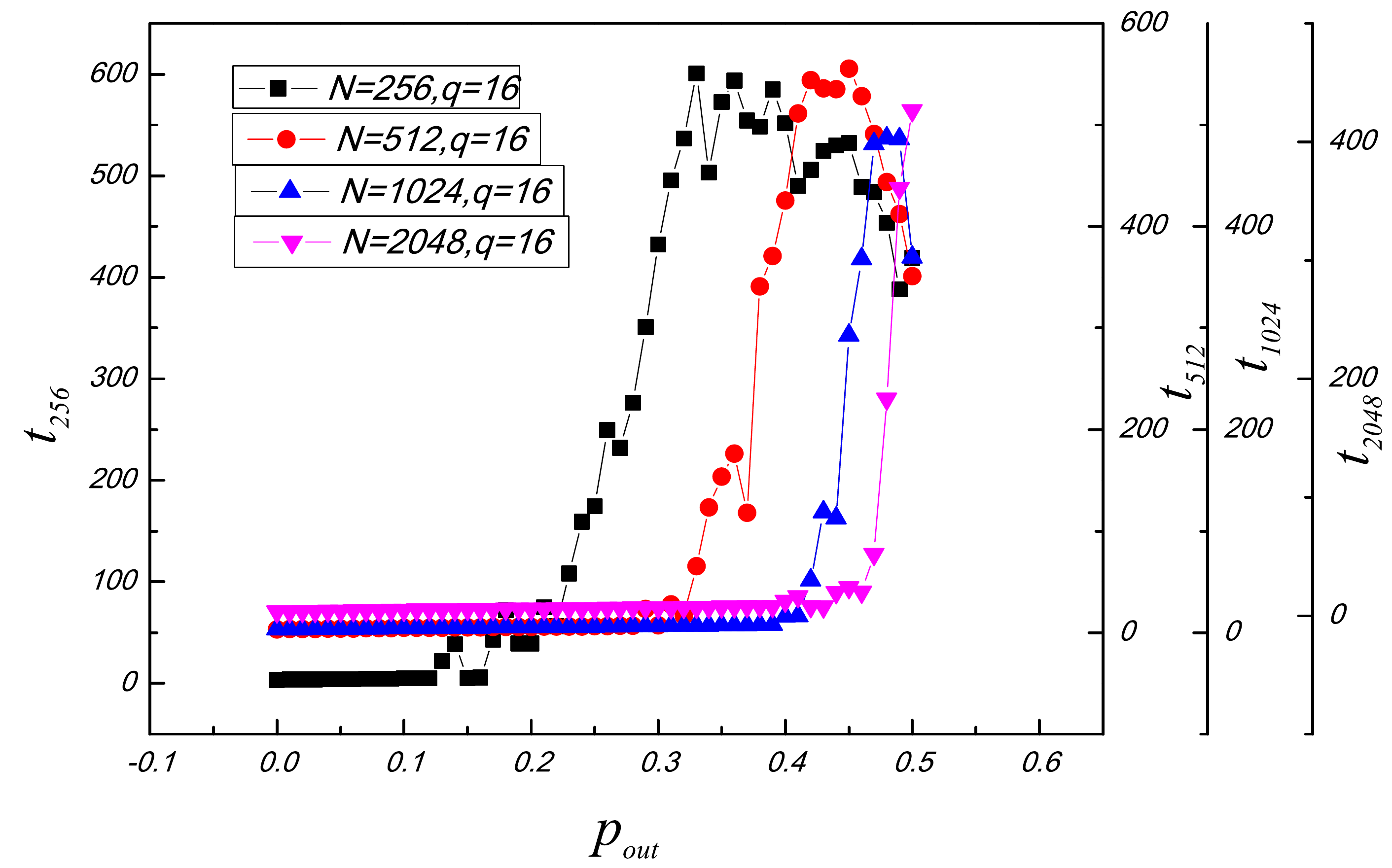}}
\subfigure[\ $\tau$ with
$q=40$]{\includegraphics[width=\subfigwidthtoo]{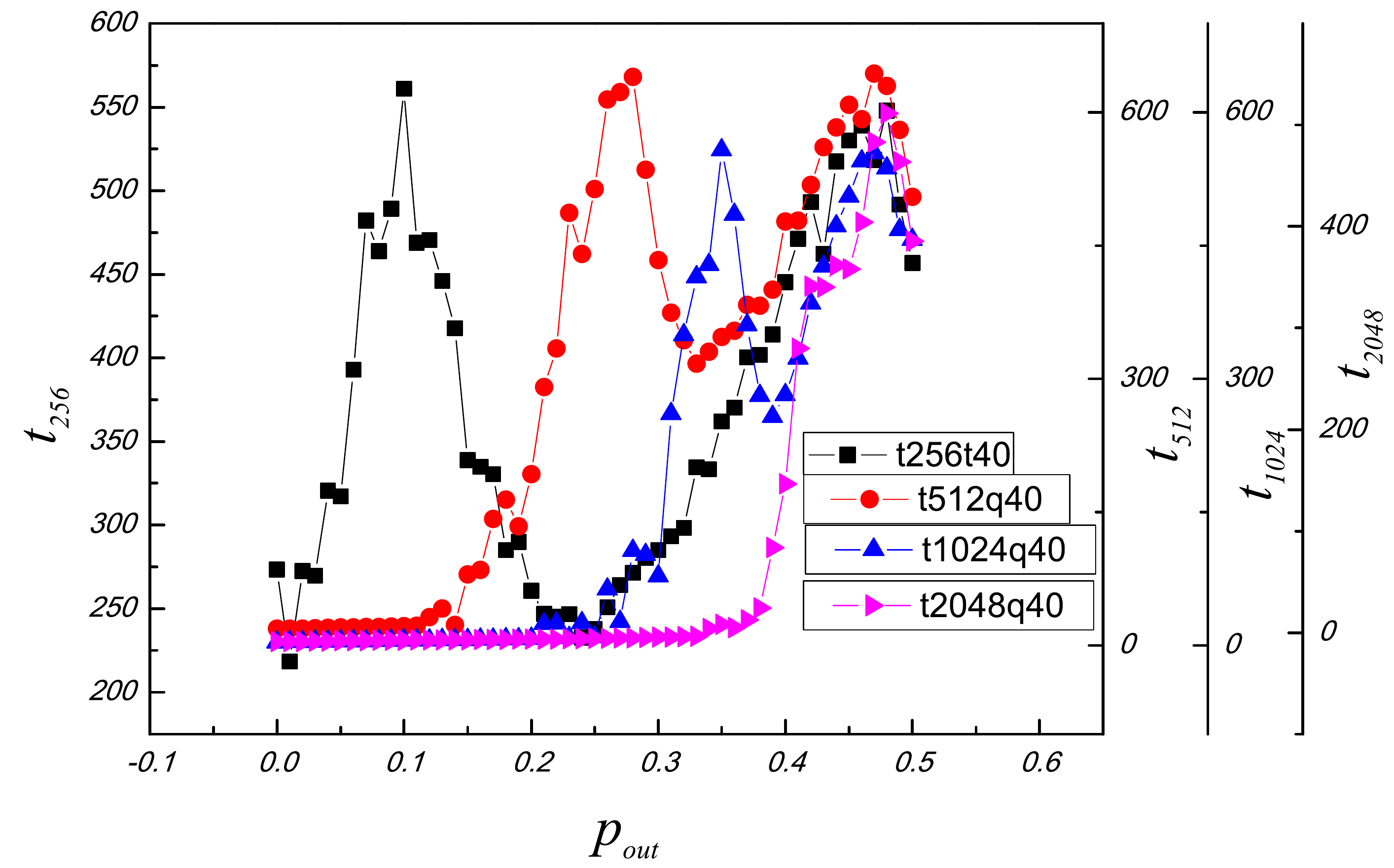}}
\subfigure[\ $\tau$ with
$q=70$]{\includegraphics[width=\subfigwidthtoo]{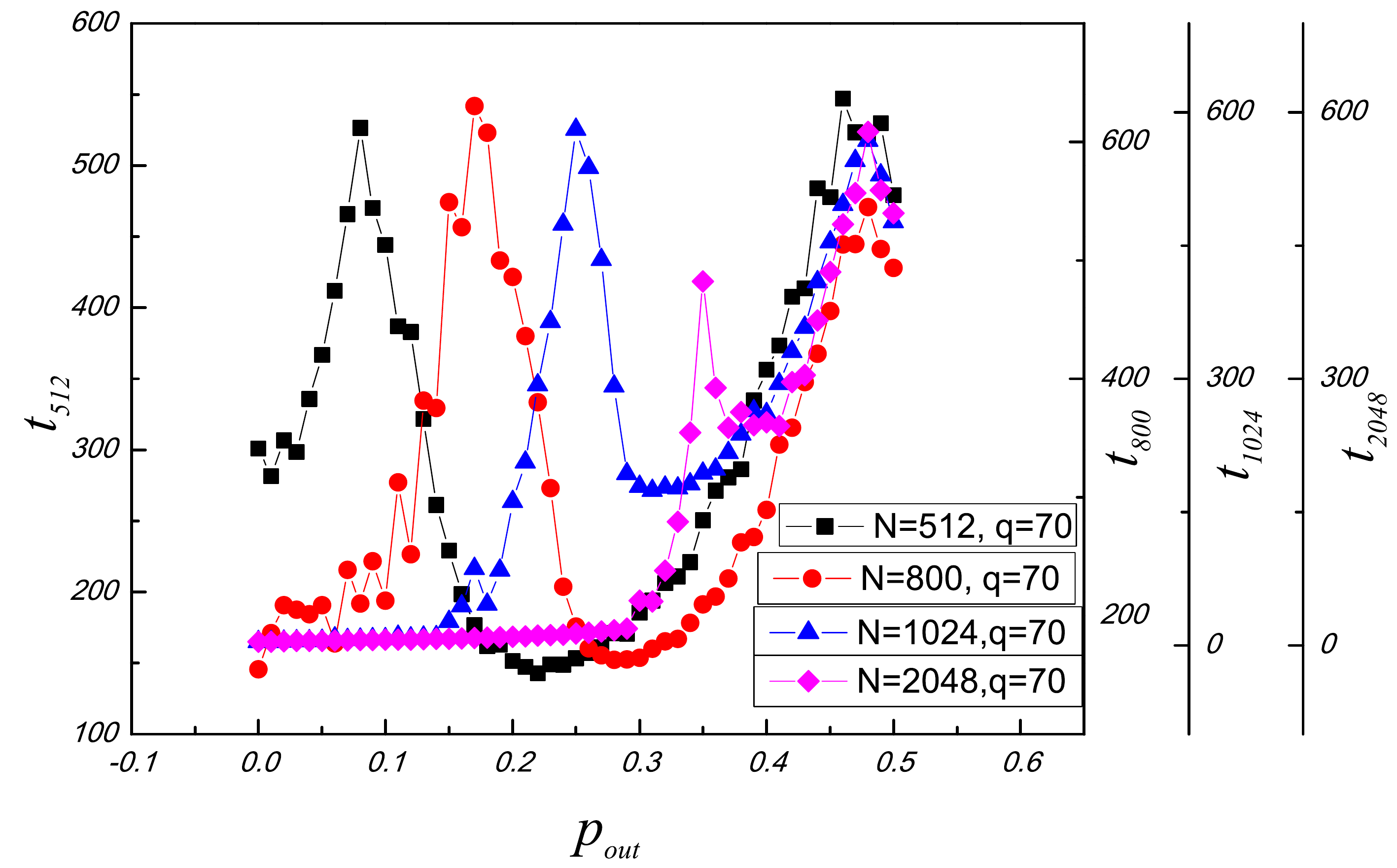}}
\end{center}
\caption{(Color online) Corresponding to \figref{fig:susAllq} and
\secref{sec:sec2}, the convergence time $\tau$
[\subfigref{fig:chianalog}{a}] as a function of noise $p_{out}$ for
systems with fixed $q$. Panels (a), (b) and (c) show the results for
$q=16$, $40$, and $70$, respectively. System sizes vary from $N=256$
to $2048$ as indicated in each plot. The noise level $p_{out}$ at
the first peak of the convergence time corresponds to the initial
transition point $p_1$ in \figref{fig:p1q} at zero temperature. As
the system size increases, the first peak in the convergence time
moves to the right. They share the same trend as in
\figsref{fig:sus2dq}{fig:p1q}.} \label{fig:tq}
\end{figure*}
% --- end 2D plots -----------------------------------------------

\subsection{$\chi(T,p_{out})$ at fixed $\alpha=q/N$}
\label{sec:sec1}

We show the phase transitions in terms of three-dimensional (3D) plots
with the computational susceptibility $\chi(T,p_{out})$ for a range
of system sizes $N$ and numbers of communities $q$.
First, we fix the ratio $\alpha=q/N$ and study the phase transitions as $N$
increases.
Then we test a range of systems with fixed $q$ as $N$ increases.

\subsubsection{$\chi(T,p_{out})$ at $\alpha=0.016$}

In \figref{fig:susAllalpha} panels (a) through (d), we begin the
analysis at a small $\alpha = q/N =0.016$ ratio.
The results for four system sizes are shown: $N=256$, $N=512$,
$N=1024$ and $N=2048$ which maintain a fixed ratio of $\alpha$
across the respective rows.
Each plot shows the easy, hard, and unsolvable phases.

The two ``ridges'' in each plot denote the hard phases.
The height of the first ridge at low temperature decreases as the system
size increases while the height of the second ridge at high temperature
increases in the same process.
This finite size scaling behavior for the hard phase at high temperature
indicates that the phase transition at high temperature exists in
the thermodynamic limit.
However, the phase transition at low temperature will disappear
in the same limit.
In the meantime, the ridge in the high temperature will gradually
expand into the low temperature region as the system size increases.
Thus, for the systems with the small ratio of $\alpha$, the phase
transition will exist in almost the entire temperature range in the
thermodynamic limit (see \secref{sec:NIcliques}).

The ``easy'' phase shrinks and the unsolvable phase expands as $N$
increases. In detail, the approximate area of the easy phase on the
left corner in panel (a) is in the range of $T\in (0,20)$ and
$p_{out}\in (0,0.4)$. The area of the unsolvable phase on the right
upper corner is in the range of $T\in (20,+\infty)$ and
$p_{out}\in(0,0.4)$. As the system size increases from $N=256$ in
panel (a) to $N=1024$ in panel (c), the area of the easy phase
shrinks to the range of $T\in(0,5)$ and $p_{out}\in (0,0.4)$ while
the unsolvable phase expands to $T\in(5,+\infty)$ and
$p_{out}\in(0,0.4)$. As the system size further increases to
$N=2048$ in panel (d), the easy phase further shrinks to the range
of $T\in (0,4)$ and $p_{out}\in (0,0.4)$ while the unsolvable phase
expands to $T\in (4,+\infty)$ and $p_{out}\in (0,0.4)$. We note that
the range of $p_{out}$ for the easy phase does not decrease as the
system size increases.

In order to track the range of the hard phases, we further display a
set of ``boundary'' plots in \figref{fig:sus2dalpha} as well as the first
transition point $p_1$ as the function of temperature in \figref{fig:p1alpha}.
For the system series with the fixed $\alpha=0.016$ discussed above,
the 2D ``hard phase'' boundaries and the values of the first transition
points are in panel (a) of \figref{fig:sus2dalpha} and \figref{fig:p1alpha},
respectively.

In \subfigref{fig:sus2dalpha}{a}, the area of the hard phase shrinks,
and its area at high temperature becomes narrower as the system size
increases.
Specifically, the width of the hard phase for $N=256$ is about $\Delta T=6$,
while it only extends to $\Delta T=1$ for the $N=2048$.
Together with the 3D phase diagrams in panels (a)--(d) of \figref{fig:susAllalpha},
we conclude that the hard phase at the high temperature becomes sharper
in the thermodynamic limit.
% for the system series with the fixed ratio $\alpha=0.016$.

The boundaries of the hard phase at low temperature are more easily seen
in \subfigref{fig:p1alpha}{a} where we plot the first transition point $p_1$
as the function of temperature $T$ for a range of systems.
The plots confirm the observations in \subfigdref{fig:susAllalpha}{a}{d}
regarding the constant $p_{out}$ range.
That is, the range of $p_{out}$ for the easy phase does not decrease as the
system size increases [in \subfigref{fig:p1alpha}{a}, $p_1$ collapses before
$T\leq 5$ for all the systems].
This behavior hints that the first transition point $p_1$ at low temperature
and small $\alpha$ remains constant in the thermodynamic limit.

As depicted in \subfigref{fig:chianalog}{a}, the convergence time
$\tau$ provides another view of the phase transition. We plot $\tau$
as a function of noise level $p_{out}$ in \subfigref{fig:talpha}{a}
for systems with a fixed ratio of $\alpha=q/N=0.016$. The value
$p_{out}$ at the first peak of the convergence time in each system
is consistent with the first transition point $p_1$ observed in
\subfigref{fig:p1alpha}{a}. As the system size increases, the peak
convergence time shifts to the left, which corresponds to the lower
value of $p_1$.

\subsubsection{$\chi(T,p_{out})$ at $\alpha=0.07$}

For $\alpha=0.07$, the phase transitions are presented in
\figref{fig:susAllalpha} panels (e) through (h).
The phases in panel (e) are noisy compared to panels (f) through (h), and
all of the systems are more complicated than the plots with $\alpha=0.016$.
As $N$ increases, the phase transitions become more clear.
However, contrary to panels (a) through (d), the phase transition at low
temperature becomes more prominent as $N$ increases, and the transition
at high temperature stays roughly constant.
Specifically, the height of the susceptibility peak at low temperature
increases from $\chi=0.01$ at $N=256$ in panel (e), $\chi=0.05$ at $N=512$
in panel (f), $\chi=0.1$ for $N=1024$ in panel (g), and finally reaches
$\chi=0.2$ in panel (h) with $N=2048$.
The phase transitions in this series appear to be persistent.

The easy phase (lower left of each panel) decreases in area as the system
size increases.
This is the same trend that was observed in the previous $\alpha=0.016$
series implying that the easy phase will tend to decrease in the
thermodynamic limit up to a threshold (see \secref{sec:NIcliques}).
Specifically, the easy phase in the smallest system in panel (e) covers
the range of $T\in (0,3)$ and $p_{out}\in (0,0.3)$ while in the large
system in panel (h) covers $T\in(0,1.5)$ and $p_{out}\in(0,0.2)$.
The range for $p_{out}$ in the easy phase decreases as the $N$ increases
which differs from the $\alpha=0.016$ data where the noise $p_{out}$
stayed at a roughly constant range of $p_{out}\in(0,0.4)$.
In both series for $\alpha=0.016$ and $0.07$, the value of the initial
transition point $p_1$ decreases in the thermodynamic limit.

The corresponding 2D plots of the hard phase boundaries and the first
transition points $p_1$ are displayed in \subfigref{fig:sus2dalpha}{b}
and \subfigref{fig:p1alpha}{b}, respectively.
For the series with $\alpha=0.07$ in \subfigref{fig:sus2dalpha}{b},
the area of the hard phase becomes narrower at both low and high
temperatures as the system size increases.
In detail, the width of the hard phase for $N=256$ is about $\Delta T=1.3$,
while the width shrinks to about $\Delta T=0.3$ at $N=2048$.
Together with the 3D phase diagrams in \subfigdref{fig:susAllalpha}{e}{h},
the phase transitions become sharper in the thermodynamic limit.

As shown in \subfigref{fig:p1alpha}{b}, the first transition point
$p_1$ decreases as the system size increases, even in the low
temperature limit. This is consistent with the first peak of the
convergence time $\tau$ at zero temperature in
\subfigref{fig:talpha}{b}. This indicates that the system becomes
progressively harder to solve in the thermodynamic limit over the
whole temperature range.

\subsubsection{$\chi(T,p_{out})$ at $\alpha=0.15$}

In panels (i) through (l) of \figref{fig:susAllalpha}, $\alpha=0.15$
and the clusters are smaller on average resulting in systems that are
more difficult to solve.
In panels (i) and (j), almost the entire region is covered by small peaks
which indicates mixing of the hard and unsolvable phases thus making the
phase boundaries hard to detect.

The flat easy regions are recognizable in all panels, but the area is
small relative to the previous cases and becomes even smaller as $N$
increases into panel (l).
In panel (i), the flat easy region is roughly triangular with legs
along $T\in (0,1.5)$ and $p_{out}\in(0,0.2)$.
The easy region shrinks to a smaller triangle along $T\in (0,0.2)$
and $p_{out}\in (0,0.2)$ in panel (j) and (k).
In panel (l), it further shrinks to $T\in (0,1)$ and $p_{out}\in(0,0.1)$.
The easy phase shrinks for both $p_{out}$ and $T$ as $N$ increases which
further indicates that the initial transition point $p_1$ decreases
substantially in the thermodynamic limit.

The corresponding plots of the hard phase boundaries and the first
transition points $p_1$ are displayed in Figs.\
\ref{fig:sus2dalpha}(c) and \ref{fig:p1alpha}(c), respectively. From
\subfigref{fig:sus2dalpha}{c}, the area of the hard phase shrinks in
the thermodynamic limit. The hard phase is more identifiable
relative to the unsolvable region as $N$ increases. The initial
transition point $p_1$ drops as $N$ increases as shown in
\subfigref{fig:p1alpha}{c}. The convergence time $\tau$ for the
systems with the fixed ratio of $\alpha=q/N=0.15$ at zero
temperature is shown in \subfigref{fig:talpha}{c} where the first
peak of $\tau$ shifts to the left as the system size increases. This
is consistent with the trend observed in \subfigref{fig:p1alpha}{c}.
We further show in \figref{fig:comparison2048q140} and
\figref{fig:comparison1024q70} that the first transition points in
``computational susceptibility'', energy, entropy, convergence time
and normalized mutual information are consistent with each other.

In \figref{fig:correlation}, we provide plots of scaled waiting
correlation function data which clearly indicate {\it spin glass
type collapse}. The collapse is best at the center of the
computational susceptibility ridge \figref{fig:correlation}(b). The
collapse persists up to the ends of the susceptibility ridge (e.g.,
$p_{out} = p_{1}$ in \figref{fig:correlation}(a)) and is no longer
valid outside the susceptibility ridge (e.g., $p_{out}=0.26 >
p_{2}=0.24$ in \figref{fig:correlation}(c)).

% --- begin 3D other functions plots ---------------------------------
\begin{figure*}[t!]
\begin{center}
\subfigure[\ $\chi(T,p_{out})$ for $q=16$]{\includegraphics[width=\subfigwidth]{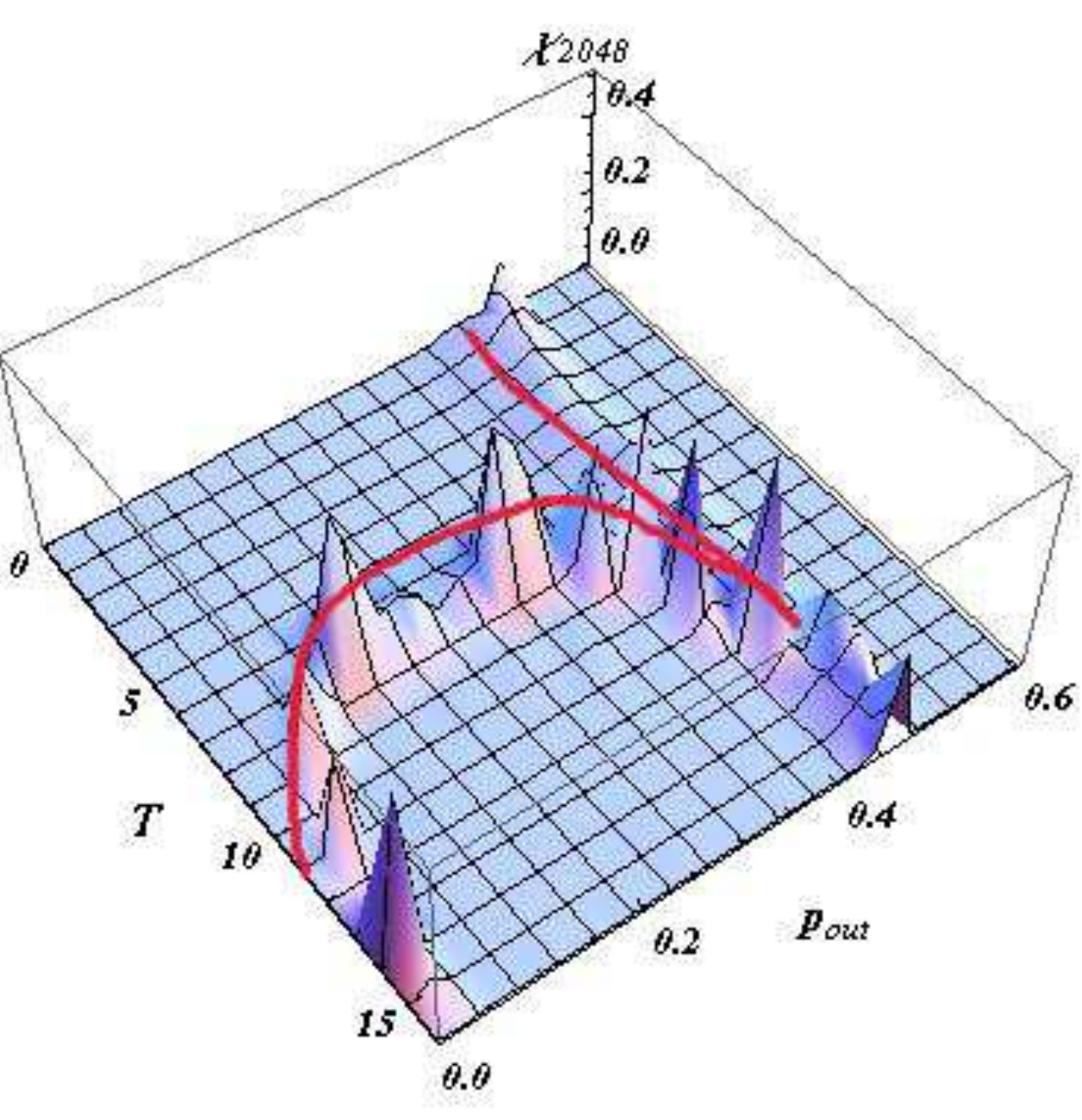}}
\subfigure[\ $I_N(T,p_{out})$ for $q=16$]{\includegraphics[width=\subfigwidth]{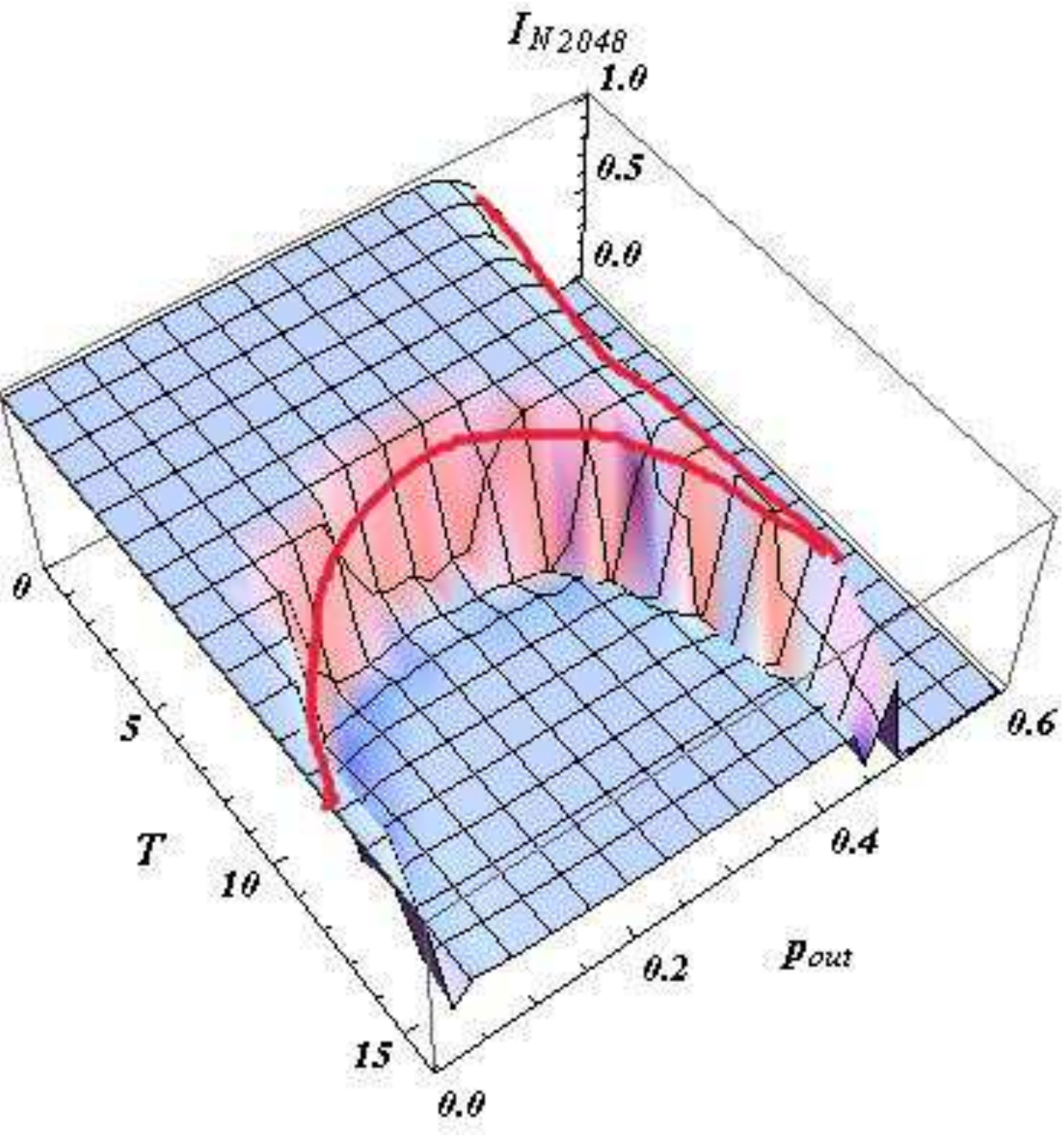}}
\subfigure[\ $H(T,p_{out})$ for $q=16$]{\includegraphics[width=\subfigwidth]{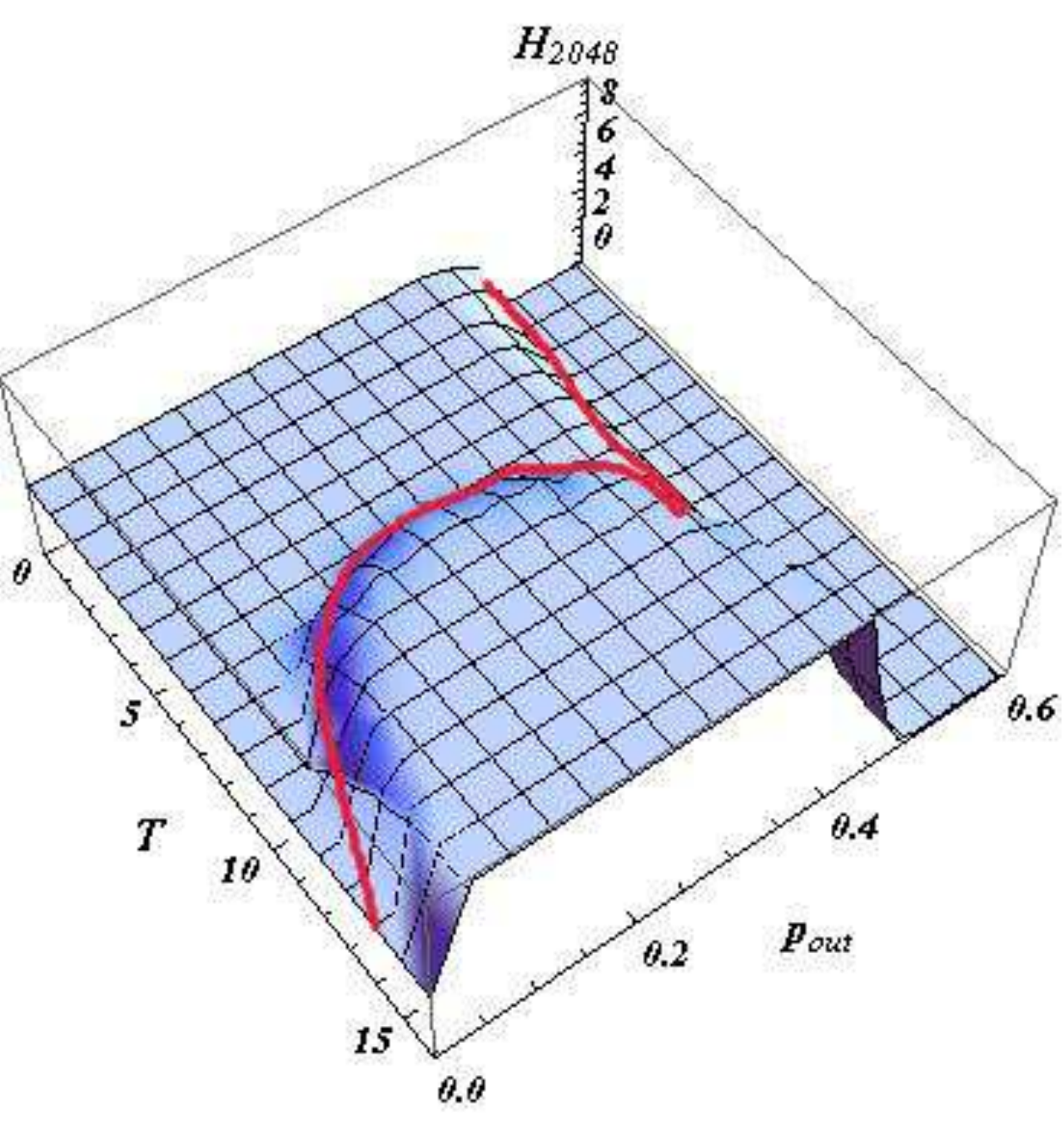}}
\subfigure[\ $E(T,p_{out})$ for $q=16$]{\includegraphics[width=\subfigwidth]{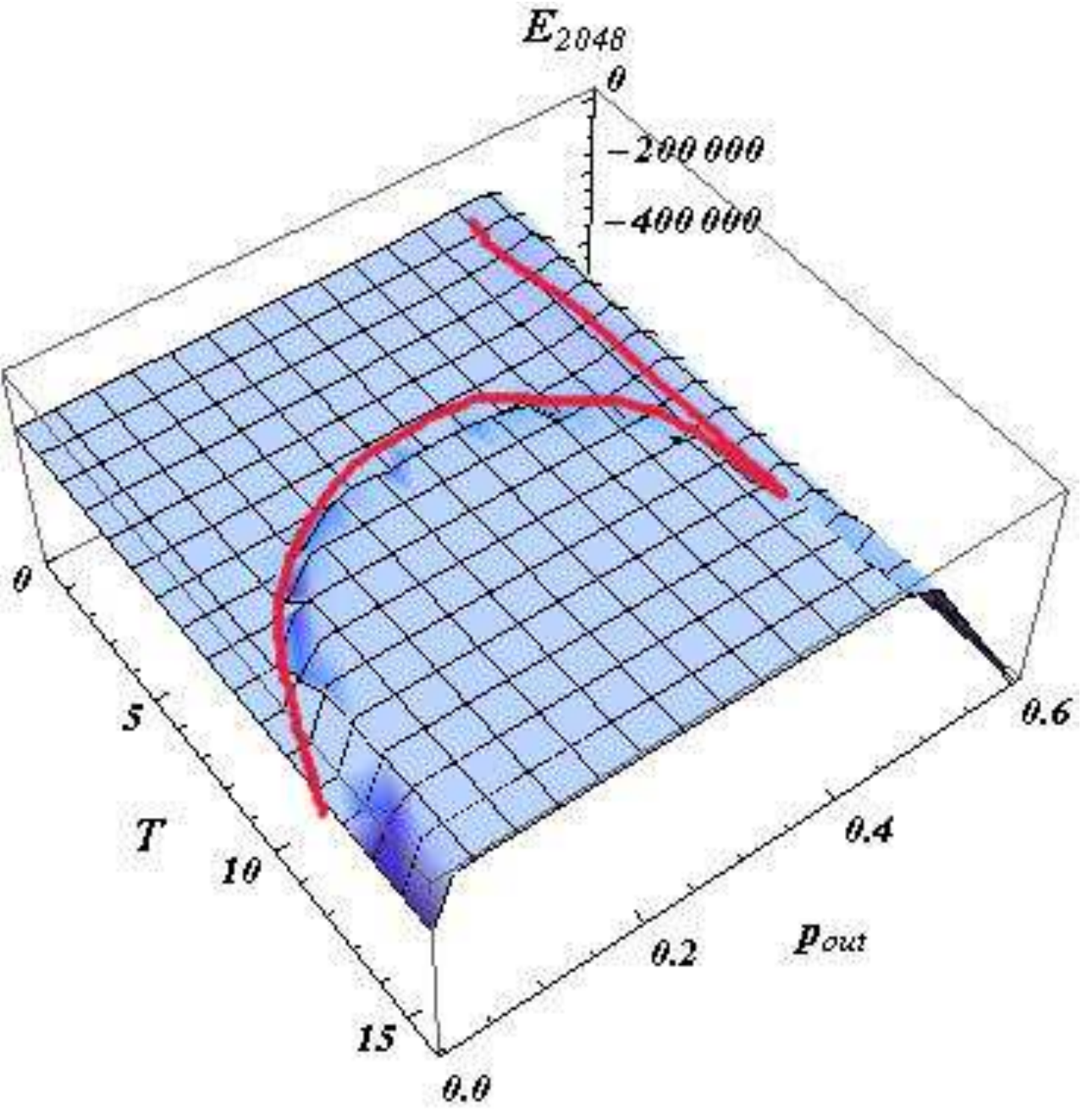}}
\subfigure[\ $\chi(T,p_{out})$ $q=32$]{\includegraphics[width=\subfigwidth]{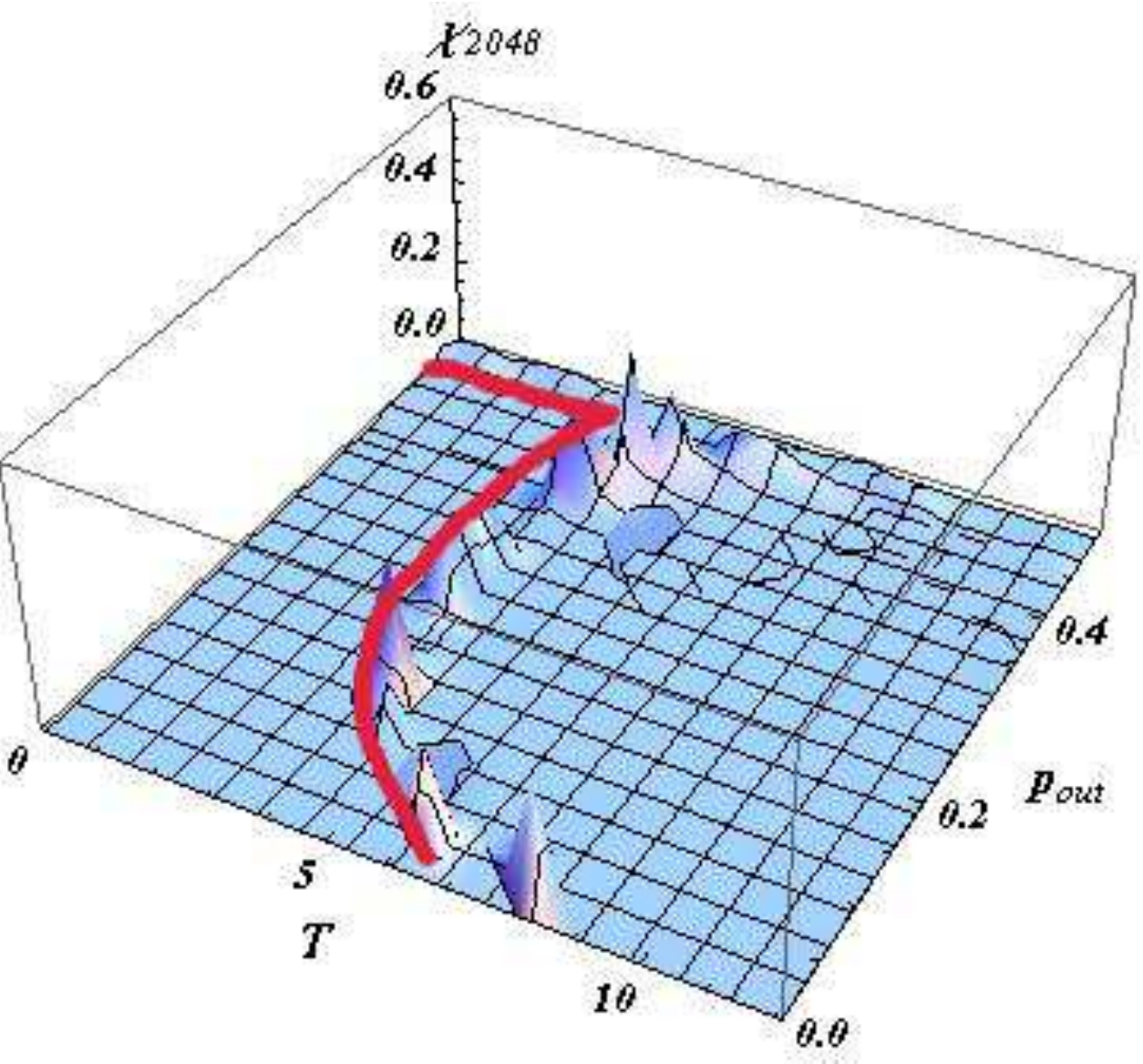}}
\subfigure[\ $I_N(T,p_{out})$ for $q=32$]{\includegraphics[width=\subfigwidth]{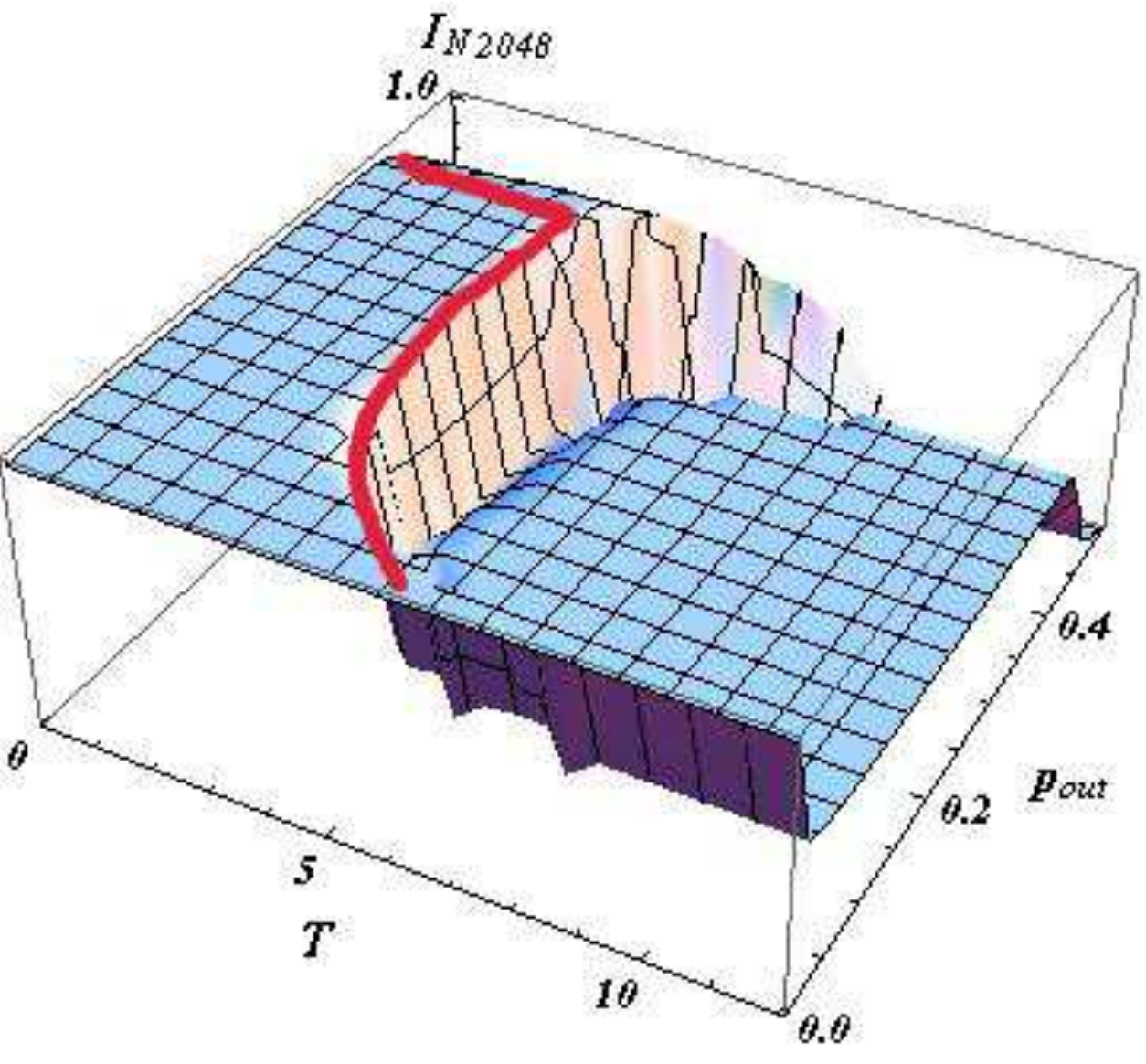}}
\subfigure[\ $H(T,p_{out})$ for $q=32$]{\includegraphics[width=\subfigwidth]{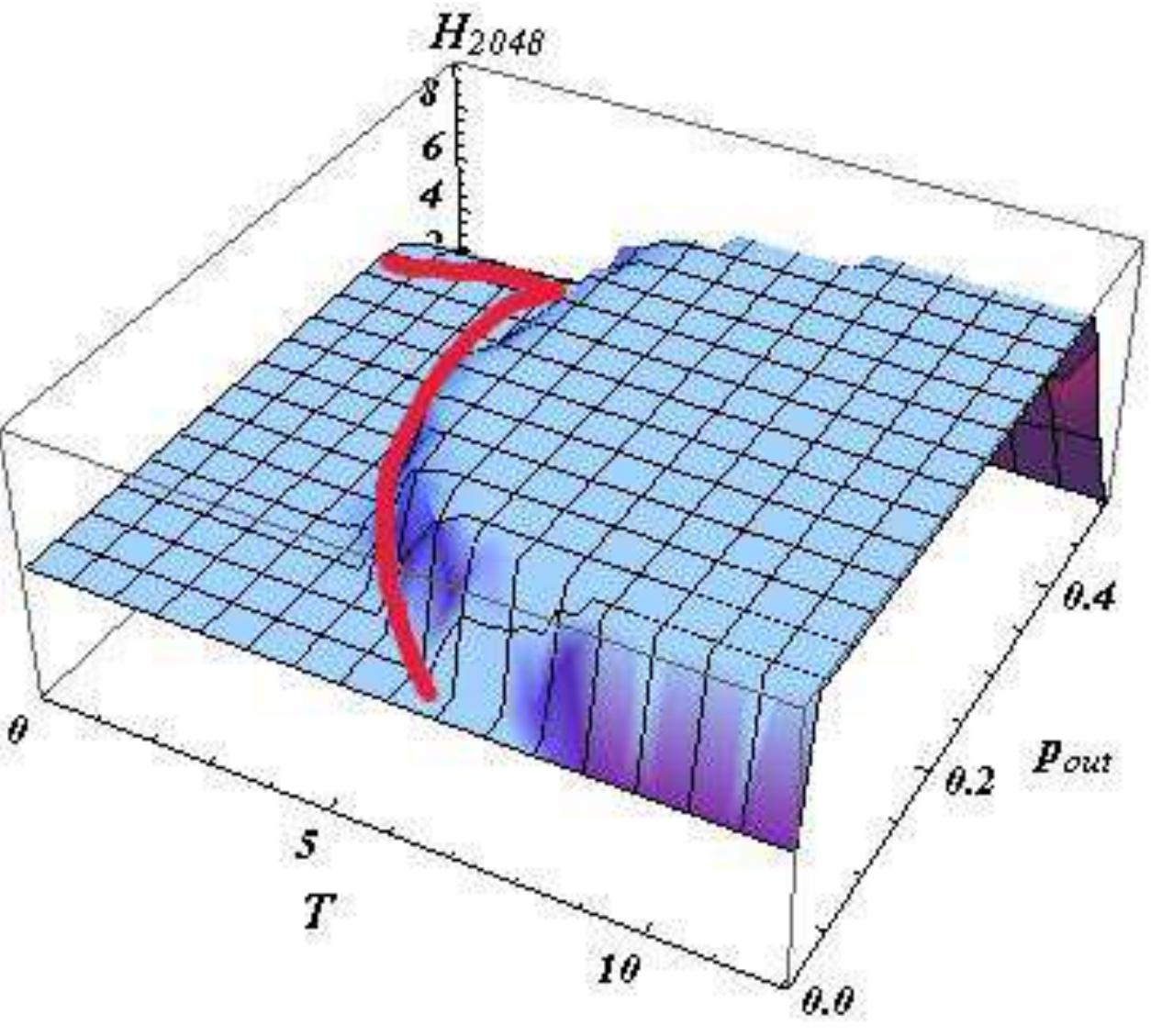}}
\subfigure[\ $E(T,p_{out})$ for $q=32$]{\includegraphics[width=\subfigwidth]{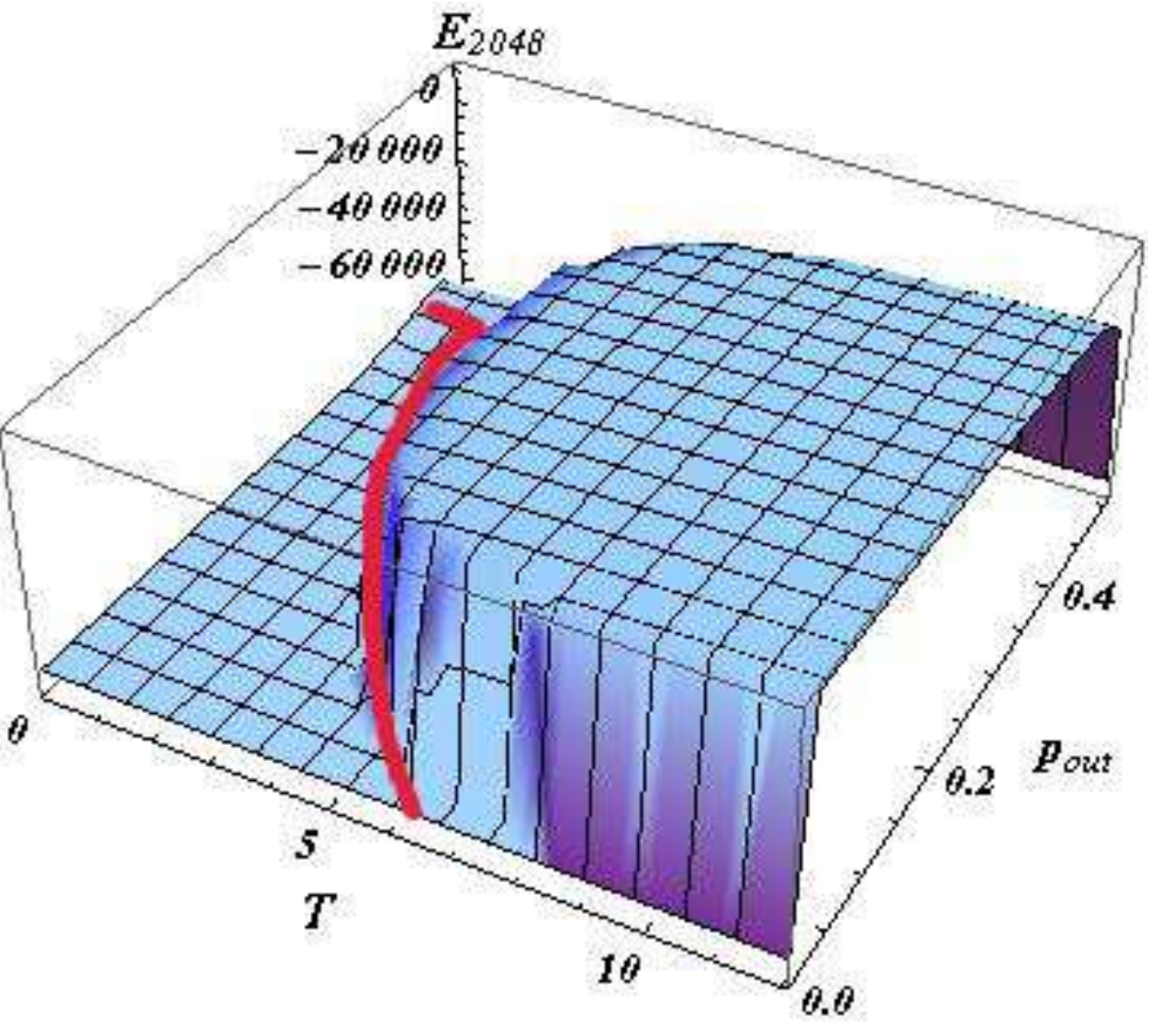}}
\subfigure[\ $\chi(T,p_{out})$ for $q=70$]{\includegraphics[width=\subfigwidth]{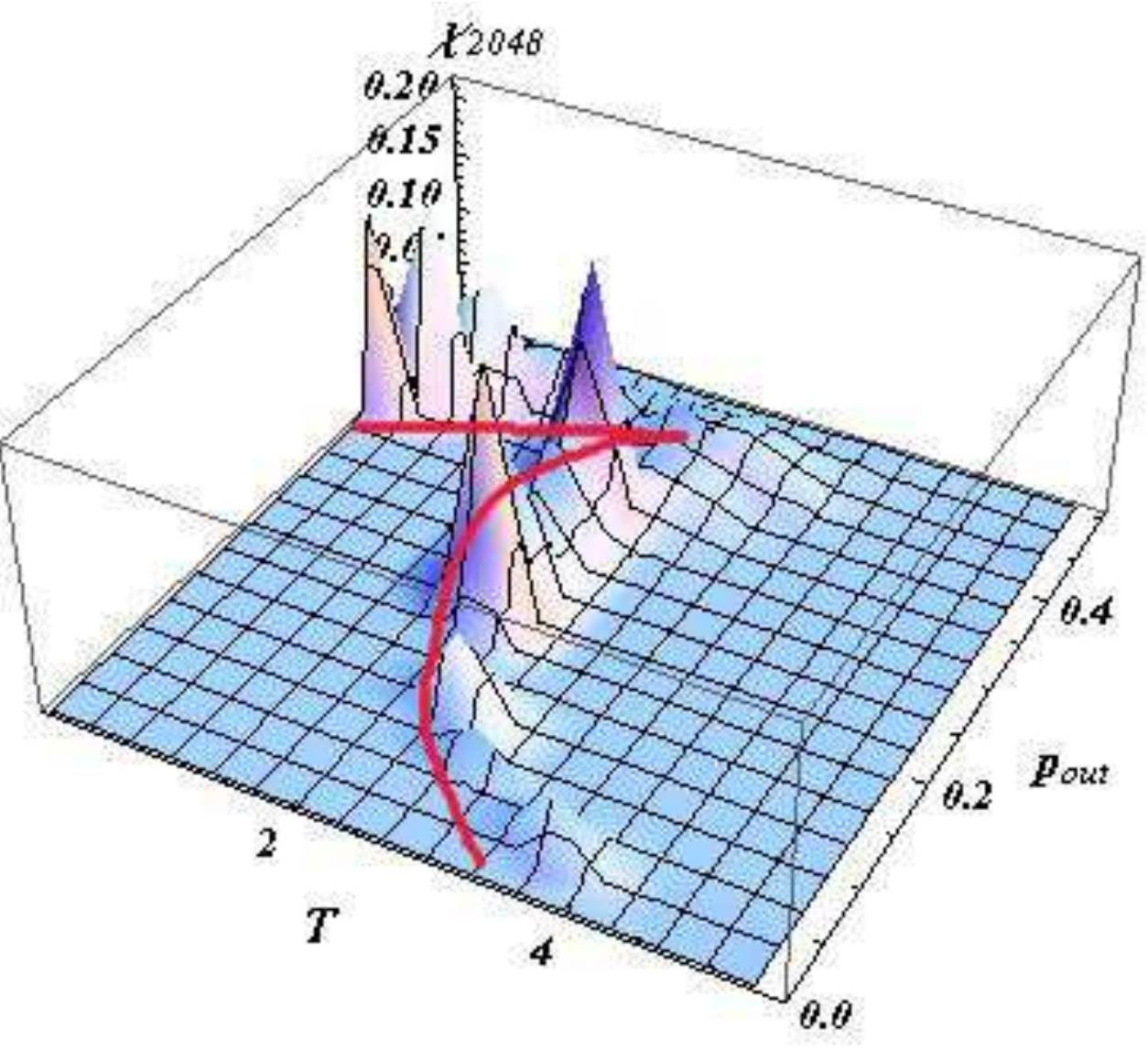}}
\subfigure[\ $I_N(T,p_{out})$ for $q=70$]{\includegraphics[width=\subfigwidth]{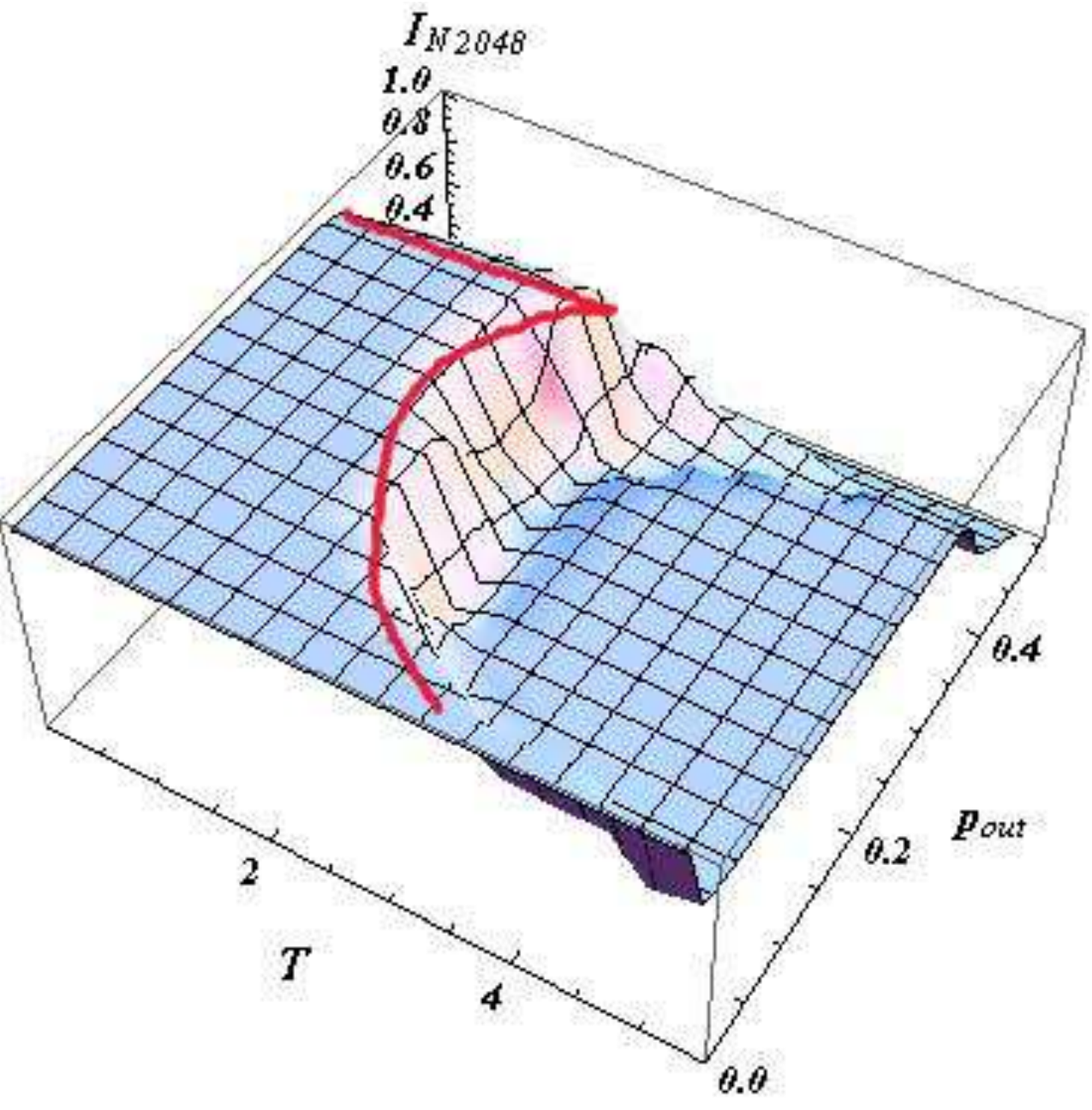}}
\subfigure[\ $H(T,p_{out})$ for $q=70$]{\includegraphics[width=\subfigwidth]{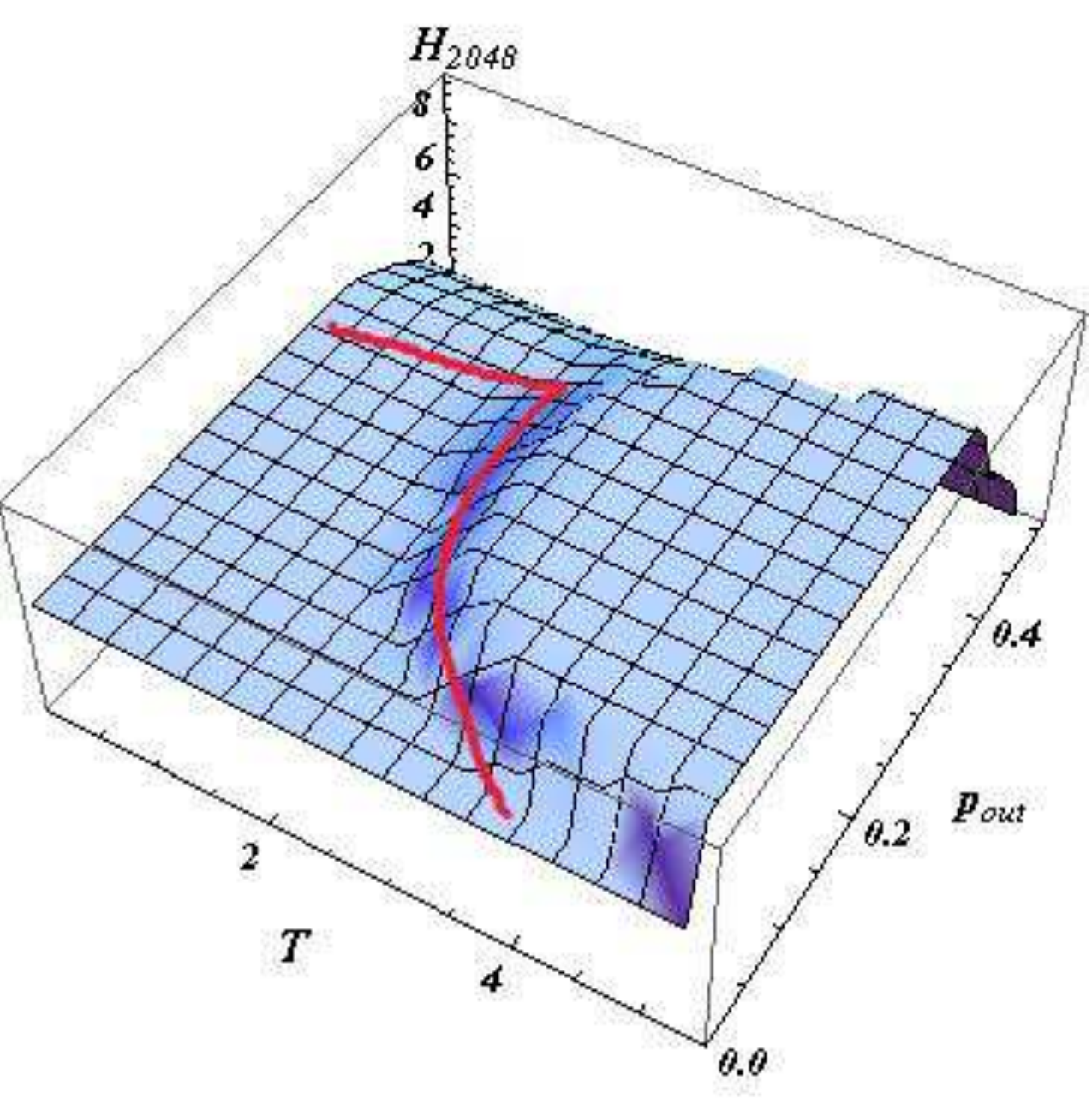}}
\subfigure[\ $E(T,p_{out})$ for $q=70$]{\includegraphics[width=\subfigwidth]{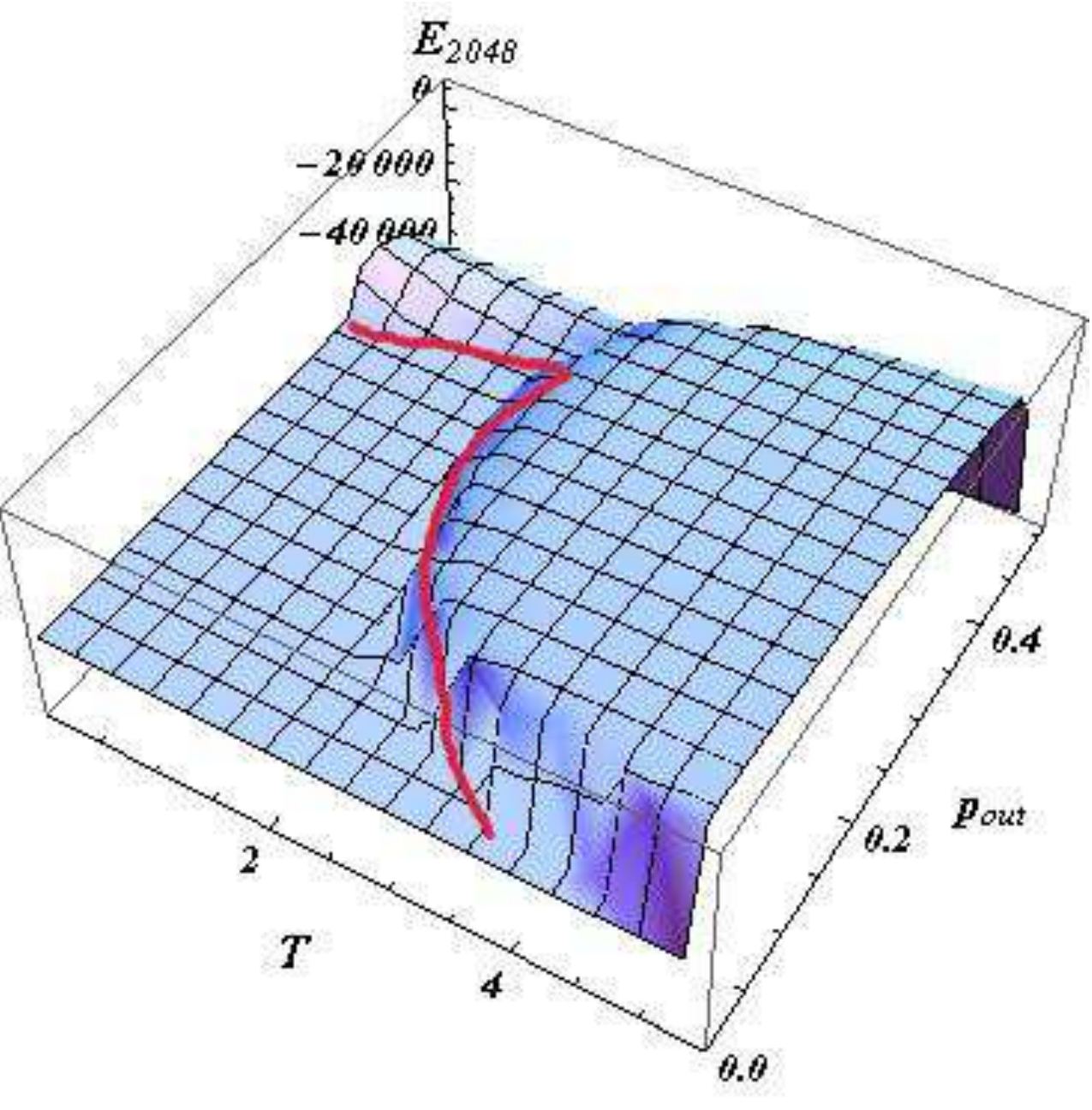}}
\subfigure[\ $\chi(T,p_{out})$ for  $q=140$]{\includegraphics[width=\subfigwidth]{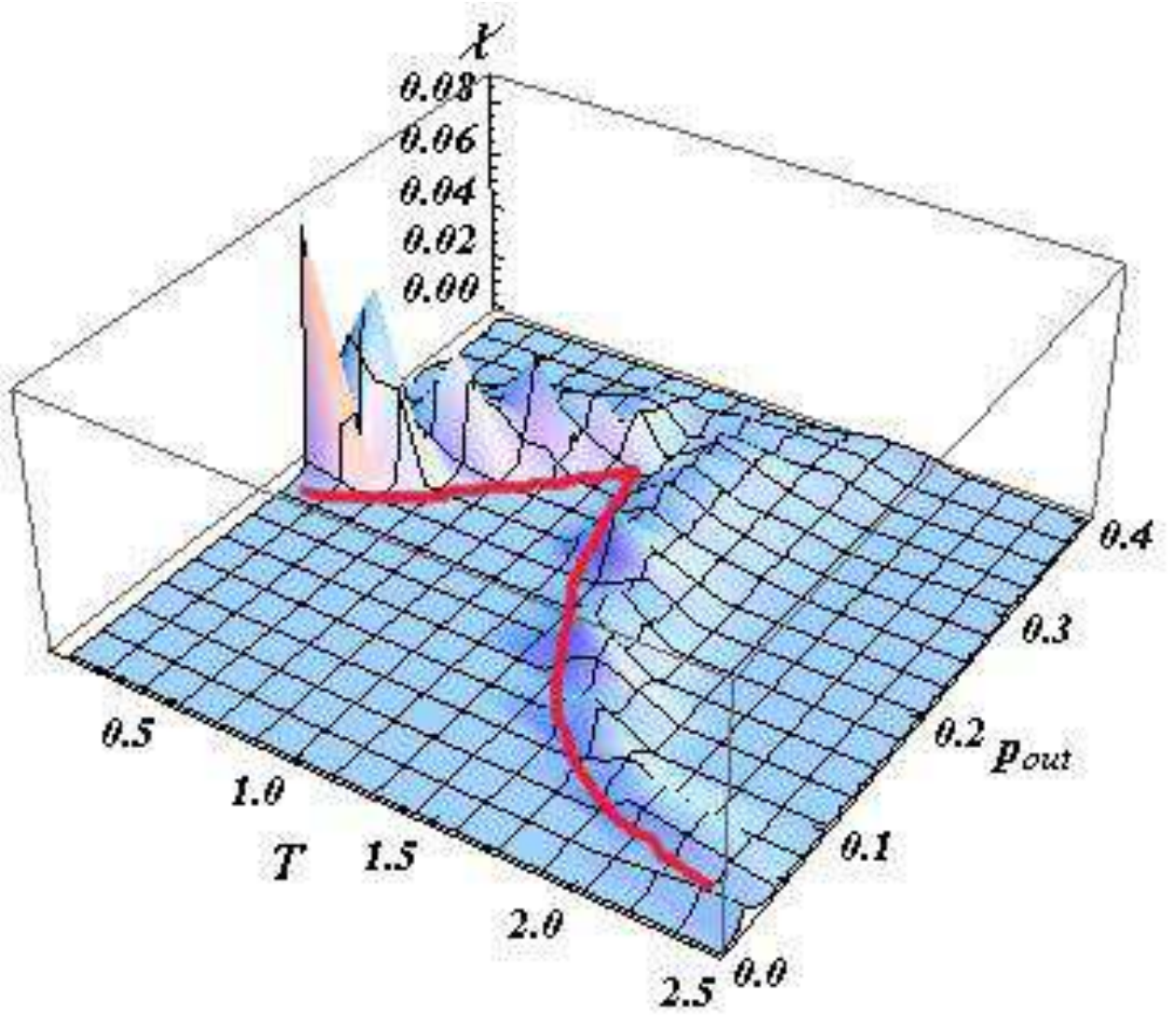}}
\subfigure[\ $I_N(T,p_{out})$ for $q=140$]{\includegraphics[width=\subfigwidth]{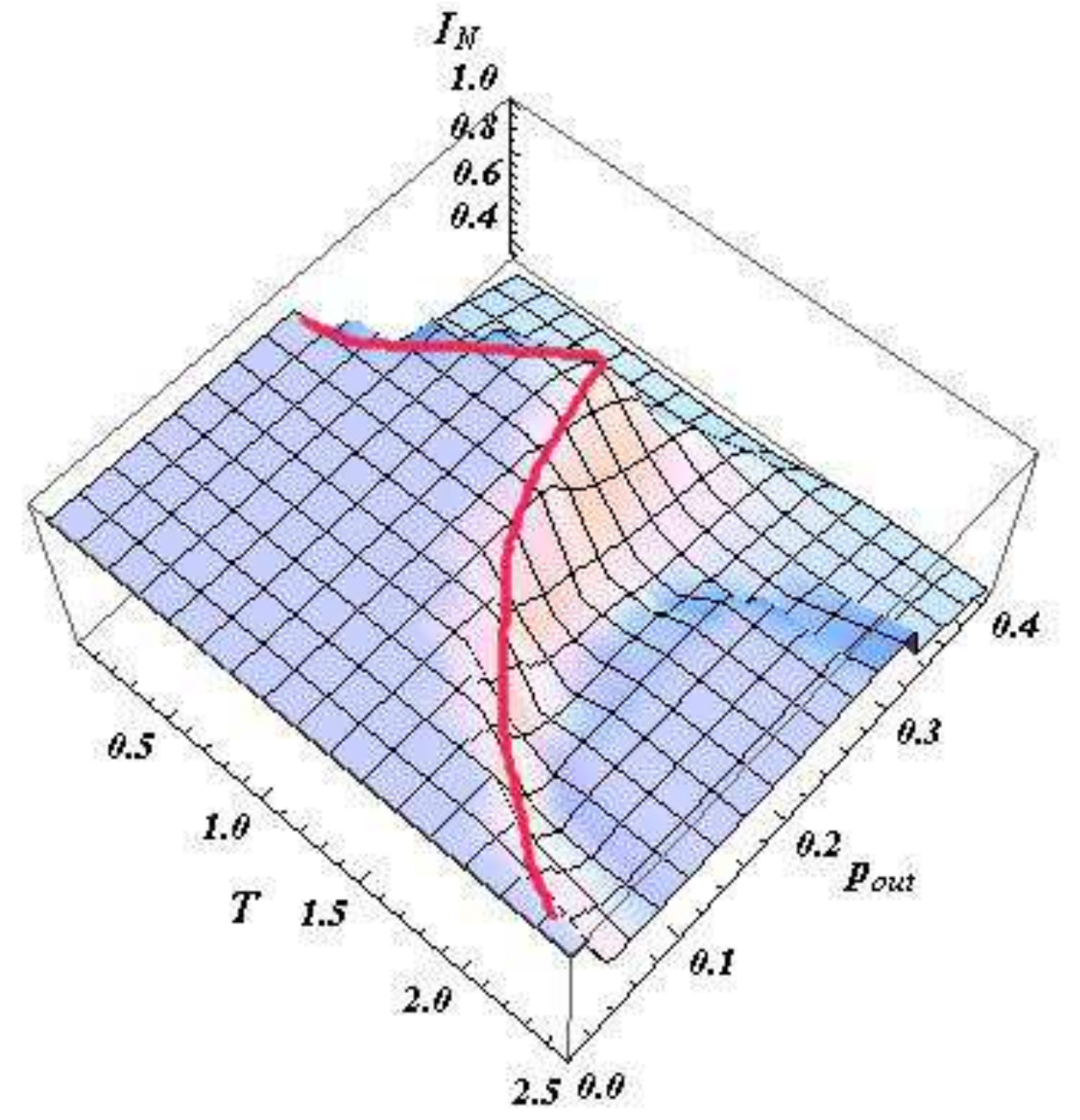}}
\subfigure[\ $H(T,p_{out})$ for $q=140$]{\includegraphics[width=\subfigwidth]{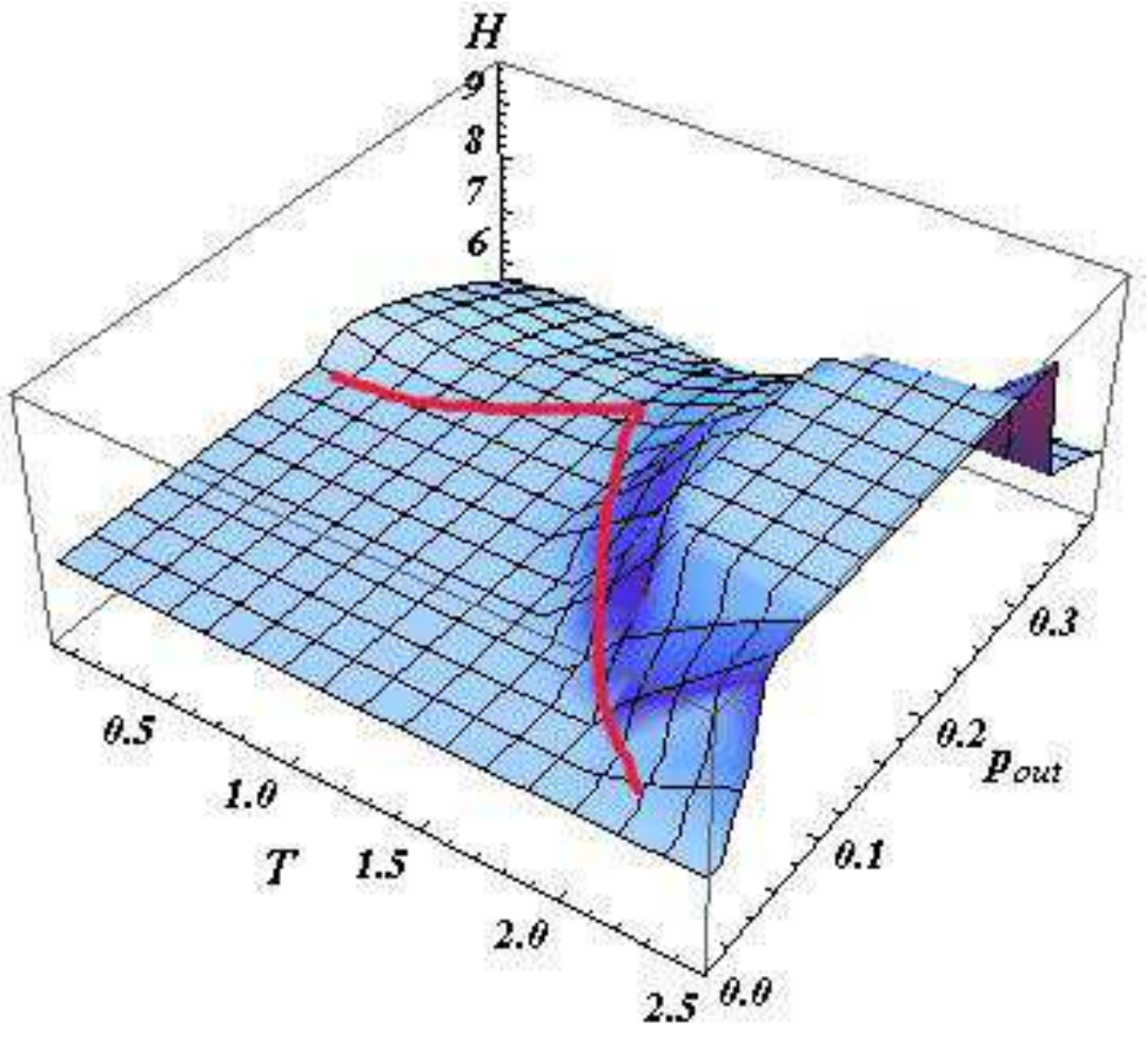}}
\subfigure[\ $E(T,p_{out})$ for $q=140$]{\includegraphics[width=\subfigwidth]{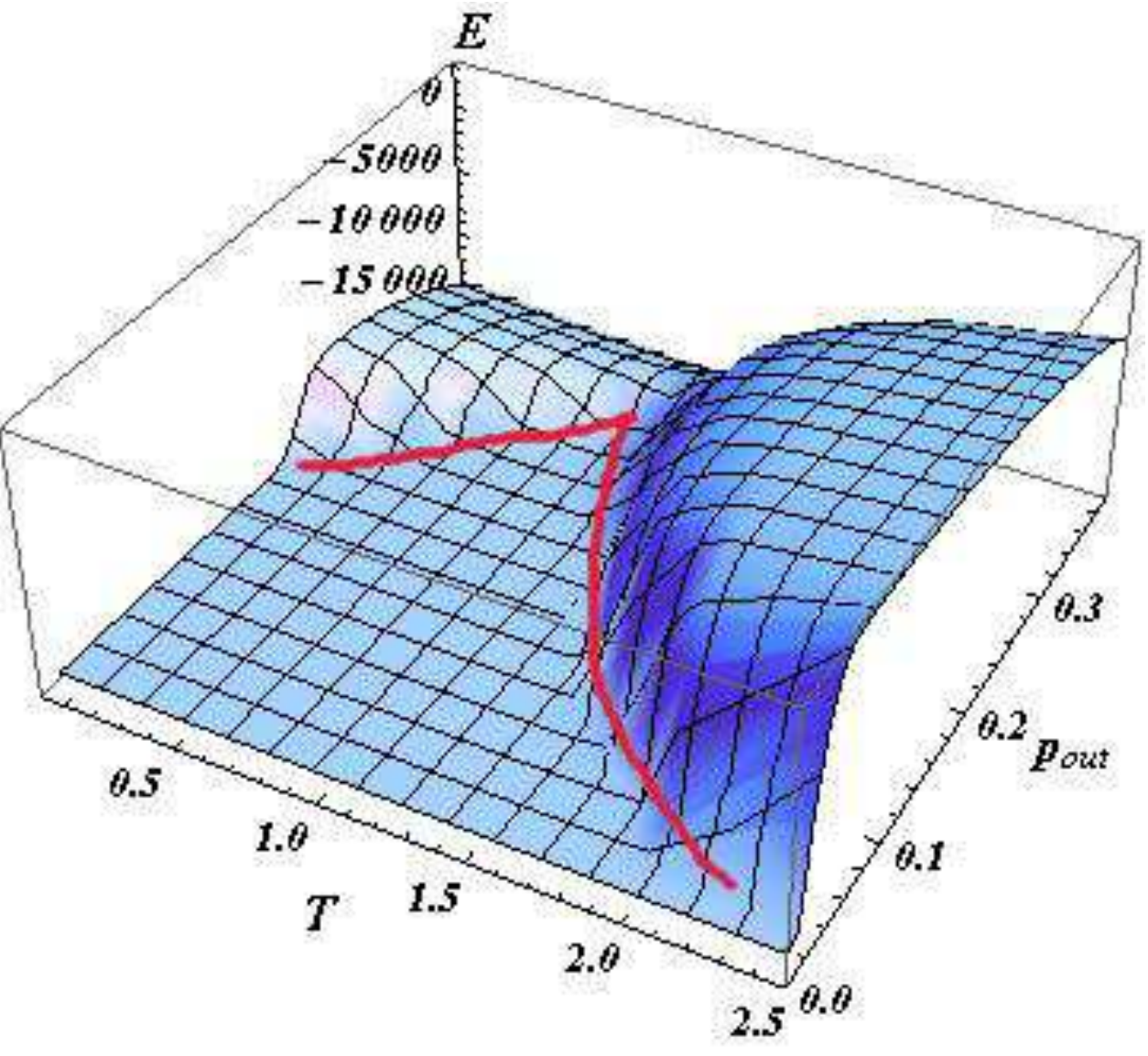}}
\end{center}
\caption{(Color online) Plots of the computational susceptibility $\chi$ (column one),
NMI $I_N$ (column two), Shannon entropy $H$ (column three), and energy $E$
(column four) as functions of temperature $T$ and intercommunity noise $p_{out}$.
System sizes all use $N=2048$, and $q$ varies from $16$ to $140$ in different rows.
All plots show the easy, hard, and unsolvable phases often by rapid shifts
in the respective measures.
The red lines serve as a guide to the eye for emphasizing the manifestation
of the hard phases in each measured quantity where we note that the boundaries
match well across each row.}
\label{fig:EHqAll}
\end{figure*}
% --- end 3D other functions plots -----------------------------------

\begin{figure}
\centering
\includegraphics[width=0.9\columnwidth]{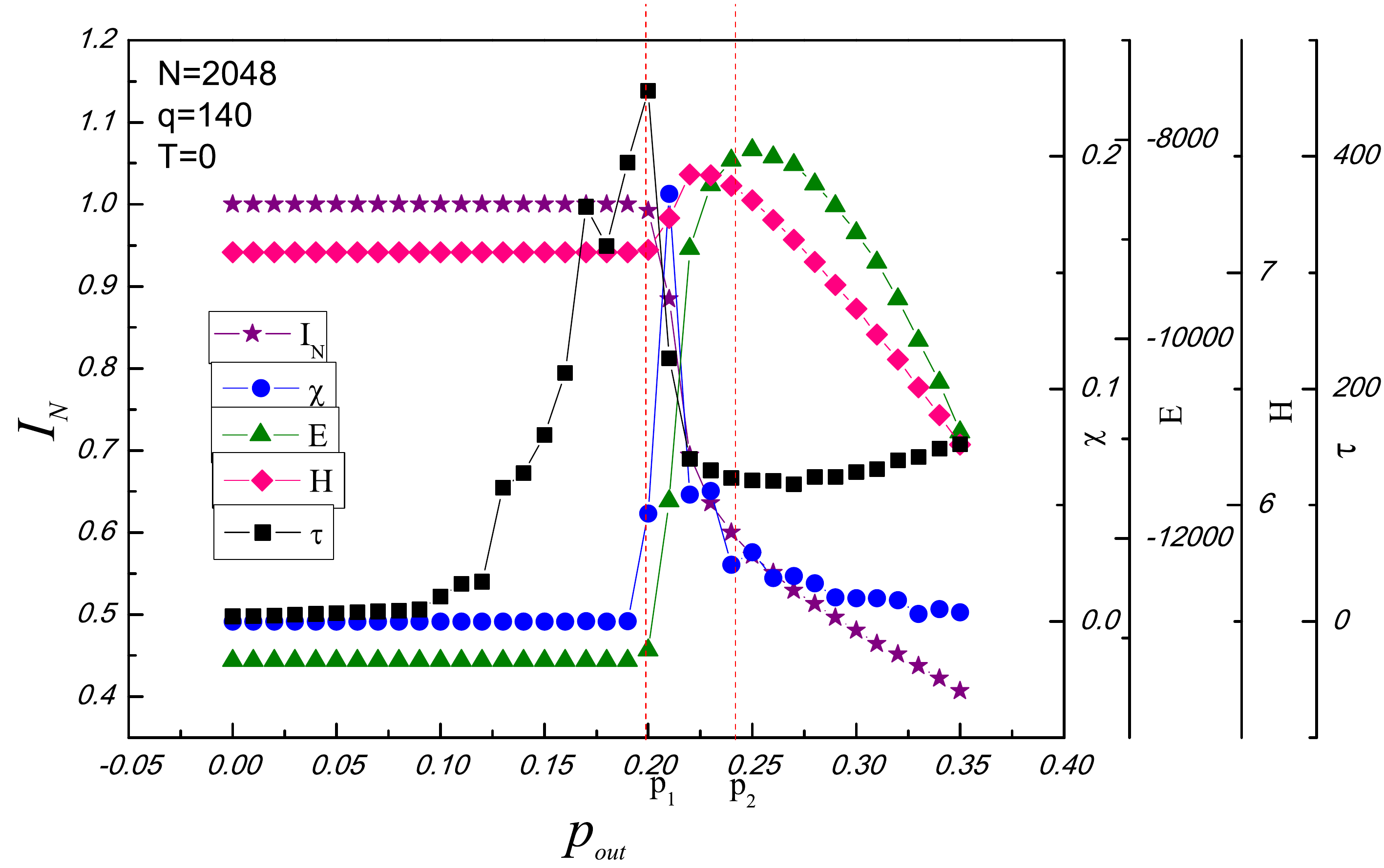}
\caption{(Color online) The plots of susceptibility $\chi$,
convergence time $\tau$, energy $E$, accuracy $I_N$ and the Shannon
entropy $H$ in terms of noise $p_{out}$ for the system $N=2048$ and
$q=140$ at a zero temperature. All the plots show three phases as
noise varies: (1) Below $p_{1}=0.2$, the system can be solved in
this ``easy'' region ( e.g., the accuracy is $I_N=1$ ); (2) When
$0.2<p_{out}<0.24$, where the benefit of extra trials is the
largest, it's ``hard'' to solve the system without misplacing nodes
(e.g., $\chi$, $E$ and $H$ achieve the peak ) ; (3) Above
$p_{2}=0.24$, the system is ``impossible'' to be perfectly solved.
$[p_1,p_2]$ are generous bounds in transition crossover regions.
Note that the two transitions are demonstrated to be of spin-glass-type
by observing the scaling of the correlation function between $[p_1,p_2]$
in \figref{fig:correlation}.}
\label{fig:comparison2048q140}
\end{figure}

\begin{figure*}[t]
\begin{center}
\subfigure[\ $\chi$, $\tau$, $E$, $H$, $I_N$ for the system of size
$N=1024$]{\includegraphics[width=3.5in]{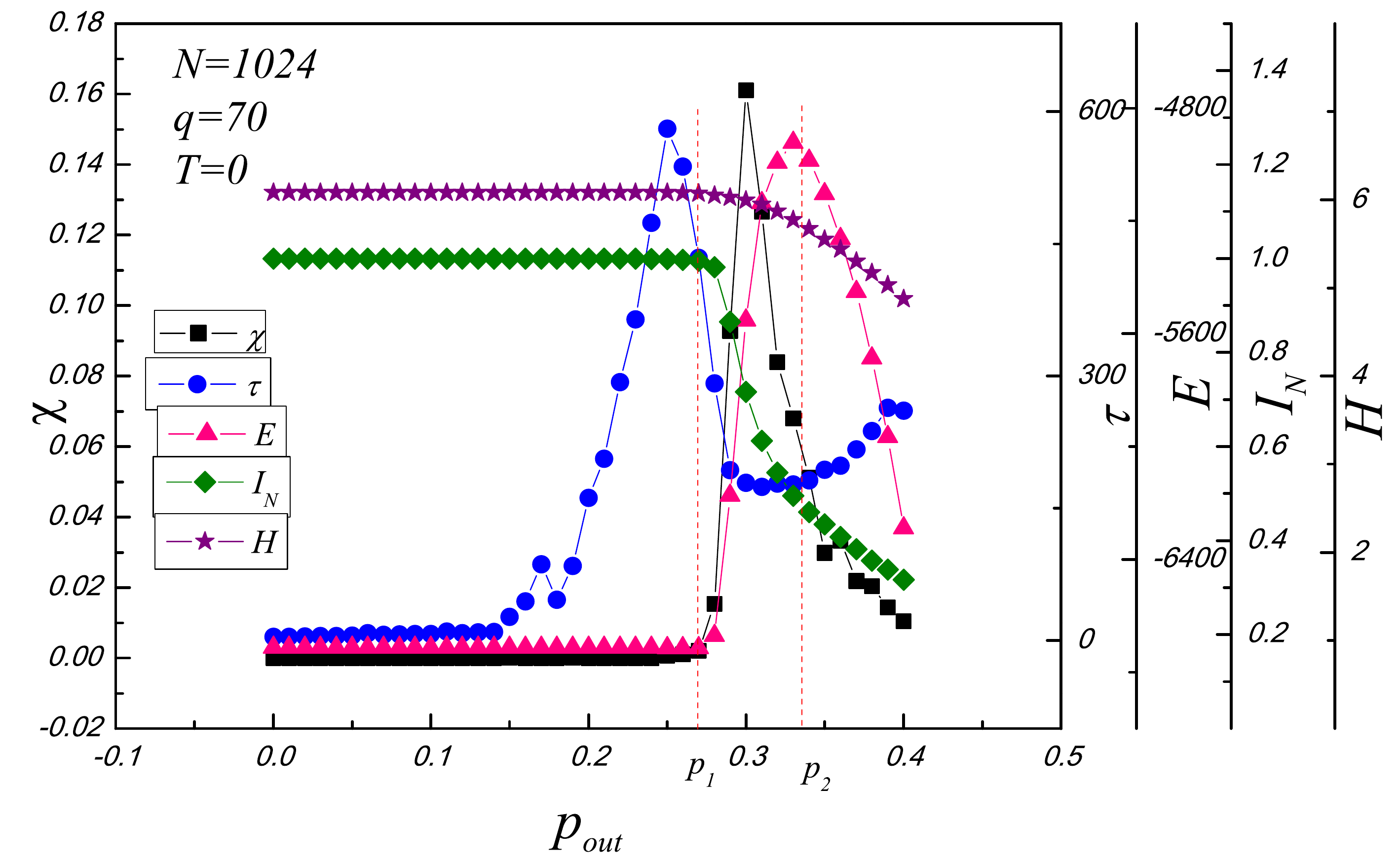}}
\subfigure[\ $I_N$ for the system of $N=1024$ with different
$q$]{\includegraphics[width=3in]{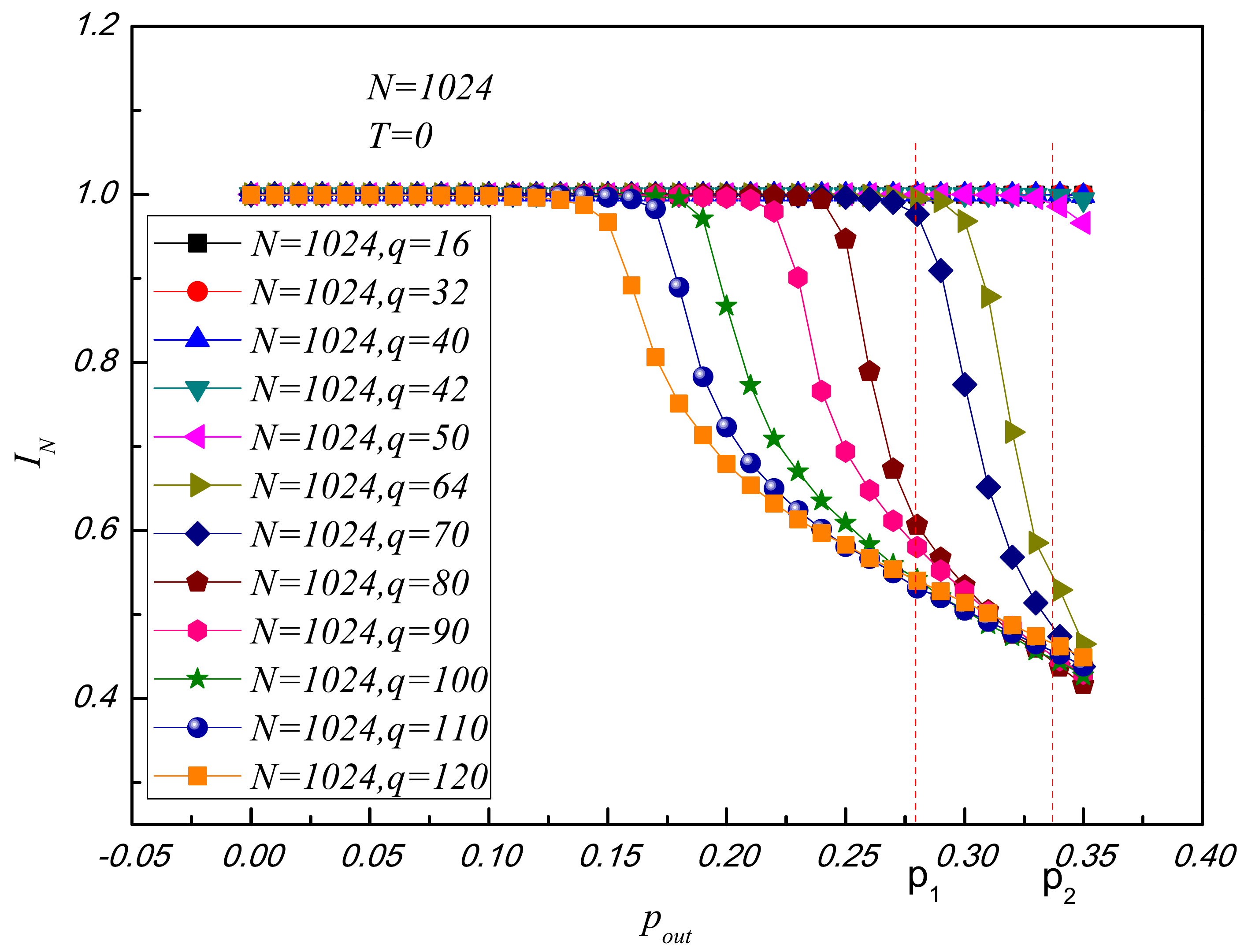}}
\end{center}
\caption{(Color online) (a), The plots of susceptibility $\chi$,
convergence time $\tau$, energy $E$, accuracy $I_N$, Shannon entropy
$H$ in terms of noise $p_{out}$ for the system $N=1024$ and $q=70$
at a zero temperature. (b), The normalized mutual information $I_N$
in terms of noise $p_{out}$ for a series of systems with the size of
$N=1024$ but different number of communities $q$. From both plots,
we are able to detect the first and second transition point $p_1$
and $p_2$. $p_1$ is the point where the $I_N$ drops from $1$, $\chi$
increases from $0$, $\tau$ achieves the peak, $E$ and $H$ increases
from some constant value. $p_2$ is the position where the $I_N$
curves with different number of communities collapse shown in (b).
$p_2$ also corresponds to the peak of energy and entropy as shown in
(a).} \label{fig:comparison1024q70}
\end{figure*}

\begin{figure*}[t]
\begin{center}
\subfigure[\ $p_{out}=0.2$ is within the zero temperature ``hard''
phase, where the collapse is
perfect.]{\includegraphics[width=2in]{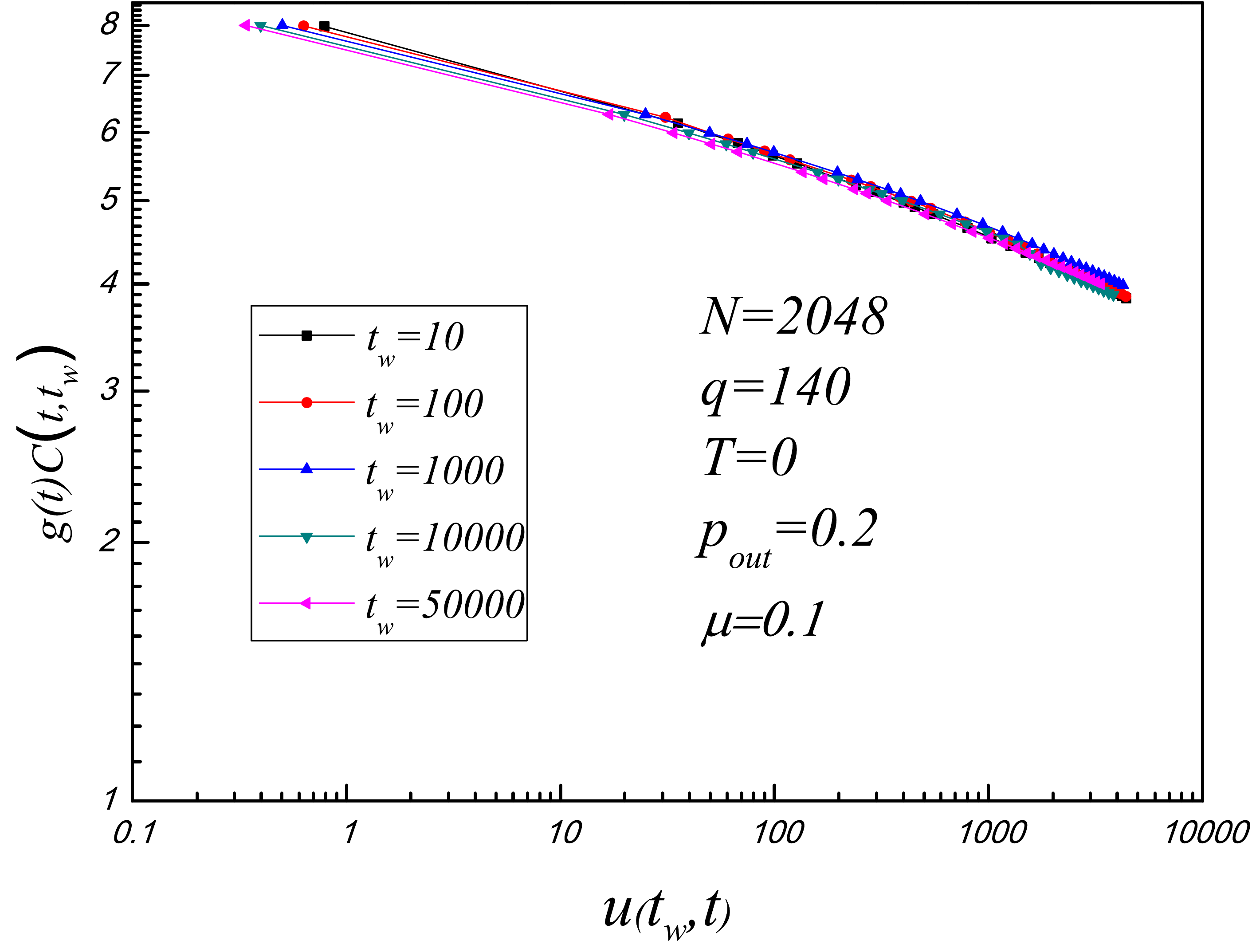}}
\subfigure[\ $p_{out}=0.22$ is within the zero temperature ``hard''
phase, where the collapse is
perfect.]{\includegraphics[width=2in]{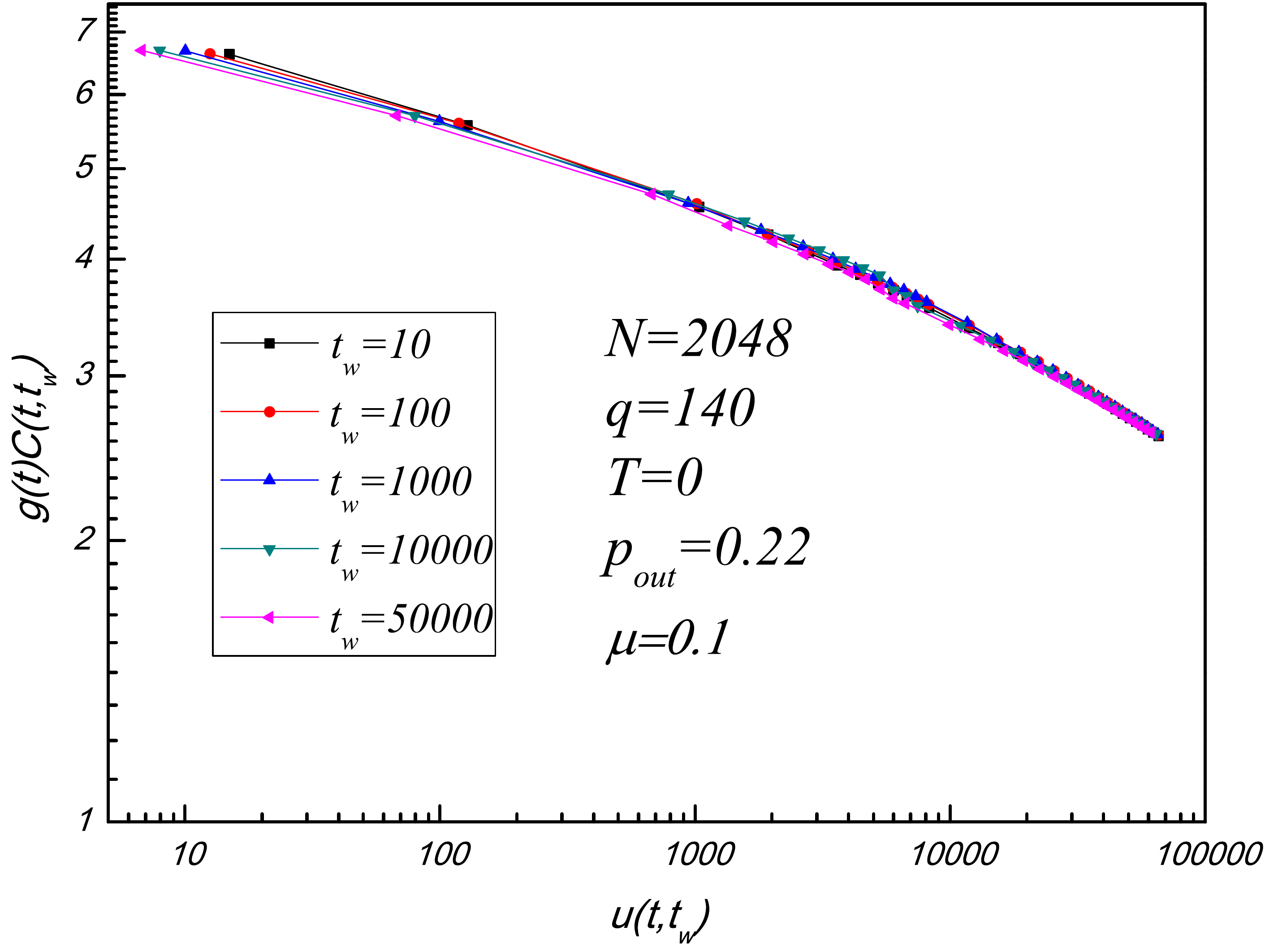}}
\subfigure[\ $p_{out}=0.26$ is within the zero temperature
``unsolvable'' phase, where the collapse is
poor.]{\includegraphics[width=2in]{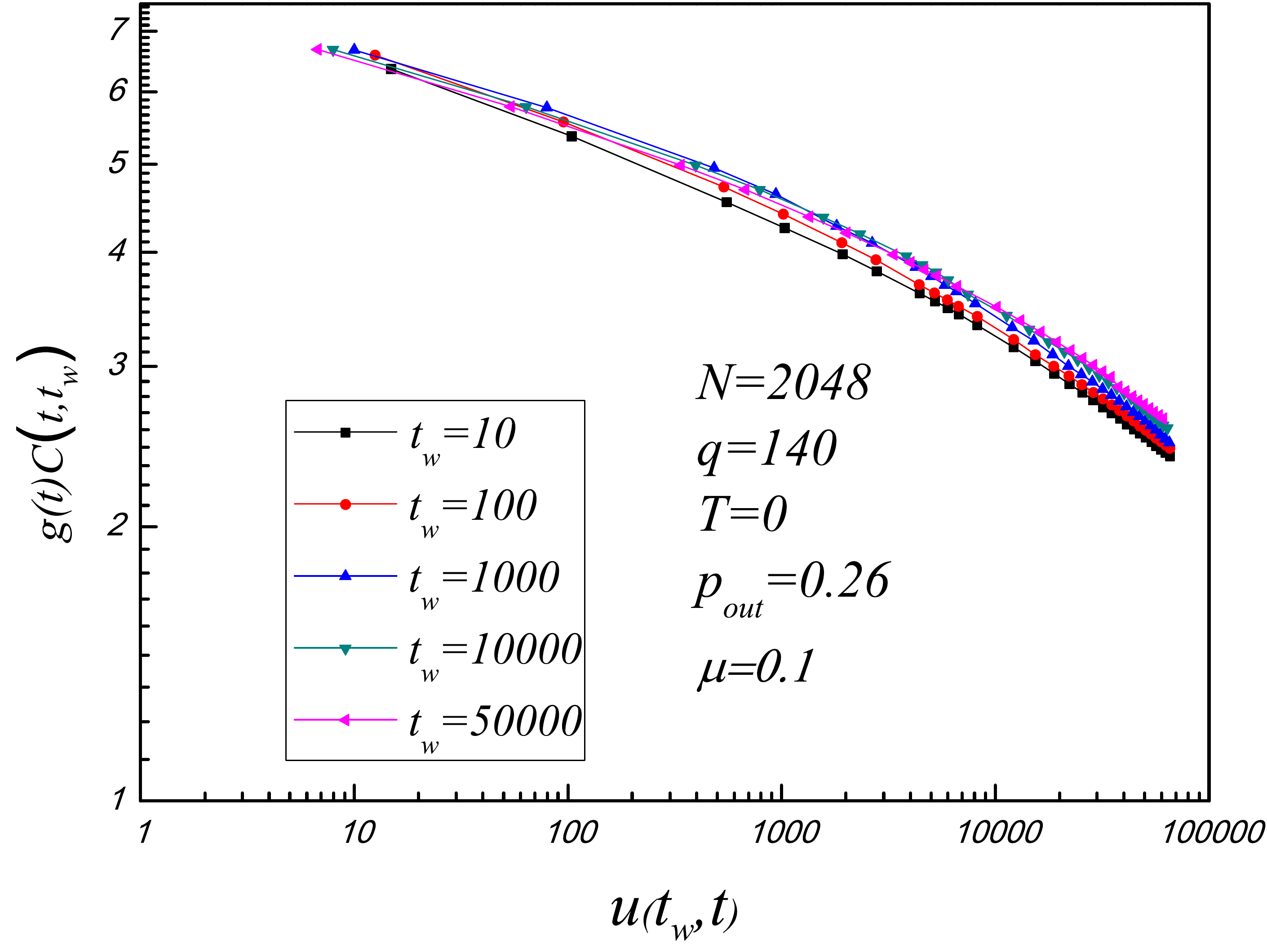}}
\end{center}
\caption{(Color online). We show a collapse of the correlation
curves for different waiting times $t_w$ for a system with $N=2048$
nodes, $q=140$ communities. $p_{out}$ varies from $0.2$ in panel (a)
to $0.26$ in panel (c). The first and second transition points for
this system are $p_1=0.2$ and $p_2=0.24$. The temperature is $T=0$.
The vertical axis is $g(t)C(t_w,t)$ where $g(t)=8-\log_{10}(t)$,
$C(t_w,t)=\frac{1}{N}\sum_{i=1}^N\delta_{\sigma_i(t_w),\sigma_i(t_w+t)}$
is the correlation function. The horizontal axis is
$u(t_w,t)=\frac{1}{1-\mu}[(t+t_w)^{1-\mu}-t_w^{1-\mu}]$ where
$\mu=0.1$. The noise $p_{out}=0.2$ (a) and $p_{out}=0.22$ (b) lie
within the ``hard'' region where the collapse of correlation
function is perfect. The noise $p_{out}=0.26$ (c) is above the
second transition point $p_2$ in the ``unsolvable'' region, where
the collapse becomes poor. That the collapse of the correlation
function starts to degrade right after the second transition point
$p_2$ at zero temperature indicates that this transition is of the
spin-glass type.}
\label{fig:correlation}
\end{figure*}

\subsection{$\chi(T,p_{out})$ at fixed $q$} \label{sec:sec2}

%Now we keep the number of communities $q$ fixed and increase the number
%of nodes $N$ to examine phase transitions in another series of systems.

%\subsection{$\chi(T,p_{out})$ at $q=16$}

We fix the number of communities at $q=16$, $40$, or $70$ and increase the
system size $N$ from $256$ to $2048$.
The plots of computational susceptibility $\chi(T,p_{out})$ for $q=16$
series of systems are shown in panels (a) through (d) of \figref{fig:susAllq}.
As in \secref{sec:sec1}, the ridges indicate hard phases which become more
prominent as $N$ increases while the ridges at low temperature remain at relatively
low constant values.

The areas of the easy phases on the lower left corner expand as the system size
increases from panel (a) to (d).
This trend of increasing area is the reverse of the behavior in the fixed $\alpha$
systems systems in \secref{sec:sec1}.
This is easy to understand since, $q$ increases with $N$ here, and the high internal
edge density $p_{in}$ causes the larger clusters to be more strongly defined.

%\subsection{$\chi(T,p_{out})$ at $q=40$}

We increase the number of communities to $q=40$ for the systems in panels
(e) through (h).
$N$ varies from $256$ to $2048$, and $\alpha=q/N$ decreases as $N$ increases
so that the systems again become less complicated because the communities
become more strongly defined.
The hard and unsolvable phases in the small $N=256$ system in panel (e) are
difficult to distinguish.
Only the easy phase can be easily identified by noting the flat region
on the lower left of each panel.
$\chi(T,p_{out})$ peaks at increasing heights at both the low and high temperatures
from panels (f) to (h) indicating that the phase transitions become more prominent
as the system size increases.

%\subsection{$\chi(T,p_{out})$ at $q=70$}

We further increase the number of communities to $q=70$ and study the
phase transitions for the same range of system sizes.
The hard phase at high temperature in panel (i) is difficult to detect.
$\chi(T,p_{out})$ clearly shows the three phases in panels (j) and (k).
The easy phases again become larger as the system size increases.
$\chi(T,p_{out})$ in the hard phase increases as $N$ increases indicating
that the phase transitions at both low and high temperatures are more obvious
from panel (i) to (k).

In \figsref{fig:sus2dq}{fig:p1q}, we also show corresponding 2D
plots for the boundaries of the hard phase and the first transition point
$p_1$ as the function of temperature $T$.
In \figref{fig:sus2dq}, the area of the hard phase becomes narrower
as the system size increases.
At $q=40$, for example, the width of the hard phase for the smallest
system at $N=256$ is about $\Delta T=1.5$.
As $N$ increases, the hard phase width shrinks to $\Delta T=1$ at $N=512$
and down to $\Delta T=0.5$ for $N=2048$ which further indicates that the
phase transition becomes sharper in the thermodynamic limit.
In \figref{fig:p1q}, the first transition point $p_1$ increases over
the entire temperature range as $N$ increases.
This behavior is consistent with the system complexity trend as previously
mentioned.

In \figref{fig:tq}, we further plot the convergence time $\tau$ as
the function of noise $p_{out}$ for a fixed number of communities
$q$ at zero temperature. $p_{out}$ for the first peak of the
convergence time matches the first transition point $p_1$ in
\figref{fig:p1q}. As the system size increases, the peak moves to
the right. This is also consistent with \figref{fig:p1q} where the
system becomes less complicated as $N$ increases.

\subsection{Other information theoretic and thermodynamic quantities} \label{sec:sec3}

We further fortify and provide our results of the phase diagram of
our systems as ascertained via other information theoretic and
thermodynamics quantities. These measures include the average
normalized mutual information $I_N$ between replica pairs, Shannon
entropy $H$, and energy $E$ as shown in \figref{fig:EHqAll}. We
additionally show the corresponding computational susceptibility
$\chi$ from \figref{fig:susAllalpha} or \ref{fig:susAllq} for
comparison. All panels are for a system of size $N=2048$. In panels
(a) through (d), $q=16$ which corresponds to
\subfigref{fig:susAllalpha}{d}. Panels (e) through (h) plot results
for $q=32$ with $\alpha=0.015$ which corresponds to
\subfigref{fig:susAllq}{d}. Panels (i) through (l), display the
results for $q=70$ which corresponds to
\subfigref{fig:susAllalpha}{l}. Finally, Panels (m) through (p)
display results for $q=140$ and $\alpha=0.07$ corresponding to
\subfigref{fig:susAllq}{h}.

All panels consistently display the three different complexity
phases: the ``easy'' (flat region, lower left), ``hard'' (varied
central regions), and ``unsolvable'' phases (far right or top). The
existence of the hard phase is reflected by the ridges at both low
and high temperatures in the susceptibility $\chi$ plot which often
corresponds rapids shifts (up or down) in the other measures. In
each plot, the red line serves as a guide to the eye to emphasize
the boundaries between different phases. The boundaries are
consistent with each other across the respective rows.

In Ref.\ \cite{ref:huCDPTsgd}, we also demonstrated the spin glass character
of the phase transition by observing the exceptional collapse of time
autocorrelation curves (over four orders of magnitude of time at high and
low temperatures) in the vicinity of the hard phase.
We further elucidated on evidence regarding phase transitions \cite{ref:huCDPTsgd}
in identifying community structure via a dynamical approach (some other dynamical
methods include \cite{ref:arenassync,ref:gudkov}) where ``chaotic-type'' transitions
that we speculated upon may extend into the node dynamics for large systems.

%\subsection{Discussion on the finite size effect}

% --- depiction of cliques --------------------------------------------
\begin{figure}[t!]
\begin{center}
\subfigure[]{\includegraphics[width=\subfigwidth]{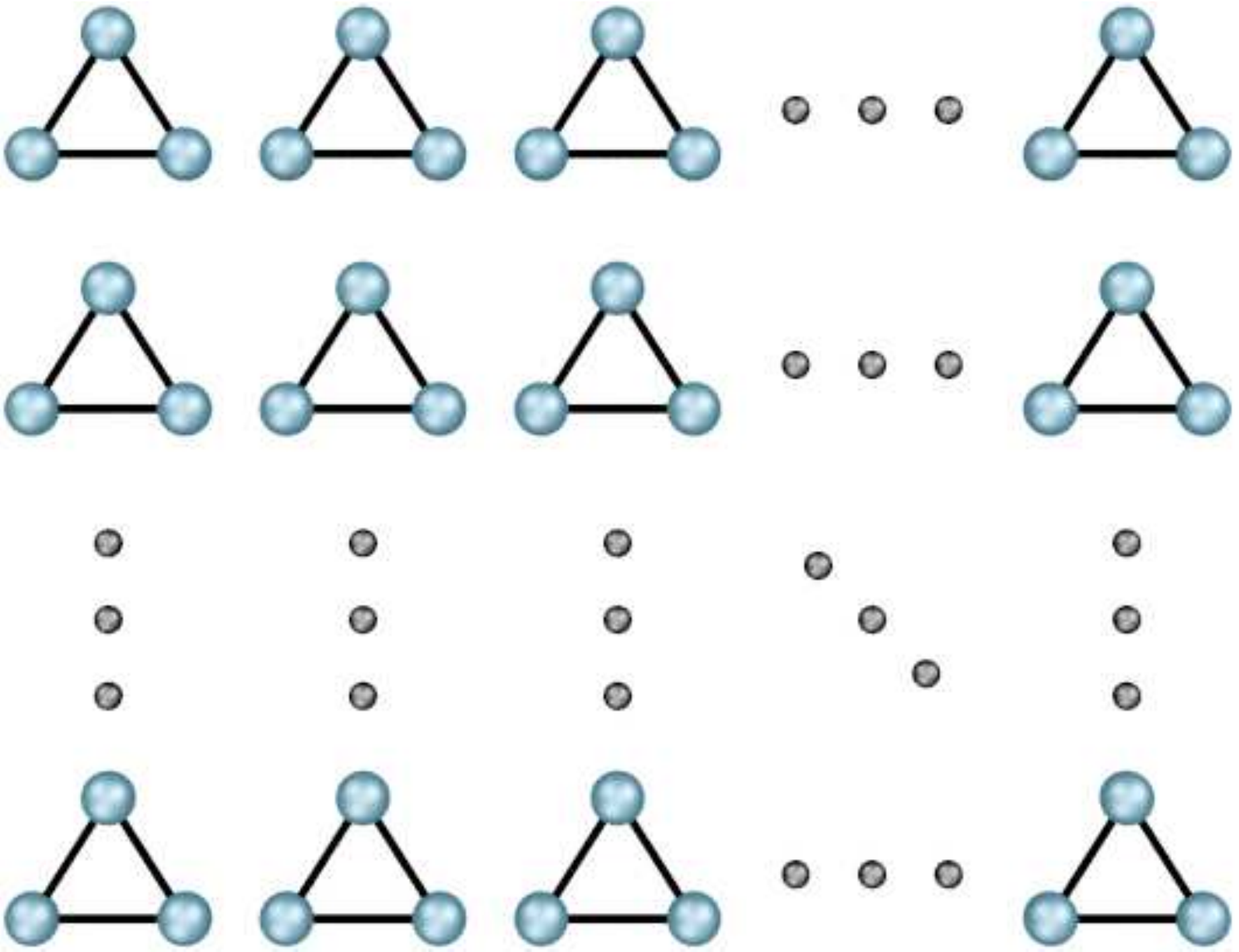}}
\subfigure[]{\includegraphics[width=1.45in]{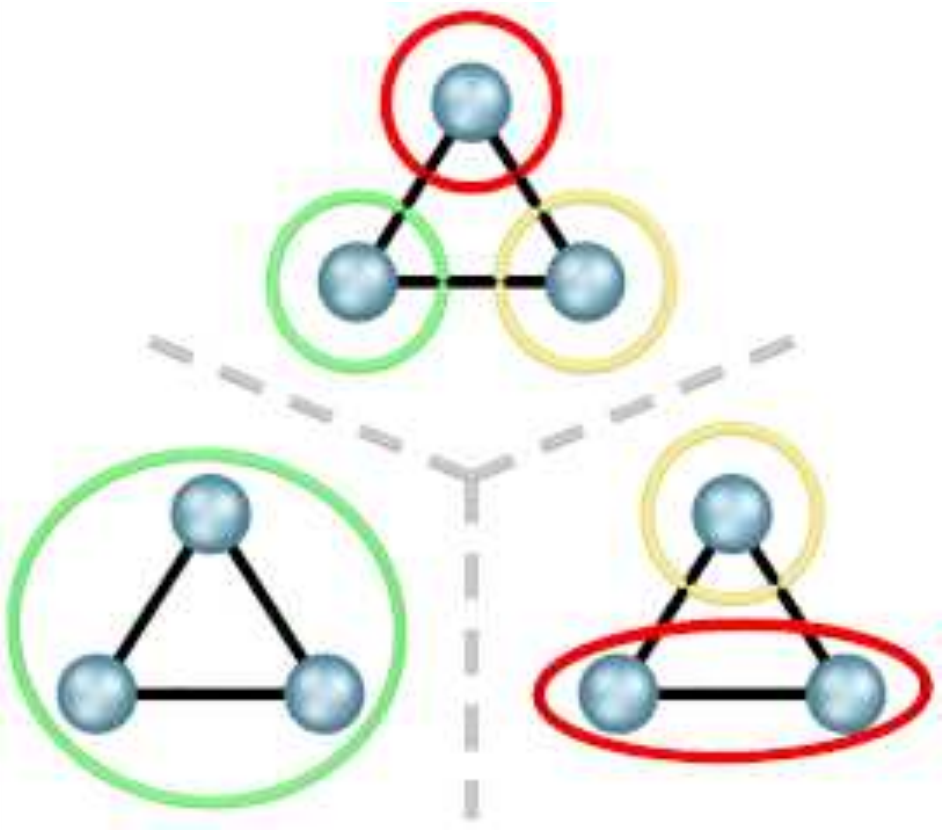}}
\end{center}
\caption{(Color online) Panel (a) depicts $q$ independent cliques (maximally connected clusters).
Panel (b) indicates the different combinations of $l=3$ nodes which must
be summed (including three copies of the $2$-$1$ configuration)
in order to determine the partition function for a single clique.}
\label{fig:cliqueexample}
\end{figure}
% --- end depiction of cliques ----------------------------------------

\section{non-interacting cliques}
\label{sec:NIcliques}

As depicted in \figref{fig:cliqueexample}, we analytically estimate a minimum
transition temperature by examining a system with $q$ non-interacting cliques.
In panel (a), each of the $q$ communities consists of $l$ nodes which are
maximally connected, but no noise exists between these cliques.
The presence of noise will, in general, lower the temperature $\Tcross$ of the
transition point which manifests as departure from the easy phase
in certain regions of \figsref{fig:susAllalpha}{fig:susAllq}.

Within our algorithm and model, communities do not interact in an explicit sense.
In addition, with this model problem the situation is greatly simplified because
no edges are assigned between cliques,
so we use \eqnref{eq:ourpotts} to calculate the partition function of the system
by counting the energy contribution of all edges within each cluster over the
number of combinations for partitioning the clusters.
%According to \eqnref{eq:ourpotts}, the energy of the intercommunity edges is zero.
%The energy of the intracommunity edges are calculated from \eqnref{eq:ourpotts}.
%In the following, we fix $\gamma=1$ in \eqnref{eq:ourpotts}. %assumed for whole paper
%If two nodes within one clique have the same cluster membership,
%the energy associated with the edge between these two nodes is $-1/2$.
%$H=-\frac{1}{2}(A_{ij}-(1-A_{ij}))=-\frac{1}{2}(2A_{ij}-1)|_{A_{ij}=1}=-\frac{1}{2}$.
As a further simplification, we also set the energy contribution for a single edge
to be $-2$ so that the Hamiltonian gives an energy of $-1$ for each edge.

\subsection{Partition function}  \label{sec:largeqexpansion1}

First, we investigate the smallest non-trivial clique size with $l=3$ nodes.
The partition function for the decoupled cliques is,
\begin{equation}
  \mathbf{Z} = \left(\mathbf{Z}_{l}\right)^q=\sum_{\sigma_i,\sigma_j}e^{-\beta H_{i,j}}
\end{equation}
where $\mathbf{Z}_{l}$ is the partition function for a single clique
and $\beta = 1/T$ is the inverse temperature.
Considering the $l=3$ cluster combinations depicted in \subfigref{fig:cliqueexample}{b},
$\mathbf{Z}_3$ is
\begin{eqnarray}
  \mathbf{Z}_{3} & = & qe^{6\beta}+3q(q-1)e^{2\beta}+3q(q-1)(q-2)\nonumber \\
                 &   & +~q(q-1)(q-2)(q-3).
  \label{eq:Z-clique3}
\end{eqnarray}
%The first term on the right hand side of \eqnref{eq:Z-clique3}
%accounts for when all three nodes have the correct cluster membership.
%$q$ possible cluster permutations are reflected in the factor of $q$ in front
%of this term.
%The number $6$ appearing in the exponential denotes the total
%energy of $3$ edges within the cliques. The second and the third
%terms denote the possibility that there are one or two ``mistake''
%nodes (which have the wrong cluster membership) within the clique,
%respectively. The prefactor $3$ in the second and third terms
%denotes the $3$ possible ways to choose one or two nodes
%by ``mistake'' within the $l=3$ clique.
%The fourth term means all the three nodes are in the wrong cluster.
%From the above explanation,
The first term represents the optimal local cluster solution,
and the sum of the remaining terms accounts for the remaining sub-optimal
local partitions.
We define $\omega_l$ as the ratio of Boltzmann weights of the sub-optimal
partitions to the optimal solution.
%For $\mathbf{Z}_3$, the ratio $\omega_3$ is
For $\mathbf{Z}_3$, the ratio $\omega_3$ is
\begin{equation}
  \omega_3 = \frac{q(q-1)\left[3e^{2\beta}+3(q-2)+(q-2)(q-3)\right]}{qe^{6\beta}}.
\end{equation}
$\omega_l<1$ indicates that the optimal solution is dominant, while
$\omega_l\rightarrow \infty$ means the system is disordered. We can
define $\omega_l=1$ as the transition point from the ordered phase
to the disordered phase, and the corresponding ``crossover''
temperature $\Tcross$ is found by solving the transcendental
equation
\begin{eqnarray}
  3(q-1)e^{-4/\Tcross}+3(q-1)(q-2)e^{-6/\Tcross}\nonumber \\
  +~(q-1)(q-2)(q-3)e^{-6/\Tcross}=1. \label{eq:Tc1}
\end{eqnarray}
In the limit of large $q$, this equation simplifies to
%\begin{equation}
%  q\left[3e^{-4/\Tcross}+3qe^{-6/\Tcross}+q^2e^{-6/\Tcross}\right]\simeq 1,
%\end{equation}
%and further simplification, results in $q^3e^{-6/\Tcross}\simeq 1$ which gives
\begin{equation}
  q^3e^{-6/\Tcross}\simeq 1,
\end{equation}
which yields our estimate for the crossover temperature
\begin{equation}
  \Tcross\simeq \frac{2}{\log q}\label{eq:Tcrossl3}
\end{equation}
for the $l=3$ clique system.

If we generalize to arbitrary clique size $l$, the corresponding partition
function for a single clique becomes %in \eqnref{eq:Z-clique3} becomes
\begin{eqnarray}
  \mathbf{Z}_{l} & = & qe^{2\beta {l\choose 2}}+l q(q-1)e^{2\beta {l-1\choose 2}}\nonumber \\
        &   & +~\frac{l(l-1)}{2}q(q-1)(q-2)e^{2\beta{l-2\choose 2}}\nonumber \\
        &   & +~\cdots + q(q-1)(q-2)\cdots (q-l).
  \label{eq:Z-general}
\end{eqnarray}
Again, the first term in \eqnref{eq:Z-general} is the Boltzmann weight
of the optimal clique partition, and the other terms sum the weights
of the incorrect partitions.  $\omega_l$ is
\begin{equation}
  \omega_l = \frac{\sum_{k=1}^{l}{l\choose k}{q\choose k+1}(k+1)!e^{2\beta{l-k\choose 2}}}
             {qe^{2\beta{l\choose 2}}},
\end{equation}
and $\omega_l=1$ returns the cross-over temperature $\Tcross$ for
arbitrary cliques of size $l$.
\begin{comment}
% I think this is extra, so I have commented it out for now.
For $l=4$, the solution $\beta=1/\Tcross$ to $\omega=1$ is estimated to be
\begin{equation}
  6q-11q^3+6q^3-q^4\nonumber \\+(-12+18q-6q^2)x+(4-4q)x^3+x^6 = 0,
  \label{eq:solution-Z4}
\end{equation}
that is, $e^{2\beta}$ is equal to the root $x$ of
\eqnref{eq:solution-Z4}.
In the large $q$ limit, \eqnref{eq:solution-Z4} simplifies to $x^6\simeq q^4$,
which gives $e^{2\beta}\equiv x\simeq q^{2/3}$.
In the $l=4$ case, $\Tcross\simeq3/\log(q)$.

For $l=5$, the solution $\beta=1/\Tcross$ to $\omega=1$ is given
by the solution of the equation,
\begin{eqnarray}
  -120q+250q^2-175q^3=50q^4-5q^5-720 {q\choose 6}\nonumber\\
  +~(60q-110q^2+60q^3-10q^4)x\nonumber\\+(-20q+30q^2-10q^3)x^3\nonumber\\
  +~(5q-5q^2)x^6+qx^{10}=0. \label{eq:solution-Z5}
\end{eqnarray}
In the limit of large $q$, \eqnref{eq:solution-Z5} will become
$qx^{10}\simeq q^6$, and the solution in this case is
$e^{2\beta}=x=\simeq q^{1/2}$. So, when $l=5$, $\Tcross=4/\log
q$.
$\Tcross$ in the system with larger clique size $l$ can be
approximated by the same procedure.
\end{comment}
We summarize the crossover temperature relations in column one of
Table \ref{tab:tableqandl} where we express $e^{2/\Tcross}$ in terms
of powers of $q$ for several values of $l$. The general relation is
\begin{equation}
  \Tcross\simeq \frac{l-1}{\log q}.
  \label{eq:crossoverT1}
\end{equation}

\begin{table}
\begin{tabular}{|c|c|c|c|}
  \hline
  $l$    & $e^{2/\Tcross}$ & $e^{2/\Tcrossq}$ & $e^{2/\Tcrossp}$           \\\hline
  $2$    & $q^2$           & $q$              & $pq^2$                     \\\hline
  $3$    & $q$             & $q^{2/3}$        & $p^{1/3}q$                 \\\hline
  $4$    & $q^{2/3}$       & $q^{1/2}$        & $p^{1/6}q^{2/3}$           \\\hline
  $5$    & $q^{1/2}$       & $q^{2/5}$        & $p^{1/10}q^{1/2}$          \\\hline
  $6$    & $q^{2/5}$       & $q^{1/3}$        & $p^{1/15}q^{2/5}$          \\\hline
  \vdots & \vdots          & \vdots           & \vdots                     \\\hline
  $l$    & $q^{2/(l-1)}$   & $q^{2/l}$        & $p^{2/(l^2-l)}q^{2/(l-1)}$ \\\hline
\end{tabular}
\caption{In column one, the crossover temperature $\Tcross$ from an
``ordered'' to a ``disordered'' state is determined by defining the ratio
$\omega_l=1$ of the sum of Boltzmann weights of sub-optimal node assignments
to the weight of the optimal assignment into clique communities as a function
of the cluster size $l$ and the number of communities $q$ in the large $q$ limit.
In column two, we estimate $\Tcrossq\simeq\Tcross$ through different means
by calculating the probability $p=1/q$ that two nodes (in the same clique ideally)
are determined to be in the same cluster.
In the last column, we generalize column two for an arbitrary probability $p$.}
\label{tab:tableqandl}
\end{table}

\subsection{Symmetry Breaking}\label{sec:largeqexpansion2}

We can inquire about the crossover temperature $\Tcross$ from
another perspective. Take two nodes $i$ and $j$ in the same clique.
If the probability that a solution assigns them to the same
community is high, then the system is in the ``ordered'' state. If
this probability is $1/q$, the system is in its ``disordered''
phase. We can define a crossover temperature $\Tcrossq$ at which the
probability of node $i$ and $j$ being in the same cluster exceeds
$1/q$ and thus symmetry between Potts spins is broken. This
probability $P(\sigma_i=\sigma_j) =
\langle\delta_{\sigma_i,\sigma_j}\rangle$ is
%\begin{eqnarray}
  %p(\sigma_i=\sigma_j) & = & \langle\delta_{\sigma_i,\sigma_j}\rangle\nonumber\\
  %                     & = & \frac{\mbox{Tr}_{\sigma_i}\delta_{\sigma_i,\sigma_j}
  %                                 e^{-\beta H}}{\mbox{Tr}_{\sigma}e^{-\beta H}}
  %\label{eq:p0}
%\end{eqnarray}
\begin{equation}
  P(\sigma_i=\sigma_j) = \frac{\mbox{Tr}_{\sigma_i}\delta_{\sigma_i,\sigma_j}
                         e^{-\beta H}}{\mbox{Tr}_{\sigma}e^{-\beta
                         H}},
  \label{eq:p0}
\end{equation}
where $\sigma_i$ and $\sigma_j$ denote the cluster memberships for nodes $i$ and $j$,
respectively.
Expressing the numerator and in terms of $l$ and $q$, \eqnref{eq:p0} becomes,
\begin{eqnarray}
  P(\sigma_i=\sigma_j) & = & \Big\{qe^{2\beta{l\choose 2}}
                            +(l-2)q(q-1)e^{2\beta{l-1 \choose2}}         \nonumber\\
                       &   & +~\cdots +q(q-1)\cdots (q-l-2)\Big\}                 \nonumber\\
                       &   & \Big/\Big\{qe^{2\beta {l\choose2}}+lq(q-1)
                             e^{2\beta{l-1\choose 2}}                    \nonumber\\
                       &   & +~\cdots +q(q-1)(q-2)\cdots (q-l)\Big\}.
  \label{eq:p1}
\end{eqnarray}
%\begin{widetext}
%\begin{equation}
%  p(\sigma_i=\sigma_j) = \frac{qe^{2\beta{l\choose 2}}
%                           +(l-2)q(q-1)e^{2\beta{l-1 \choose2}}
%                           +\cdots +q(q-1)\cdots (q-l-2)}
%                            {qe^{2\beta {l\choose2}}+lq(q-1)
%                               e^{2\beta{l-1\choose 2}}
%                           +\cdots +q(q-1)(q-2)\cdots (q-l)}.
%  \label{eq:p1}
%\end{equation}
%\end{widetext}
In the limit of large $q$, \eqnref{eq:p1} simplifies to
%\begin{eqnarray}
%  p(\sigma_i=\sigma_j) & \simeq & \left\{qe^{2\beta{l\choose 2}}+\sum_{k=1}^{l-2}
%                                   {l-2\choose k}q^{k+1}e^{2\beta{l-k\choose2}}\right\}\nonumber\\
%                       &        & \Bigg/\left\{qe^{2\beta{l\choose 2}}+\sum_{k=1}^{l}{l\choose k}
%                                  q^{k+1}e^{2\beta{l-k\choose 2}}\right\}.
%  \label{eq:p1s}
%\end{eqnarray}
\begin{equation}
  P(\sigma_i=\sigma_j) \simeq \frac{qe^{2\beta{l\choose 2}}+\sum_{k=1}^{l-2}
                                   {l-2\choose k}q^{k+1}e^{2\beta{l-k\choose2}}}
                              {qe^{2\beta{l\choose 2}}+\sum_{k=1}^{l}{l\choose k}
                                  q^{k+1}e^{2\beta{l-k\choose 2}}}.
  \label{eq:p1s}
\end{equation}
Choosing $P(\sigma_i,\sigma_j)=1/q$ yields in a crossover
temperature $\Tcrossq$ at which the system goes from being unbroken
q-state symmetry to ordered. When $l=3$, \eqnref{eq:p1s} becomes,
\begin{equation}
  q^2e^{6\beta}+q^3e^{2\beta} \simeq qe^{6\beta}+3q^2e^{2\beta}+3q^3+q^4.
  \label{eq:p1l3}
\end{equation}
In the large $q$ limit, $e^{2\beta}\simeq q^{2/3}$, and the crossover temperature
is $\Tcrossq=3/\log q$.
The asymptotic expressions for several values of $q$ and $l$ are summarized
in column two of Table \ref{tab:tableqandl}.
For general $q$ and $l$, the relation is
\begin{equation}
  \Tcrossq\simeq \frac{l}{\log q}.
  \label{eq:crossoverT2}
\end{equation}

%In the limit of the large number of communities $q$, the
%crossover temperature $\Tcrossq$ in \eqnref{eq:crossoverT2}
%given by setting the probability $p(\sigma_i,\sigma_j)=1/q$, is
%consistent with the one $\Tcross$ in \eqnref{eq:crossoverT1} by
%requiring the ratio of the incorrect to true Boltzmann weight
%$\omega=1$.
For a general crossover probability $P(\sigma_i,\sigma_j)=p$ with $l=3$,
%we can set $p(\sigma_i,\sigma_j)=p$, $0<p<1$ is an
%arbitrary number less than $1$.
%To solve the equation $p(\sigma_i,\sigma_j)=p$, we use the
%expression for $p(\sigma_i,\sigma_j)$
the crossover temperature $\Tcrossp$ is determined by solving
\begin{equation}
  e^{6\beta}+qe^{2\beta}\simeq pe^{6\beta}+3qpe^{2\beta}+3q^2p+q^3p.
  \label{eq:p1l3_2}
\end{equation}
In the large $q$ limit, \eqnref{eq:p1l3_2} is $e^{2\beta}\simeq p^{1/3}q$,
where $\Tcrossp\simeq 2/(\log q+1/3\log p)$.
\begin{comment}
When $l=4$, the solutions to $p(\sigma_i,\sigma_j)=p$ become
equivalent to the solutions of the following,
\begin{equation}
  4pq^3+pq^4+(-q^2+6pq^2)x\nonumber\\+(-2q+4pq)x^3+(-1+p)x^6=0
\end{equation}
When $q\rightarrow\infty$, the above equation can be solved as
$e^{2\beta}=p^{1/6}q^{2/3}$, and $\Tcrossq=2/(1/6\log p+2/3\log q)$.
\begin{equation}
  \Tcrossq=\frac{2}{1/(l^2-l)\log p+1/(l-1)\log q}.
  \label{eq:crossoverT3}
\end{equation}
\end{comment}
Results for $\Tcrossp$ for several values of $q$ and $l$ are
shown column three of Table \ref{tab:tableqandl}.
For general $q$ and $l$, the relation is
\begin{equation}
  \Tcrossp\simeq\frac{1}{p^{1/l}q}.
\end{equation}

% --- begin asymptotic/simulation comparison ------------------------------
\begin{figure}[t]
\begin{center}
\includegraphics[width=0.9\linewidth]{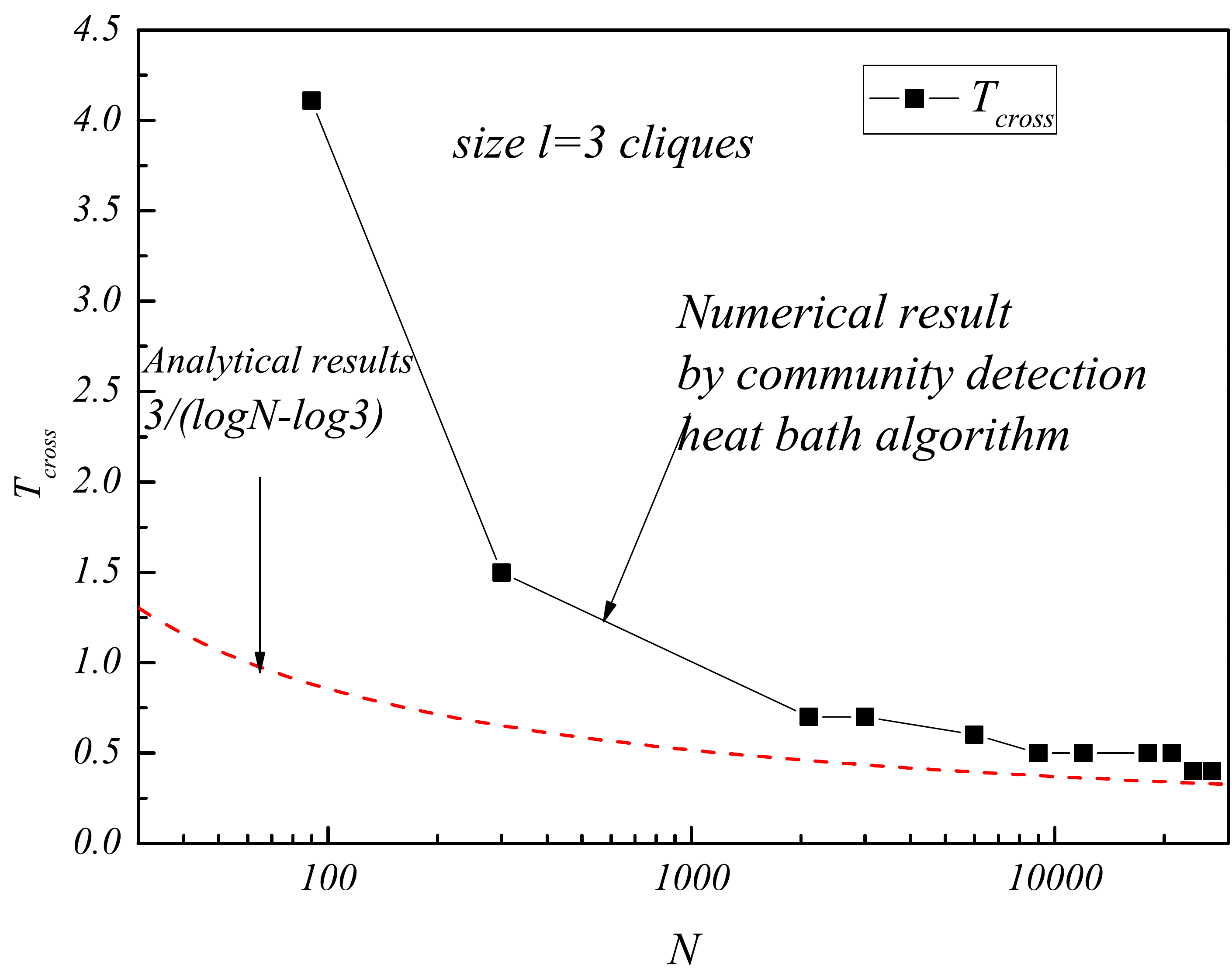}
\end{center}
\caption{(Color online) The crossover temperature at which the system cannot be perfectly
solved as the function of the system size $N$.
The data here uses cliques of size $l=3$.
The dashed line is the analytical result and the solid line is determined
by the heat bath community detection algorithm optimizing the Hamiltonian
of \eqnref{eq:ourpotts}.}
\label{fig:TcN}
\end{figure}
% --- end asymptotic/simulation comparison --------------------------------

\subsection{Simulated crossover temperature}

We can also simulate the crossover temperature $\Tcross$ or $\Tcrossp$
as a function of system size $N$ by solving the non-interacting clique problem
using our heat bath community detection algorithm (see \Appref{app:HBA}).
As seen in \figref{fig:TcN}, the simulated and analytic asymptotic behaviors
agree well in the large $N$ limit, so the crossover temperature for this trivial
system is $\Tcross = 0$.

The crossover temperature derived in this section deals with a
\emph{heat-induced} disorder. That is, it marks the onset of a
``liquid'' phase that transitions at a lower heat bath temperature
as the system size grows. In practice, one uses a SA algorithm that
applies a cooling scheme (as opposed to a constant temperature HBA)
to improve the attempt at locating the ground state of the system.
That is, it applies a high temperature exploration of the general
landscape finished by low temperature ``fine tuning'' of the solution.
For the non-interacting cliques in this section, SA would obviously still
identify the ground state because the energetic fluctuations would trivially
diminish as the system is cooled toward $T=0$.

With increasing $p_{out}$ at low $T$, disorder imposed by the
glass-type transition is induced by the complexity of the energy
landscape,
%(note that the system is effectively cooled instantaneously
%when the heat bath algorithm is initialized at a specific temperature),
but the transition is qualitatively comparable in the sense of the
induced disorder in the solutions found by the HBA. The glass phase
also experiences a transition to a liquid-like disordered state at a
temperate that increases slowly with the level of noise, but here, a
SA solver will not necessarily transition readily to the ideal
solution as the system is cooled because of the inherent complexity
of the energy landscape. The greedy algorithm used in
\cite{ref:rzlocal} (equivalent to the HBA at $T=0$) applied to the
Potts model of \eqnref{eq:ourpotts} is already very accurate
\cite{ref:rzmultires,ref:rzlocal,ref:lancLFRcompare}, so we expect
that the greatest benefit of SA over a greedy-oriented solver using
\eqnref{eq:ourpotts} will manifest in the hard region near the onset
of the ``glassy'' transition.

% --- bird images ---------------------------------------------------------
% --- moved here for figure placement -------------------------------------
\begin{figure}
\centering
\subfigure[\ Original]{\includegraphics[width=0.45\columnwidth]{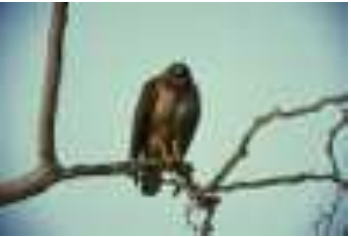}}
\subfigure[\ Easy]{\includegraphics[width=0.45\columnwidth]{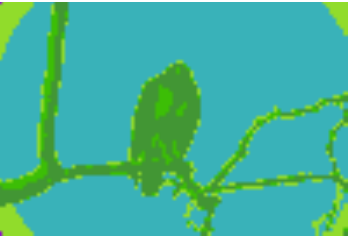}}
\subfigure[\ Hard]{\includegraphics[width=0.45\columnwidth]{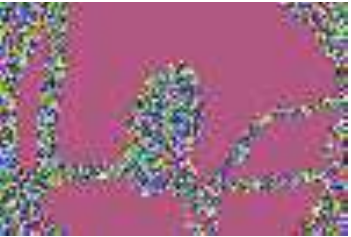}}
\subfigure[\ Unsolvable]{\includegraphics[width=0.45\columnwidth]{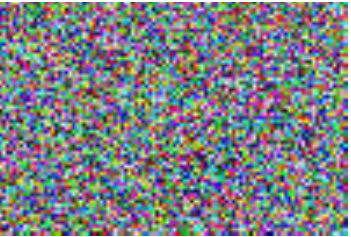}}
\caption{(Reproduced from Ref.\ \cite{ref:HRNimages})
We show an image where we apply our community detection algorithm
to detect the relevant structures.
This case seeks to identify a bird and tree against a sky background.
The original images is in panel (a), and the segmentation results are shown
in panels (b--d) corresponding to the easy, hard, and unsolvable regions
of the community detection problem, respectively.
\Figref{fig:birdplot} shows the phase diagram identifying these respective regions.}
\label{fig:birdimages}
\end{figure}
% --- end bird images -----------------------------------------------------

% --- bird phase plot -----------------------------------------------------
\begin{figure}
\centering
\includegraphics[width=0.9\columnwidth]{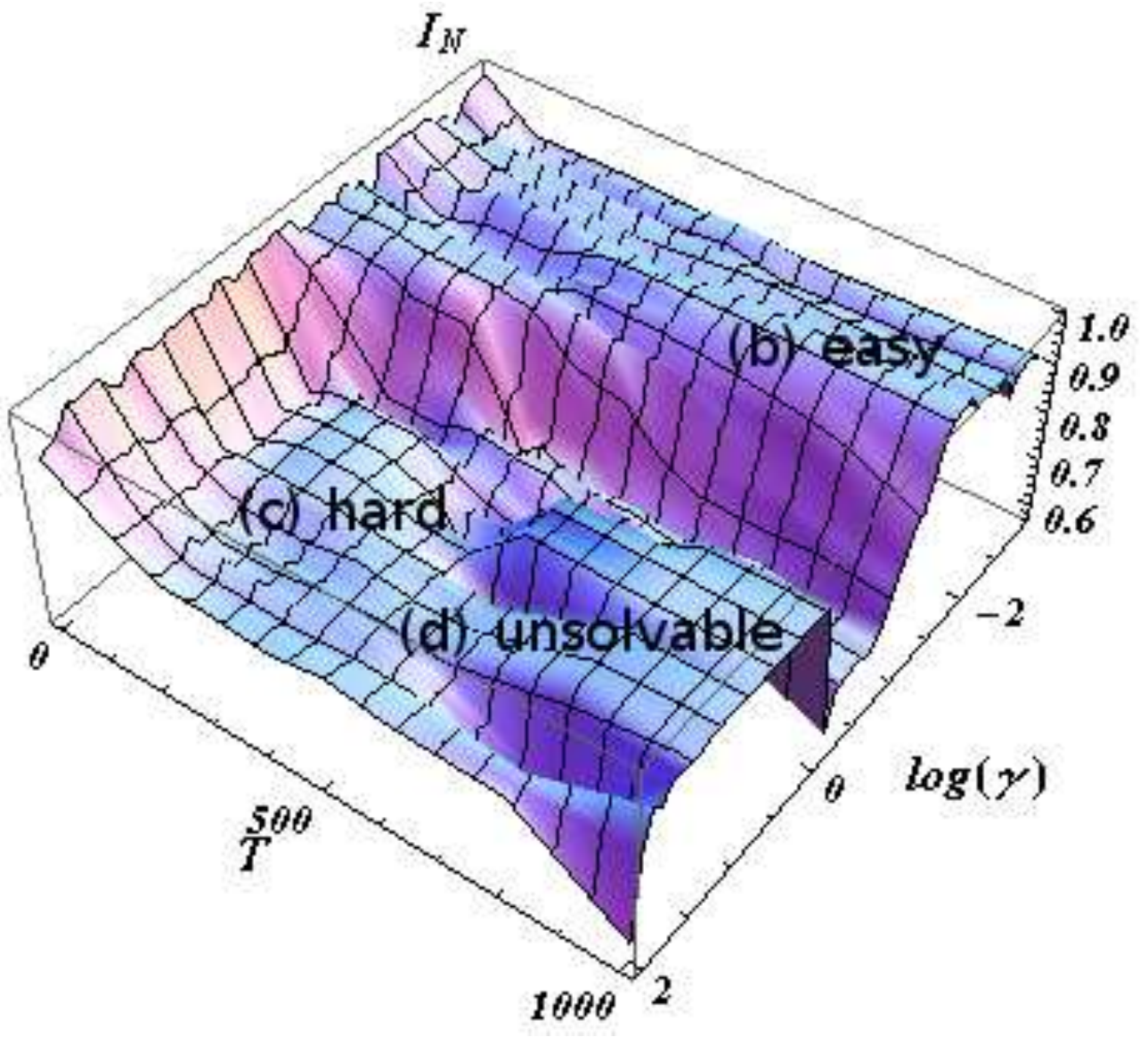}
\caption{(Reproduced from Ref.\ \cite{ref:HRNimages})
We show a three-dimensional phase diagram of NMI ($I_N$) versus $\log(\gamma)$
and $\log(T)$ for the image segmentation of the bird in \figref{fig:birdimages}.
$T$ is heat bath temperature for a stochastic community detection solver
(see \Appref{app:HBA}), and $\gamma$ is the model weight in \eqnref{eq:ourpotts}.
We note that the optimal values in the easy and hard regions correspond to
the ``physical'' segmentations of the bird and tree against the background,
but the bird is undetectable in the unsolvable region.}
\label{fig:birdplot}
\end{figure}
% --- end bird phase plot -------------------------------------------------

\subsection{A discussion of the crossover temperature} \label{sec:crossoverdisc}

For a spin system with fixed size $N$, a larger number of spin states $q$
corresponds to a more disordered system.
If we expand the partition function of the Potts model in terms of $1/q$,
it is explicitly represented as a sum over configurations with progressively
larger clusters of identical spins \cite{ref:kasteleyn}.
That is, two spins with the same index $\sigma_i=\sigma_j$ are connected.
Then three spins $\sigma_i=\sigma_j=\sigma_k$ are connected, etc.
The resulting terms illustrate that increasing $q$ emulates increasing
temperature $T$.

Our analysis in this section applies to general graphs with ferromagnetic interactions
(equivalent to the ``label propagation'' community detection algorithm \cite{ref:LPA})
on regular, fixed-coordinate lattices \cite{ref:mercaldo,ref:preissmann,ref:juhasz}.
%We showed that ferromagnetic Potts models with \emph{infinite} spin states are
%effectively always in their disordered high temperature phase in the thermodynamic
%limit, large $N$, limit.
Increasing the number of system states $q$ causes the system to be increasingly
disordered.
Thus, in the community detection problem, increasing number of communities $q$
linearly with the system size $N$ (such that the average community size remains
constant), the solvable (easy) phase shrinks to a ``small'' region as $N\to\infty$.

Figures \ref{fig:EHqAll}(m--p) illustrate the distinction in the different
regions or types of disorder: entropic (high complexity) and energetic (high $T$).
Interestingly, in some cases, additional noise emulates a higher temperature
solution process in the sense that it provides additional avenues to explore
different configurations.
Such an effect may occur in \subfigref{fig:EHqAll}{a-d} where the accuracy [$I_N$
in panel (b) \emph{increases} for a short time with increasing noise $p_{out}$].

\subfigref{fig:EHqAll}{n} further shows a crossover region $0.24
\lesssim p_{out}\lesssim 0.32$ where mid-range temperatures improve
the solution accuracy (higher $I_N$). Although this data uses a
constant temperature heat bath (no cooling schedule), this is the
effect of a stochastic solver (see \Appref{app:HBA}), allowing it to
navigate the difficult energy landscape more accurately than a
greedy solver. On the left (lower $T$), the more greedy nature of
the solver prevents an accurate solution in the presence of high
noise. On the right, the higher temperature of the heat bath itself
hinders an accurate solution. In effect, the HBA ``wanders'' at
energies above the meaningful, but locally complex, features of the
energy landscape resulting in more random solutions.

The results here incorporate a ``global'' model parameter $\gamma$ in \eqnref{eq:ourpotts}.
That is, the model asserts globally optimal $\gamma$(s) for the entire graph.
For large graphs, this condition is less likely to be true across the full scope
of the network, but one can explore methods to obtain locally optimal $\gamma_\ell$
(in time or space) for each region or cluster $\ell$ \cite{ref:rzLMRO}.
Utilizing locally optimal $\gamma_\ell$s will likely work to circumvent the temperature
transition at low levels of noise.
The successful selection of a local $\gamma_\ell$ in the glassy (high noise) region
is more difficult because of the complex nature of the local energy landscape.

In the following section, we study the free energy of several systems for ferromagnetic
Potts models and then generalize to \emph{arbitrary} weighted Potts models, including
antiferromagnetic interactions, on \emph{arbitrary} graphs \cite{ref:rhzglobaldisorder}.

\section{An example of a Phase transition in an image segmentation problem} \label{sec:birdexample}

We illustrate the phase transition effect with an realistic image segmentation
example \cite{ref:HRNimages}.
In \figref{fig:birdimages}, we apply our community detection algorithm to
detect a bird and tree against a sky background.
We display the results in \figref{fig:birdplot} where we plot NMI ($I_N$)
versus $\log(\gamma)$ in \eqnref{eq:ourpotts} and $\log(T)$ where $T$ is
the temperature for our stochastic community detection solver (see \Appref{app:HBA}).
For this problem, we apply edge weights by replacing the $A_{ij}$ elements
in \eqnref{eq:ourpotts} with ``attractive'' and ``repulsive'' weights $w_{ij}$
which are defined by regional intensity differences within the image
\cite{ref:HRNimages}.

We label the easy (b), hard (c), and unsolvable (d) regions in the phase plot
for the bird image in panel (a).
Panel (b) shows that our algorithm clearly detects the bird and tree against the
background, meaning that the NMI information measure identifies the physically
relevant clusters in the problem.  In panel (c), the background is segmented
separately, but the bird and tree are composed of many small clusters.
Panel (d) shows that the bird is undetectable in the unsolvable region.

\section{free energy: Simple results} \label{sec:f-analytic}

In the following analysis, we explicitly show the large $q$ and large $T$ expansions
for the free energy per site in three example systems (a non-interacting clique system,
simple interacting clique system, and a random graph) before generalizing the analysis
to arbitrary unweighted and weighted graphs.
Previous works examined disorder transitions for random-bond Potts models
\cite{ref:juhaszrbPMlargeq,ref:mercaldorbPMlargeq} and Ref.\ \cite{ref:changlargeqzeroes}
studied zeros of the partition function in the large $q$ limit.
Large $q$ behavior was shown to approach mean-field theoretical results
on fixed lattices \cite{ref:pottslargeqmean,ref:pearcelargeqmean}.
%and on more general networks with varying node degrees \cite{ref:rhzglobaldisorder}.
For the unweighted systems, we use a binary distribution for the interaction strength
$J=1$ or $0$ (\ie{}, the energy contribution of an edge is either ``on'' or ``off'').

\subsection{Free energy of a non-interacting clique system under a large $q$ expansion} \label{sec:f-q}

If we generalize the non-interacting clique system in \figref{fig:cliqueexample} to cliques
of size $l$, the partition function is
%\begin{eqnarray}
%  Z & = & \bigg[
%              qe^{\beta J {l\choose 2}}+lq(q-1)e^{\beta J {l-1\choose 2}} \nonumber\\
%    &   &   +~\frac{l(l-1)}{2}q(q-1)(q-2)e^{\beta J {l-2\choose 2}} \nonumber\\
%    &   &   +~\cdots + q(q-1)(q-2)\cdots(q-l)
%          \bigg]^q
%\end{eqnarray}
\begin{eqnarray}
  \mathbf{Z} & = & \bigg[ qe^{\beta J {l\choose 2}}+lq(q-1)e^{\beta J {l-1\choose 2}} \nonumber\\
    &   &   +~\frac{l(l-1)}{2}q(q-1)(q-2)e^{\beta J {l-2\choose 2}} \nonumber\\
    &   &   +~\cdots + q(q-1)(q-2)\cdots(q-l) \bigg]^q.
\end{eqnarray}
When $q\to\infty$,
\begin{eqnarray}
  \mathbf{Z} & \approx & \bigg[ qe^{\beta J {l\choose 2}}+lq^2e^{\beta J{l-1\choose 2}}
                  + \frac{l(l-1)}{2}q^3e^{\beta J {l-2\choose 2}} \nonumber\\
    &         & +~\cdots  + q^{l+1} \bigg]^q.
\end{eqnarray}
%After factoring $q^{l+1}$ and taking the logarithm of both sides, it becomes
%\begin{eqnarray}
%  \log Z %& =  & q\log\bigg[q^{l+1}+\cdots +\frac{l(l-1)}{2}q^3e^{\beta J{l-1\choose 2}}\nonumber\\
%         %&    &  + lq^2e^{\beta J{l-1\choose 2}}+qe^{\beta Jl}\bigg]\nonumber\\
%         %& =  &  q\log q^{l+1}+q\log\bigg[ 1+\cdots +\frac{l(l-1)}{2}q^{2-l}e^{\beta J{l-2\choose 2}}\nonumber\\
%         %&    &  + lq^{1-l}e^{\beta J{l-1\choose 2}}+q^{-l}e^{\beta J {l\choose 2}}\bigg] \\
%    & \approx &  (l+1)q\log q + lqq^{-1}+q\frac{l(l-1)}{2}q^{-2}e^{\beta J}\nonumber\\
%         &    &   +~\cdots + q\frac{l(l-1)}{2}q^{2-l}e^{\beta J {l-2\choose 2}}\nonumber\\
%         &    &   +~lqq^{1-l}e^{\beta J {l-1\choose 2}}+qq^{-l}e^{\beta J {l\choose 2}}
%         \label{eq:f-q-approx}
%\end{eqnarray}
%after writing the right-hand summation of \eqnref{eq:f-q-approx} in a general form,
The free energy per site, $f=-\frac{k_B T}{N}\log \mathbf{Z}$, (with
the Boltzmann constant set to $k_B=1$) is
\begin{equation}
%\begin{eqnarray}
  %f & =       & \frac{1}{N}\log Z \nonumber\\
    %& =       & \frac{l+1}{N}q\log q + \frac{l}{N}qq^{-1}+\frac{l(l-1)}{2N}qq^{-2}e^{\beta J}\nonumber\\
    %&         & + \cdots + \frac{l(l-1)}{2N}qq^{2-l}e^{\beta J {l-2\choose 2}}\nonumber\\
    %&         & + \frac{l}{N}qq^{1-l}e^{\beta J{l-2\choose 2}}
    %              + \frac{qq^{-l}}{N}e^{\beta J{l\choose 2}}\nonumber\\
  %f & \approx & \log q +\sum_{k=0}^{l-2} \left[
  %                {l-1\choose k}\frac{1}{k+1}  e^{\beta J {l-1\choose k-1}}
  %              \right] q^{-(k+1)}.~~
  f \approx-T\log q -T\sum_{k=0}^{l-2} \left[
                  {l-1\choose k}\frac{1}{k+1}  e^{\beta J {l-1\choose k-1}}
                \right] q^{-(k+1)}.
  \label{eq:f-q}
%\end{eqnarray}
\end{equation}
From \eqnref{eq:f-q}, we further simply the free energy per site
%\begin{equation}
%  f \approx \log q+\sum_{k=0}^{l-2}a(k) e^{\beta J {l-1\choose k-1}} q^{-(k+1)}
%  \label{eq:f-qNIfinal}
%\end{equation}
\begin{eqnarray}
  f & \approx & -T\log q - T\sum_{k=0}^{l-2}a(k) e^{\beta J {l-1\choose k-1}} q^{-(k+1)} \nonumber\\
  f & \approx & -T\log q - T a(0) \frac{e^{\beta J}}{q}
  \label{eq:f-qNIfinal}
\end{eqnarray}
where
%$\begin{equation}
 $ a(k)={l-1\choose k}\frac{1}{k+1}$.
%\end{equation}
We will compare \eqnref{eq:f-qNIfinal} with the high $T$ expansion
in the next section. Despite the functional dependence of
$\exp(\beta J)$, the large $q$ limit dominates the expansion,
forcing the system to be approximately equivalent to a large
temperature limit.

\subsection{Free energy of a non-interacting clique system as
ascertained from a high temperature expansion} \label{sec:f-T}

Note that the \emph{most} ordered Potts graph is a system of non-interacting
cliques (maximally connected sub-graphs).
That is, the presence of noise (extraneous intercommunity edges) will only
serve to \emph{increase} the overall disorder in the system.
One exception is that increased disorder can emulate increased temperature
$T$ for both greedy and stochastic community detection solvers
(see also \secref{sec:crossoverdisc}).

We can construct the high $T$ expansion easily by means of Tutte polynomials
\cite{ref:pottstutte} (see \Appref{app:tuttepolyUW}) where we again solve a
system of $q$ cliques of size $l$.
%Calculating the partition function with the general Potts model
%of \eqnref{eq:ourpotts} (including \emph{antiferromagnetic} interactions
%within clusters is more difficult), but
\Eqnref{eq:ourpotts} and a ferromagnetic Potts model have the same
\emph{ground state} energy for this clique system (see also Secs.\
\ref{sec:f-arbitraryZ}, \ref{sec:f-wPottsT}, and \ref{sec:f-wPottsq}
for more general derivations), so the partition function in terms
of the Tutte polynomial $t(G;x,y)$ for a graph $G$ is
\begin{equation}
  \mathbf{Z} = q^{k(G)}v^{|V|-k(G)}t(G;x,y)
  \label{eq:ZarbitraryG}
\end{equation}
where $q$ is the number of clusters or states,
$v=\exp(\beta J)-1$, $G$ denotes the graph, $k(G)$ is the number of connected
components in $G$, $|V|=N$ is the number of vertices, $x=(q+v)/v$ and $y=v+1$.
For the non-interacting clique system, $k(G)=q$ and $N=lq$.
We denote the Tutte polynomial of a single clique of size $l$ as $K_l(G;x,y)$.

$K_2(G;x,y) = x$, so the partition function is
\begin{eqnarray}
  \mathbf{Z} & = & q^{q} v^{q} x^q , \nonumber\\
  \mathbf{Z} & = & q^{q} v^{q} \left( \frac{q}{v} + 1\right)^q ,%\\
  %Z & = & \left[ q \left( e^{\beta J} -1\right)
  %              \left( \frac{q}{\left( e^{\beta J} -1\right)} + 1\right) \right]^q \%\
  %Z & = & \left[ q \left( q + e^{\beta J} \right) \right]^q
  \label{eq:ZNICfinalHighTl_two}
\end{eqnarray}
where we used $N=2q$.
In a high $T$ approximation, $x\approx q/v\gg 1$, so the partition
function becomes $Z\approx q^{2q}$, and the free energy is
\begin{equation}
  f \approx -T \log q,
  \label{eq:fNICaltl_two}
\end{equation}
which simply states that the system is completely random
in the large $T$ limit.

For triangle cliques, $K_3(G;x,y)=x^2+x+y$.
The graph $G$ is composed of disjoint triangles, so the Tutte
polynomial is $t(G;x,y)=\left(x^2+x+y\right)^{q}$, and
%\begin{equation}
%  t(G;x,y=0)= \left(x^2+x+y\right)^{q}.
%  \label{eq:tNICyzero}
%\end{equation}
the partition function becomes
\begin{equation}
  \mathbf{Z} \approx q^q v^{2q} \left( x^2+x+y \right)^q.
  \label{eq:ZNICfinal}
\end{equation}
In a high $T$ approximation $y\approx 1$, but $x\approx q/v\gg 1$ in
either the large $q$ or large $T$ limits, so we make a further
approximation of $y\approx 0$. Then, $K_3(G;x,y=0) = x^{q}
(x+1)^{q}\approx x^{2q}$. The partition function simplifies to
$\mathbf{Z}\approx q^{3q}$, so the free energy per site for $l=3$ is
again
\begin{equation}
  f \approx -T \log q
  \label{eq:fNICalt}
\end{equation}
which is identical to the $l=2$ result because we consistently applied the
approximation $q/v \gg 1$ to $x=(q/v+1)\approx q/v$ and $(x+1)=(q/v+2)\approx q/v$.

Generalizing to an arbitrary clique size $l$ in the high $T$ approximation,
the Tutte polynomial
$K_l(G;x,y=0)$ is
\begin{equation}
  \Kl{l}(G;x,y=0) = \Pochhammer{x+l-1}{x},
  \label{eq:cliquetofKhighT}
\end{equation}
The partition function is
\begin{equation}
  \mathbf{Z}\approx q^{lq} \left( \frac{v}{q}\right)^{(l-1)q}
              \Pochhammer{\frac{q}{v}+l-1}{\frac{q}{v}},
  \label{eq:ZNICalt}
\end{equation}
and $v=e^{\beta J} - 1 \approx \beta J$, so the free energy per site yields
\begin{equation}
  f \approx -T\log q - \frac{l-1}{l}T \log\left( \frac{\beta J}{q} \right)
        - \frac{T}{lq}\log\left[
              \Pochhammer{\frac{q}{\beta J} + l -1}{\frac{q}{\beta J}}
          \right].
  \label{eq:fNIClalt}
\end{equation}
The leading $\log q$ term represents the infinite $T$ limit which is
approximately constant in large systems for any clique size $l$.
That is, the partition function $\mathbf{Z}_{T\to\infty}\approx q^N$
for every system. The $l=2$ and $3$ results above illustrate that
when $l\ll q$, the ratio of gamma functions in \eqnref{eq:fNIClalt}
simplifies to $x^{lq}$, and the free energy for the non-interacting
clique system is approximately $f \approx \log q$
in the large $T$ limit. %temperature ($T<\infty$) limit.

The second term in \eqnref{eq:fNIClalt} gives the leading order correction
for high $T$.
It is absent in the explicit $l=2$ and $3$ results above because we applied
the approximation $q/v \gg 1$.
Together, the last two terms imply that increasing the temperature $T$
(decreasing $\beta$) emulates increasing the number of communities $q$
for a ferromagnetic Potts model.

% --- Clique circle ---------------------------------------------------------
\begin{figure}[t]
\begin{center}
\includegraphics[width=\subfigwidthnew]{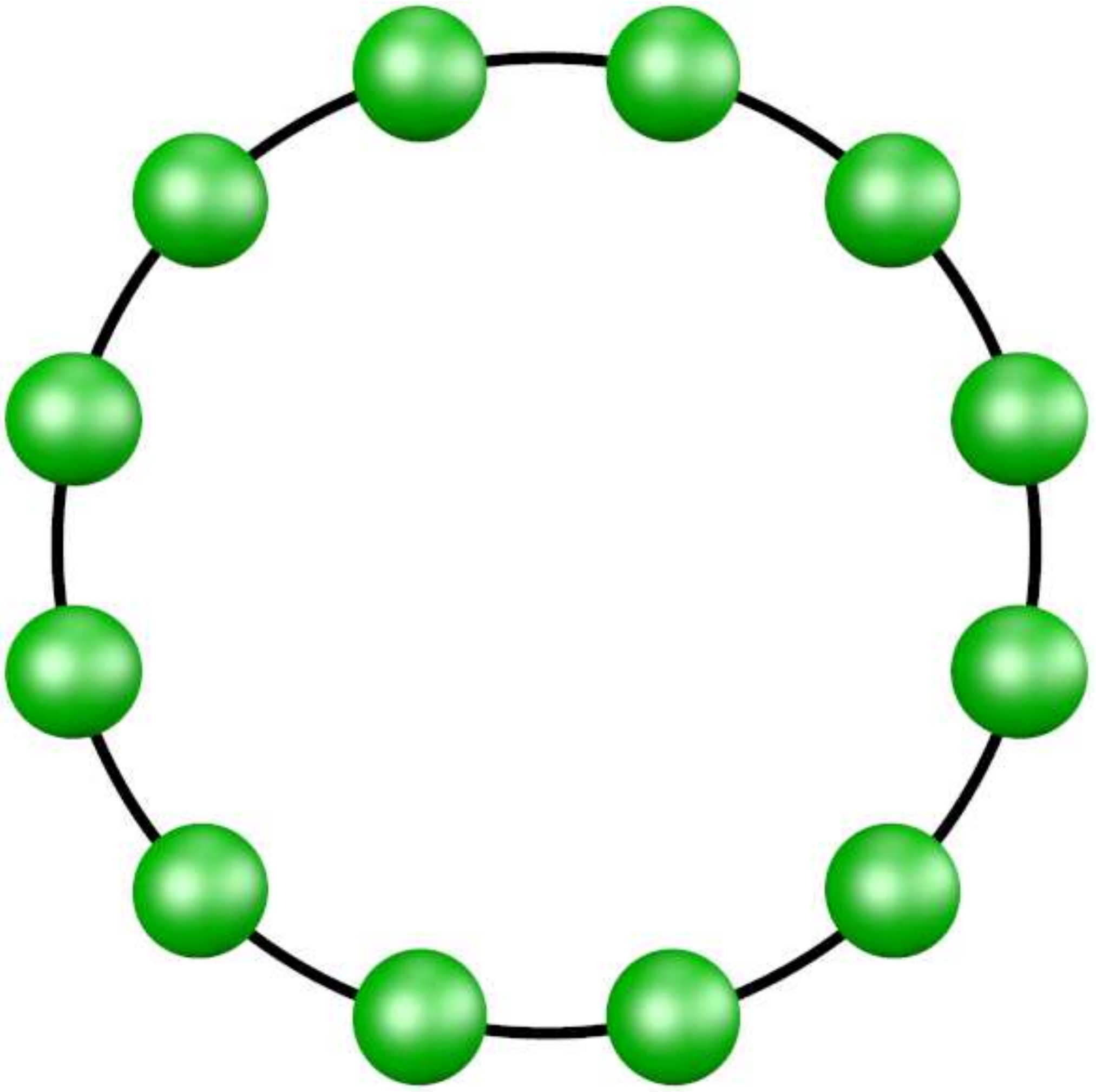}
\end{center}
\caption{(Color online) A depiction of a circle of cliques (maximally connected clusters)
of size $l$ connected by single edges.
In contrast to \figref{fig:cliqueexample}, this system adds a simple interaction
between cliques. We analyze the configuration in \secref{sec:f-qandT1} and show
that a ferromagnetic Potts model behaves the same in the large $q$ and large $T$
limits.}
\label{fig:cliquecircle}
\end{figure}
% --- end clique circle -----------------------------------------------------

\subsection{Free energy for the ``circle of cliques'' in the high $q$
or the high $T$ expansion} \label{sec:f-qandT1}

We now investigate the slightly more complicated system depicted
in \figref{fig:cliquecircle}: a ``circle of cliques''  where each
complete sub-graph cluster is connected to its neighbors by a single edge.
%That is, instead of working on non-interacting clique system in \secref{sec:f-q}
%and \secref{sec:f-T} as depicted in \figref{fig:cliqueexample}, we examine
%one of the simplest interacting clique systems.
%We will discuss the large $q$ and large $T$ expansion for this system
%in the next section.
We construct $q$ cliques of size $l=3$ and apply the Tutte polynomial
method \cite{ref:pottstutte} to solve the system.
%Calculation of the partition function with the general Potts model
%of \eqnref{eq:ourpotts} (including \emph{antiferromagnetic} interactions
%within clusters) is difficult, but the ground state of \eqnref{eq:ourpotts}
%and a ferromagnetic Potts model has the same energy, so we continue with the
%ferromagnetic model.
As in the previous sub-section, the ground state of \eqnref{eq:ourpotts}
and a ferromagnetic Potts model have the same energy, so we use a ferromagnetic
model.
In terms of the Tutte polynomial $t(G;x,y)$ for a graph $G$, the partition
function is given by \eqnref{eq:ZarbitraryG}.
%\begin{equation}
%  Z = q^{k(G)}v^{|V|-k(G)}t(G;x,y)
%\end{equation}
%where $q$ is the number of clusters or states,
%$v=\exp(\beta J)-1$, $G$ denotes the graph, $k(G)$ is the
%number of connected components, $|V|$ is the number of vertices,
%%$t(G;x,y)$ is the Tutte polynomial expression for graph $G$,
%$x=(q+v)/v$ and $y=v+1$.

\Eqnref{eq:Cqexpanded} in \Appref{app:cliquecircle} derives the exact Tutte
polynomial for \figref{fig:cliquecircle} with $l=3$,
and \eqnref{eq:CqexpandedhighTmore}
gives the high $T$ expansion $t(G;x,y=0)=\left(1+x\right)^{q+1} x^{2q-3}$.
%We repeat the high $T$ approximation given in \eqnref{eq:CqexpandedhighTmore}
%\begin{equation}
%  t(G;x,y=0)=\left(1+x\right)^{q+1} x^{2q-3}.
%  \label{eq:tGxyzero}
%\end{equation}
%
%where $n$ is the number of cliques in the system.
%The partition function for \figref{fig:cliquecircle} is
%\begin{equation}
%  Z\approx qv^{N-1}x(1+x)^q \frac{x^{2q-2}-1}{x-1}.
%  \label{eq:Z3loop}
%\end{equation}
%$N=3n$, and $q$ is the number of possible states.
% new derivation with a slightly different approximation
%Since $x\gg 1$, we approximate $\left(x^{2q-1}-1\right)\approx x^{2q-1}$
%and $\left(x-1\right)^{-1}\approx x+1$.
%\Eqnref{eq:Z3loop} becomes
%\begin{equation}
%  Z \approx qv^{N-1}\left[\sum_{k=0}^{n+1}{n+1\choose k}x^{k}\right]x^{2n}
%  \label{eq:Z3loop3}
%\end{equation}
%\begin{eqnarray}
%  Z & \approx & qv^{N-1}\left[\sum_{k=1}^{n+1}{n+1\choose k}x^{k}\right]x^{2n}\nonumber\\
%  Z & \approx & qv^{N-1}\Bigg[x^{2n} + {n+1\choose 2}x^{2n+1} + {n+1\choose 3}x^{2n+2} \nonumber\\
%    &         & \phantom{-qv^{N-1}} + \cdots +{n+1\choose n}x^{3n} + x^{3n+1}\Bigg].
%  \label{eq:Z3loop3}
%\end{eqnarray}
Substituting $N=3q$ and the approximation $x\approx q/v$ (in either
the large $q$ or large $T$ limits), the partition function becomes
\begin{equation}
  %Z \approx q^{2q-2}v^{q+2}\left[\sum_{k=0}^{q+1}{q+1\choose k}q^{k}v^{-k}\right]
  \mathbf{Z} \approx q^{2q-2}v^{q+2} \left(1+x\right)^{q+1} x^{2q-3}
  \label{eq:Z3loop3}
\end{equation}
%The explicit expansion is
%\begin{eqnarray}
%  Z & \approx & q^{3q}\Bigg[ q^{-q-2}v^{q+1} + {q+1\choose 2}q^{-q-1}v^{q+1} + \cdots \nonumber\\
%    &         & \phantom{q^{3q}} + {q+1\choose q}q^{-2}v^{2} + q^{-1}v \Bigg],
%  \label{eq:Z3loop4}
%\end{eqnarray}
We factor out $q^{3q}$, and then apply the approximations:
$v=\exp(\beta J)-1\approx \beta J$, $x\approx q/v \approx q/(\beta J)\gg 1$,
and $q\gg 1$.
%and $\log (1+x')\approx x'$ for the summation term.
The free energy per site is then
%\begin{eqnarray}
%  f & \approx & \frac{3q}{N}\log q+\frac{1}{N}\Big[q^{-1}v+\cdots +q^{-q}v^{q}\Big] \nonumber\\
%  f & \approx & \log q+\frac{1}{N}\sum_{k=0}^q {q\choose k}q^{-k}v^{k}
%  \label{eq:Z3loop5}
%\end{eqnarray}
%\begin{equation}
%  f \approx \log q + \frac{1}{N}\sum_{k=0}^{q+1} {q+1\choose k}q^{-k+1}v^{k+1}
%  \label{eq:Z3loop5}
%\end{equation}
%In the high $T$ approximation, $v=\exp(\beta J)-1\approx \beta J$,
%and \eqnref{eq:Z3loop5} becomes
%\begin{equation}
%  %f \approx \log q+\frac{1}{N}\sum_{k=0}^{q+1} {q+1\choose k} J^{k+1}
%  %            q^{-(k+1)}\beta^{k+1}.
%  f \approx \log q + \frac{q+1}{3q}\log\left(\frac{q}{v}+2\right)
%            + \frac{2q-3}{3q}\log\left(\frac{q}{v}+1\right).
%  \label{eq:Z3loop6}
%\end{equation}
\begin{equation}
  %f \approx \log q+\frac{1}{N}\sum_{k=0}^{q+1} {q+1\choose k} J^{k+1}
  %            q^{-(k+1)}\beta^{k+1}.
  %f \approx \log q + \frac{1}{3} \log\left( \frac{\beta J}{q} \right)
  %                 + \log\left( \frac{q}{\beta J} \right).
  f \approx -T\log q - \frac{2T}{3} \log\left( \frac{q}{\beta J} \right)
  \label{eq:Z3loop7}
\end{equation}
As in the previous sub-section, the leading $\log q$ term represents
the infinite $T$ limit.
\Eqnref{eq:Z3loop7} affirms the implication of \eqnref{eq:fNIClalt}
regarding the corresponding behavior of large $q$ or $T$.
Specifically, increasing the temperature (decreasing $\beta$) emulates
increasing the number of communities $q$ for a ferromagnetic Potts model.
%which has the same ground state as \eqnref{eq:ourpotts} for a clique system.

\subsection{Free energy of a random graph in a large $q$ or a large $T$ expansion}
\label{sec:f-qandTrandom}

% --- random graph ----------------------------------------------------------
\begin{figure}[t]
\begin{center}
\includegraphics[width=\subfigwidthnew]{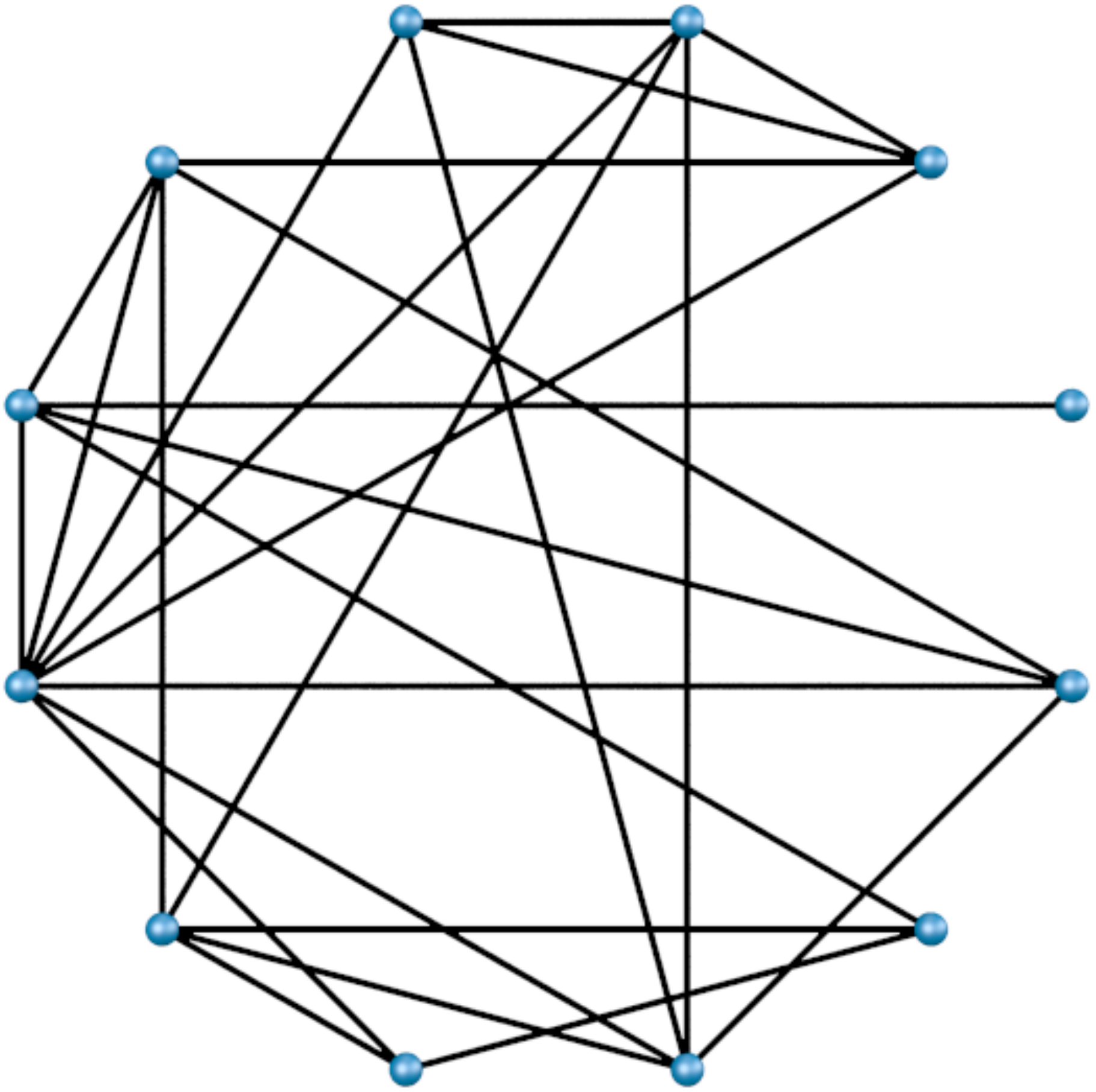}
\end{center}
\caption{(Color online) A sample depiction of a random graph with $N$ nodes.
In \secref{sec:f-qandTrandom}, we analyze such a system by randomly removing
edges from a clique configuration of $N$ nodes under the assumption that we
maintain a connected graph.
%As with each example in \secref{sec:crossoverdisc}, w
We show that a ferromagnetic Potts model on a random graph behaves the same
in the large $q$ and large $T$ limits.}
\label{fig:randomgraph}
\end{figure}
% --- end random graph ------------------------------------------------------

We apply the Tutte polynomial method of \Appref{app:tuttepolyUW} to determine
the high $T$ and high $q$ partition function for a random graph.
For calculation purposes, we begin with a complete graph of size $N$.
Then we randomly remove edges to construct a random graph such that any two
nodes are connected by and edge with a probability $p$.
The derivation repeatedly applies lemma \ref{lem:lemmaone} stated in
\Appref{app:tuttepolyUW}.

We denote the Tutte polynomial of a complete graph (clique) of size $l$
as $\Kl{l}$.
$t(G)$ for a clique with $d$ duplicated edges (multiply defined edges between
two nodes) or loops (self-edges) is defined as $\Kld{l}{d}$.
For economy of notation, we also define $\Klm{l}{m}$ as the Tutte
polynomial of a graph with $m$ \emph{missing edges} (\ie{}, not a clique).
Note that $\Kld{l}{0}\equiv\Klm{l}{0}\equiv\Kl{l}$.
For the following derivation, we work under the \emph{assumption} that
when we delete or contract any edge, the random graph remains connected.

%\begin{assumption}
%  When we delete and contract any edge, the random graph remains connected.
%\end{assumption}

Under the high temperature $T$ or high number of clusters $q$ approximations,
$y\ll x$ and $y\simeq 0$.
\Eqnref{eq:cliquetofKhighT} gives the exact expression of the Tutte polynomial
$K_l(G;x,y=0)$ for a clique at $y=0$.
%For $y=0$, we know the exact expression of the Tutte polynomial
%$K_l(G;x,y=0)$ for a clique.
%\begin{equation}
%  %\K{l}(G;x,y=0) = \frac{\Gamma(x)}{\Gamma(x-l)}
%  \Kl{l}(G;x,y=0) = \Pochhammer{x+l-1}{x}
%  \label{eq:cliquetofKhighT}
%\end{equation}
%We then show the process of calculating the Tutte polynomial for a
%random graph under the assumption and the approximations.
If we cut one edge from the complete graph $\Kl{N}$, we obtain the recursion
formula
\begin{eqnarray}
  \Kl{N} & = & \Klm{N}{1} + \Kld{N-1}{N-1}, \nonumber\\
  \Kl{N} & = & \Klm{N}{1} + \Kl{N-1}.  \label{eq:KlNone}
\end{eqnarray}
where we applied lemma \ref{lem:lemmaone} to obtain \eqnref{eq:KlNone}.
From henceforth, we assume the application of lemma \ref{lem:lemmaone}.
We are interested in the graph with missing edges, so we solve \eqnref{eq:KlNone}
for $\Klm{N}{1}$.
\begin{equation}
  \Klm{N}{1} = \Kl{N} - \Kl{N-1}.
  \label{eq:Klmone}
\end{equation}
Note that the reduced graph is represented as a summation over complete graphs.

Now we apply the Tutte recursion formula to \emph{both sides} of \eqnref{eq:Klmone}.
\begin{equation}
  \Klm{N}{2} + \Klm{N-1}{1} = \Klm{N}{1} + \Kl{N-1} - \Klm{N}{1} -
  \Kl{N-1}.
  \label{eq:Klmtwoworking}
\end{equation}
We can choose the deleted and contracted edges in the corresponding terms to
be identical because the resulting Tutte polynomial is in general
independent of the operation order.
After collecting terms and substituting the previous $\Klm{N}{1}$ result,
we solve for $\Klm{N}{2}$ to obtain
\begin{equation}
  \Klm{N}{2} = \Kl{N} - 2\Kl{N-1} + \Kl{N-2},
  \label{eq:Klmtwo}
\end{equation}
for this particular random graph.
Again, the right-hand-side of \eqnref{eq:Klmtwo} is a summation over complete
graphs.  This recursive relation for $\Klm{N}{k}$ continues until we obtain
\begin{equation}
  \Klm{N}{k} = \sum_{i=0}^k (-1)^i {k\choose i} \Kl{N-i}.
  \label{eq:KlmRR}
\end{equation}
We insert this into \eqnref{eq:cliquetofKhighT} with the pre-factor $qv^{N-1}$
to generate the partition function at high $T$
\begin{equation}
  \mathbf{Z} = q^N \left(\frac{v}{q}\right)^{N-1}
      \left[ \sum_{i=0}^{k} (-1)^i {N\choose i} \Pochhammer{x+N-i-1}{x}
      \right].
  \label{eq:ZrandomHighT}
\end{equation}
%We can write the above equation as the descending order of $x$ with
%$b(k)$ denoting the prefactor of $x^{k}$ for $k=0,1,\ldots ,N$.
%$b(k)$ is very complicated for arbitrary $x^k$, but the highest order
%$x^{N-1}$ is simple
%\begin{eqnarray}
%  b(N-1) & = & a(N)+a(N-1)+\cdots +a(0) \nonumber\\
%  b(N-1) & = & \sum_{k=0}^{N}(1-p)^k{\Lmax\choose k}p^{\Lmax-k}=1.
%\end{eqnarray}
%In terms of a series in $b(k)x^k$, \eqnref{eq:ZrandomHighT} is
%%We are not calculating $b(k)$ for each $x^k$ in the following expression
%\begin{eqnarray}
%  Z & = & qv^{N-1} \bigg[x^{N-1} + b(N-2)x^{N-2} + \cdots \nonumber\\
%    &   & \phantom{qv^{N-1}} + b(1)x + b(0) + b(-1)\bigg]
%\end{eqnarray}
We substitute $x=(q+v)/v\approx q/v$ when $v\ll q$ (high $T$ or high $q$
approximations) and again utilize $v=e^{\beta J}-1\approx \beta J$ in the
high $T$ approximation to obtain the free energy per site
\begin{eqnarray}
  f & = &  -T\log q - \frac{N-1}{N} T\log\left(\frac{\beta J}{q}\right) \nonumber\\
    &   &   - \frac{T}{N}\log \left[ \sum_{i=0}^{k} (-1)^i {N\choose i}
            \Pochhammer{\frac{q}{\beta J}+N-i-1}{\frac{q}{\beta J}} \right].~~~~
  \label{eq:randomf}
\end{eqnarray}
Note that the first two terms become $\log(q)/N \log(\beta J)$ as $N\to\infty$.
From \eqnref{eq:randomf}, we obtain the same conclusion for this random graph
as for the previously analyzed clique systems.
%(Secs.\ \ref{sec:f-q}, \ref{sec:f-T}, and \ref{sec:f-qandT1}).
While Secs.\ \ref{sec:f-q}, \ref{sec:f-T}, and \ref{sec:f-qandT1} result
in free energies with different functional forms, in each case, $q$ and $T$
have the same functional form in the arguments of the functions
in the high $T$ limit.

\subsection{Free energy of an arbitrary graph $G$ in the large $T$ expansion}
\label{sec:f-arbitraryZ}

We can construct the explicit high $T$ expansion for an arbitrary (unweighted)
graph $G$ by means of the Tutte polynomial method \cite{ref:pottstutte}.
Factoring out $q^N$ and substituting $|V|=N$, $x=q/v + 1$, and $y=v+1$
in \eqnref{eq:ZarbitraryG}, we write a trivially modified form
of the partition function
%\begin{equation}
%  Z = q^N\left[ q^{k(G)-N}v^{N-k(G)} t\left(G;\frac{q}{v} + 1, v+1 \right) \right].
%  \label{eq:ZarbitraryGworking}
%\end{equation}
%\Eqnref{eq:ZarbitraryGworking} becomes
\begin{equation}
  \mathbf{Z} = q^N\left[ \left(\frac{v}{q}\right)^{N-k(G)} t\left(G;\frac{q}{v}+1, v+1\right) \right].
  \label{eq:ZarbitraryGworking}
\end{equation}
At this point, the equation is completely general, but the corresponding
behavior for temperature $T$ and number of clusters $q$ is almost apparent
in the reciprocal relationship of $q$ and $v$.

Again, $x\approx q/v$ in either the large $q$ or large $T$ limits.
In a high $T$ approximation, $v\approx \beta J = T/J$ and $y\approx 0$ or $1$
($y=0$ is a common approximation since $x\gg y$ in the same limit).
\begin{equation}
  \mathbf{Z} \approx q^N\left[ \left(\frac{J}{qT}\right)^{N-k(G)}
        t\left(G;\frac{qT}{J}, y_{\phantom{}_{T'}}\right) \right]
  \label{eq:ZarbitraryGHighT}
\end{equation}
where $y_{\phantom{}_{T'}} = 0$ or $1$.
The free energy per site is then
\begin{equation}
  f \approx -T\log q -\frac{N-k(G)}{N}T\log\left(\frac{J}{qT}\right)
      - \frac{T}{N}\log\left[ t\left( \frac{qT}{J} , y_{\phantom{}_{T'}} \right) \right].
  \label{eq:f-TarbitaryG}
\end{equation}
%\begin{eqnarray}
%  f & = & \log q +\frac{N-k(G)}{N}\log\left(\frac{J}{qT}\right) \nonumber\\
%    &   & + \frac{1}{N}\log\left[ t\left( \frac{qT}{J} , y=0~\mathrm{or}~1 \right) \right].
%  \label{eq:f-TarbitaryG}
%\end{eqnarray}
The leading $\log q$ term appears in our previous calculations.
Again, it represents the infinite $T$ limit for an arbitrary system
which is approximately constant in large systems.
%Thus, increasing the temperature emulates increasing the number
%of communities $q$ for a ferromagnetic Potts model.

From the perspective of increasing $q$, the similarity to the large $T$
behavior is more apparent if we fix the temperature $T=T'$ and define
an effective interaction constant $J_q\equiv e^{J/T'}-1$.
We then rewrite \eqnref{eq:f-TarbitaryG} as
\begin{equation}
  f \approx -T\log q -\frac{N-k(G)}{N}T\log\left(\frac{J_q}{q}\right)
      - \frac{T}{N}\log\left[ t\left( \frac{q}{J_q} , y_q \right) \right].
  \label{eq:f-JqarbitraryG}
\end{equation}
%\begin{eqnarray}
%  f & = & \log q +\frac{N-k(G)}{N}\log\left(\frac{J_q}{q}\right) \nonumber\\
%    &   & + \frac{1}{N}\log\left[ t\left( \frac{q}{J_q} , y_{\phantom{}_{q}} \right) \right].
%  \label{eq:f-TarbitaryG}
%\end{eqnarray}
where $y_q\equiv e^{J/T'}$ is a constant.
When $N\to\infty$ and $k(G)\ll N$, the first two terms become
$T\log(q)/N \log(\beta J)$.
Comparing \eqnsref{eq:f-TarbitaryG}{eq:f-JqarbitraryG} shows the close
correspondence between increasing $q$ (at fixed $T'$) and increasing $T$.
$J_q$ grows exponentially faster than $q$ with decreasing $T'$,
so a finite (perhaps small) stable or solvable region is likely
except in the presence of high noise.

\subsection{Annealed versus quenched averages} \label{sec:f-annealedquenched}

The above proofs apply to quenched averages because the binary distribution
is constant with respect to the distribution integration.  That is, using
\eqnref{eq:f-JqarbitraryG}, we assume a probability distribution
$P\left(\{J_{ij}\}\right)$ and integrate over it to obtain the quenched
average free energy per site
\begin{eqnarray}
  f\left[ \{J_{ij}\} \right] & = & \int DJ_{ij} \prod_{i\neq j} P\left(\{J_{ij}\}\right)
          \Bigg\{ \log q \nonumber\\
        & &  +\frac{N-k(G)}{N} \log\left(\frac{J}{qT}\right) \nonumber\\
        & &   + \frac{1}{N}\log\left[ t\left( \frac{qT}{J} , y_{\phantom{}_{T'}} \right) \right] \Bigg\},
  \label{eq:f-Jquenchedavg}
\end{eqnarray}
but the integrand ($f_0$) is a constant because $J$ is independent of $\{J_{ij}\}$,
so the integral trivially simplifies to
\begin{equation}
  f\left[ \{J_{ij}\} \right] =
         f_0 \int DJ_{ij} \prod_{i\neq j} P\left(\{J_{ij}\}\right).
\end{equation}
where the integral is unity.
In a more general model with a defined $\{J_{ij}\}$ probability distribution,
the leading order $\log q$ contribution would remain unchanged, but we would
obtain correction terms from the integration over the quenched interaction
distribution $\{J_{ij}\}$.

\subsection{Free energy of non-interacting cliques for an \emph{arbitrary}
weighted Potts model under a large $T$ expansion}  \label{sec:f-wPottsT}

We can represent an \emph{arbitrary} weighted Potts model with ferromagnetic
and antiferromagnetic interactions.
That is, we can generally write
\begin{equation}
  H(\{\sigma\}) = -\frac{1}{2}\sum_{i\neq j}\left[ a_{ij}A_{ij}
                  - b_{ij}\left(1-A_{ij}\right)\right]
                  \delta(\sigma_i,\sigma_j).
  \label{eq:genpotts}
\end{equation}
where $a_{ij}$ and $b_{ij}$ are arbitrary ``attractive'' and ``repulsive''
edge weights.
This summarization includes modularity \cite{ref:gn},
a Potts model incorporating a ``configuration null model'' (CMPM) comparison
\cite{ref:smcd} (the most common variation in \cite{ref:smcd} is effectively
generalizes modularity),
CMPM allowing antiferromagnetic relations \cite{ref:traagPRE},
``label propagation'' \cite{ref:LPA,ref:barberLPA},
an Erd{\H o}s-R{\' e}nyi Potts model \cite{ref:reichardt,ref:smcd},
a ``constant Potts model'' \cite{ref:traaglocalscope},
the weighted form of the APM \cite{ref:rzmultires,ref:rzlocal},
or a ``variable topology Potts model'' suggested in \cite{ref:rzmultires}.

Note that the repulsive weights $b_{ij}$ are important in that they provide
a ``penalty function'' which enables a well-defined ground state for the
Hamiltonian for an arbitrary graph.
That is, the ground state of a purely ferromagnetic Potts model in an arbitrary
graph is trivially a fully collapsed system (perhaps with disjoint sub-graphs).
%Save for modularity and the initial label propagation algorithm, m
Several of the above models incorporate a weighting factor $\gamma$ of some
type on the penalty term which allows the model to span different scales of the
network in qualitatively similar ways.

We denote a the partition function of a graph $G^*$ with $l$ nodes and weighted
edges $\{e\}$ by $Z(G^*;q,\mathbf{v})\equiv\mathcal{K}_l$.
%We assume that $G^*$ is connected.  % is this necessary???
We assume that $J_e\ll T$ for all edges $e$,
and all pairs of nodes in $G^*$ are connected by a weighted edge $J_e$
(either ferromagnetic or antiferromagnetic).
From \Appref{app:tuttepolyW}, a recurrence relation for the multivariate
Tutte polynomial of a general weighted clique is
\begin{equation}
  \mathcal{K}_l \approx \left(q + \sum_{k=1}^{l-1} v_k \right)
                         \mathcal{K}_{l-1} + O(y_e),
  \label{eq:KlcliqueWrr}
\end{equation}
The partition function for $\mathcal{K}_l$ at high $T$ is
\begin{equation}
  \mathcal{K}_l \approx q^N \prod_{j=2}^{l} \left( 1+\sum_{k=1}^{j-1} \frac{v_k}{q}
  \right),
  \label{eq:KlcliqueW}
\end{equation}
Now, we generate a graph consisting of a set of $q$ non-interacting cliques
of size $l_i$ where $i=1,2,\ldots,q$.
\begin{equation}
  \mathcal{K}_l \approx q^N \prod_{i=1}^{q} \prod_{j=2}^{l_i}
                          \left(1 + \sum_{k=1}^{j-1} \frac{\beta J_k}{q}
                          \right).
  \label{eq:KlcliquesystemW}
\end{equation}
where we used $v_e\approx\beta J_e$ at high $T$ for general edge weights $J_e$
(even if $J_e<0$ as long as $J_e\ll T$).

The free energy is
\begin{eqnarray}
  f & \approx & -T\log q - \frac{T}{N}\sum_{i=1}^{q} \sum_{j=2}^{l_i} \sum_{k=1}^{j-1}
                  \frac{\beta J_k}{q} \label{eq:farbitraryGWsssumT}         \nonumber\\
    & \approx & -T\log q - \frac{1}{N}\sum_{i=1}^{q}\frac{E_i}{q}         \nonumber\\
   & = & -T\log q - \frac{E}{qN} \label{eq:farbitraryGWET}
\end{eqnarray}
where we invoked $\log(1+x) \approx x$ for $x\ll 1$ there. $E_i$ is
the energy of cluster $i$ according to the weighted Potts model of
\eqnref{eq:genpotts}, and $E$ is the total energy of the graph.
\Eqnsref{eq:KlcliquesystemW}{eq:farbitraryGWET} both imply that
large $q$ emulates large $T$ for an \emph{arbitrary} Potts model on
a weighted graph $G$. That is, if a community detection quality
function can be expressed in terms of the general Potts model in
\eqnref{eq:genpotts}, then large $q$ and large $T$ are essentially
equivalent.

\subsection{Free energy of non-interacting cliques for an \emph{arbitrary}
weighted Potts model under a large $q$ expansion}  \label{sec:f-wPottsq}

The multivariate Tutte polynomial \cite{ref:jacksonmvtutte} (see
also \Appref{app:tuttepolyW} and Ref.\ \cite{ref:rhzglobaldisorder})
appears in a subgraph expansion over the subset of edges
$\mathcal{A}\subseteq \mathcal{E}$ in a graph $G = (V,\mathcal{E})$
with a set of $V$ vertices and $\mathcal{E}$ edges
\begin{equation}
  %Z(G;q,\mathbf{v}) = q^{|V|}\sum_{\mathcal{A} \subseteq \mathcal{E}}
  %                              q^{k(\mathcal{A}) -|V| + |\mathcal{A}|}
  %                           \prod_{e' \subseteq \mathcal{A}} \frac{v_{e'}}{q}.
  %Z(G;q,\mathbf{v}) = \sum_{\mathcal{A} \subseteq \mathcal{E}}
  %                      q^{k(\mathcal{A})} \prod_{e' \subseteq \mathcal{A}} v_{e'}.
  Z(G;q,\mathbf{v}) =
          q^N \left[ \left(1 + \sum_{e'=1}^{|\mathcal{E}|} \frac{v_{e'}}{q} \right)
           + \cdots + q^{k(G)-N} \prod_{f'=1}^{|\mathcal{E}|} v_{f'}
          \right]
  \label{eq:ZarbitraryGq}
\end{equation}
$k(\mathcal{A})$ is the number of connected components of $G_A =
(V,\mathcal{A})$ and $v_e = \exp(\beta J_e) -1$. For our purposes,
\eqnref{eq:ZarbitraryGq} serves as an alternate representation of
$Z_G$ to facilitate the calculation of the large $q$ expansion.
%$v_e = \exp(\beta J_e) -1$, so it is easy to see that the expansion
%in large $q$ is comparable to a high $T$ expansion.

For large $q$, when $q^N\gg |v_{e}|^L$, the last term may neglect,
and for a system of non-interacting cliques of size $l_i$ with
$i=1,2,\ldots,q$, the leading order terms in large $q$ are
\begin{equation}
  Z(G;q,\mathbf{v}) \approx q^N \prod_{i=1}^q
       \prod_{j=2}^{l_i} \left(1 + \sum_{k=1}^{j-1} \frac{v_k}{q} \right).
  \label{eq:ZarbitraryGlargeq}
\end{equation}
The approximation is identical to \eqnref{eq:KlcliqueW} at high $T$.
% which was calculated
%by applying the deletion and contraction recursion relation at high $T$
%to a clique of size $l$ and taking the product of the partition functions
%for the disjoint cliques.
Ref.\ \cite{ref:rhzglobaldisorder} calculates an explicit crossover
temperature including the last subgraph $\mathcal{A}=\mathcal{E}$
that competes with the large $q$ terms as $T\to 0$.
The free energy corresponding to \eqnref{eq:ZarbitraryGlargeq} becomes
\begin{eqnarray}
  f & \approx &  -T\log q - \frac{T}{N}\sum_{i=1}^{q} \sum_{j=2}^{l_i} \sum_{k=1}^{j-1}
                     \frac{v_k}{q} \label{eq:farbitraryGWsssumq}
\end{eqnarray}
where we applied the small $x$ approximation $\log(1+x)\approx x$.

%Analogous to \eqnref{eq:farbitraryGWsssumT}, large $q$ emulates large $T$.
In order to illustrate the correspondence in large $q$ and $T$, we fix $T=T'$,
define $J^{(q)}_e\equiv\exp(\beta' J_e)-1$, and rewrite the free energy per site
\begin{eqnarray}
  f & \approx &  -T'\log q - \frac{T'}{N}\sum_{i=1}^{q} \sum_{j=2}^{l_i} \sum_{k=1}^{j-1}
                     \frac{J_k^{(q)}}{q} \label{eq:farbitraryGWsssumqJdef}.
\end{eqnarray}
Large $q$ in \eqnref{eq:ZarbitraryGq} emulates large $T$
in \eqnref{eq:KlcliquesystemW}.
As with the unweighted case in \eqnref{eq:f-JqarbitraryG} in \secref{sec:f-arbitraryZ},
$J_e^{(q)}$ is exponentially weighted in $\beta' = 1/T'$, so a non-zero
(perhaps small) region of stability is essentially ensured except in the
presence of high noise \cite{ref:rhzglobaldisorder}.
%Other results confirm \cite{ref:rhzglobaldisorder} the crossover behavior, and
We can additionally determine a rigorous bound using methods
in \cite{ref:batistaDreduc,ref:nussinovholographies,ref:rhzglobaldisorder}
\begin{equation}
  T_\times^\mathrm{UB} = \frac{\bar{J_0}}{\log\left[ \frac{p(q-1)}{(1-p)}
  \right]},
\end{equation}
where $\bar{J_0}=\frac{1}{2}\sum_j J_{j0}\left[ 1+\mbox{sgn}(J_{j0}) \right]$
is a generous upper bound summing only positive energy contributions
and $p$ is the probability for finding a given spin $\sigma_0$ in a specific
spin state $\bar{\sigma}$.
This result further agrees with our conclusions.
Note that as $p\to 1/q$, the system is completely disordered, so $T_\times\to\infty$.
As $p\to 1$, the system is perfectly ordered, so $T\to 0$.

\section{Conclusions}

We systematically examined the phase transitions for the community
detection problem via a ``noise test'' across a range of parameters.
The noise test consists of a structured graph with a
strongly-defined ground state. We add increasing numbers of
extraneous intercommunity edges (noise) and test the performance of
a stochastic community detection algorithm in solving for the
well-defined ground state. Specifically, we studied two types
(sequences) of systems. In the first such sequence of systems in
\figref{fig:susAllalpha}, we fixed the ratio $\alpha=q/N$ of the
number of communities $q$ to the number of nodes $N$. We fixed $q$
at different values and varied $N$ in the second sequence of systems
in \figref{fig:susAllq}. In \figref{fig:EHqAll}, we explored the
largest tested systems with $N=2048$ nodes in more detail where we
depicted additional measures to illustrate the transitions. All of
these systems showed regions with distinct phase transitions in the
large $N$ limit. Deviations occurred most often in smaller systems
indicating a definite finite-size effect.

The spin-glass-type phase transitions in our noise test occurred
between solvable and unsolvable regions of the community detection
problem. A hard, but solvable, region lies at the transition itself
where it is difficult, in general, for any community detection
algorithm to obtain the correct solution. We analyzed a system of
non-interacting cliques and illustrated that in the large $q$ limit,
the system experiences a thermal disorder in the thermodynamic limit
for any non-zero temperature. When in contact with a heat bath, the
asymptotic behavior of the temperatures beyond which the system is
permanently disordered varies slowly with the number of communities
$q$, specifically, $\Tcross\simeq O[1/\log q]$. This implies that
problems of practical size maintain a definite region of
solvability. Given the connection between Jones polynomials of knot
theory and Tutte polynomials for the Potts model, our results imply
similar transitions in large random knots (see
\Appref{app:trefoilknot}).

We further studied the free energy of arbitrary graphs arriving at the same
conclusion.
Increasing number of communities $q$ emulates increasing $T$ in arbitrary graphs
for a general Potts model.
The effective interaction strength for increasing $q$ scales such that this disorder
is circumvented by the often standard use of a simulated annealing algorithm, but
the ``glassy'' (high noise) region remains a challenge for any community detection
algorithm.

\section*{Acknowledgments}

This work was supported by NSF grant DMR-1106293 (ZN).
We also wish to thank S. Chakrabarty, R. Darst, P. Johnson,
V. Dobrosavljevic, B. Leonard, A. Middleton, M. E. J. Newman,
D. Reichman, V. Tran, and L. Zdeborov{\'a} for discussions and
ongoing work.

\appendix

\section{Definitions: Trials and Replicas}\label{app:replica}

We review the notion of trials and replicas on which our algorithms
are based. Both pertain to the use of multiple identical copies of
the same system which differ from one another by a permutation of
the site indices. Thus, whenever the time evolution may depend on
sequentially ordered searches for energy lowering moves (as it will
in our greedy algorithm),  these copies may generally reach
different final candidate solutions. By the use of an {\em ensemble}
of such identical copies (see, e.g., \figref{fig:MRAlandscape}), we
can attain accurate result as well as determine information theory
correlations between candidate solutions and infer from these a
detailed picture of the system.

In the definitions of ``trials'' and ``replicas'' given below, we
build on the existence of a given algorithm (any algorithm) that may
minimize a given energy or cost function. In our particular case, we
minimize the Hamiltonian of Eq.(\ref{eq:ourpotts}.
\bigskip

$\bullet$ {\underline{{\em Trials}.} We use trials alone in our bare
community detection algorithm. We run the algorithm on the same
problem $t$ independent times. This may generally lead to different
contending states that minimize Eq.(\ref{eq:ourpotts}). Out of these
$t$ trials, we will pick the lowest energy state and use that state
as the solution.
\bigskip

$\bullet$ {\underline{{\em Replicas}.}  We use both trials and
replicas in our multi-scale community detection algorithm. Each
sequence of the above described $t$ trials is termed a {\em
replica}.  When using ``replicas'' in the current context, we run
the aforementioned $t$ trials (and pick the lowest solution) $r$
independent times. By examining information theory correlations
between the $r$ {\em replicas} we infer which features of the
contending solutions are well agreed on (and thus are likely to be
correct) and on which features there is a large variance between the
disparate contending solutions that may generally mark important
physical boundaries. We will compute the information theory
correlations within the ensemble of $r$ replicas. Specifically, {\em
information theory extrema} as a function of the scale parameters,
generally correspond to more pertinent solutions that are locally
stable to a continuous change of scale.  It is in this way that we
will detect the important physical scales in the system
(\figref{fig:MRAlandscape}).
\bigskip

\section{Information theory and complexity measures}
\label{app:information}

We use information theory measures to calculate correlations between
community detection solutions and expected partitions in the noise
test problem. To begin, $N$ nodes of partition $A$ are partitioned
into $q_A$ communities of size $\{n_a\}$ where $1 \leq a \leq q_A$.
The ratio $n_a/N$ is the probability that a randomly selected node
is found in community $a$. The Shannon entropy is
\begin{equation}
  H_A = -\sum_{a=1}^{q_A} \frac{n_a}{N}\log_2\frac{n_a}{N}
  \label{eq:HA}
\end{equation}
The mutual information $I(A,B)$ between partitions $A$ and $B$ is
\begin{equation}
  I(A,B)=\sum_{a=1}^{q_A}\sum_{b=1}^{q_B}\frac{n_{ab}}{N}\log_2\frac{n_{ab}N}{n_an_b}
  \label{eq:IAB}
\end{equation}
where $n_{ab}$ is the number of nodes of community $a$ in partition
$A$ that are also found in community $b$ of partition $B$.
%
%The variation of information $V(A,B)$ between two partitions $A$ and $B$ is
%\begin{equation}
%  V(A,B)=H_A+H_B-2I(A,B)
%\end{equation}
%which has a range of $0\leq V(A,B)\leq\log_2 N$.
The normalized mutual information $I_N(A,B)$ is then
\begin{equation}
  I_N(A,B)=\frac{2I(A,B)}{H_A+H_B}.
\end{equation}
with the obvious range of $0\leq I_N(A,B)\leq 1$.
High $I_N$ %and low $V$
values indicate better agreement between compared partitions.

\section{Computational susceptibility}  \label{app:chi}

The complexity $\Sigma(e)$ of the energy landscape is related to the
number of states ${\cal{N}}(E) \sim \exp[N\Sigma(e)]$
\cite{ref:mezardpz} with energy $E$ and energy density $e = E/N$. In
the current analysis, we detect the onset of the high complexity
with no prior assumptions or approximations by computing a
``computational susceptibility'' \cite{ref:rzmultires} defined as
\begin{equation}
  \chi_n=I_N(s=n)-I_N(s=4).
   \label{eq:sus}
\end{equation}
That is, $\chi$ measures the increase in the normalized mutual
information $I_N$ as the number of trials (number of independently
solved starting points in the energy landscape) $s=n$ is increased.
Physically, we evaluate how many different optimization trials are
necessary to achieve a desired accuracy threshold.

$\chi$ evaluates the expected response of the system to additional
optimization effort. That is, a higher $\chi$ indicates that
additional optimization effort will likely result in a better
solution. A low value of $\chi$ indicates that there will be less
improvement from the additional effort whether due to a trivially
solvable system, a complex energy landscape with numerous local
minima that trap the solver (at low to moderate temperatures), or
thermal-oriented effects of randomly wandering the energy landscape.

\section{Heat Bath Algorithm} \label{app:HBA}

We extend the greedy algorithm in \cite{ref:rzlocal,ref:rzmultires}
to non-zero temperatures by applying a heat bath algorithm. After,
we connect the system to a large thermal reservoir at a constant
temperature T, the probability for a particular node to move from
community $a$ to $b$ is set by a thermal distribution
\cite{ref:smcd},
\begin{equation}
  p_{a\to b}=\frac{\exp(-\Delta E_{a\to b}/T)}{\sum_d \exp(-\Delta
  E_{a\to d}/T)}.
  \label{eq:probability}
\end{equation}
$\Delta E_{a\rightarrow b}$ is the energy change that results if the
node is moved to the new community $b$, and the index $d$ runs over
all connected clusters including its current community or a new
empty community. The steps of our heat bath algorithm are as
follows:

($1$) \emph{Initialize the system}. Initialize the network into a
``symmetric'' state by assigning each node as the lone member of its
own community (\ie{}, $q_0=N$).

($2$) \emph{Find the best cluster for node $i$}. Select a node and
determine to which clusters it is connected (including its current
community and an empty cluster). Calculate the energy change $\Delta
E_{a\to b}$ required to move to each connected cluster $b$.
Calculate and sum all Boltzmann weights. Generate a random number
between $0$ and $1$ and determine into which cluster the node is
placed.

($3$) \emph{Iterate over all nodes}. Repeat step $2$ in sequence for
each node.

($4$) \emph{Merge clusters}. Allow for the merger of community pairs
based on the same Boltzmann-weighted merge probabilities.

($5$) \emph{Repeat the above two steps}. Repeat steps $2$ through
$4$ until the maximum number of iterations is reached.

($6$) \emph{Repeat all the above steps for s trials}. Repeat steps
$1$--$5$ for $s$ trials and select the lowest energy trial as the
best solution. Each trial randomly permutes the order of nodes in
the initial state.

This HBA is similar to our greedy algorithm except that we use a
random process to select the node moves in steps (2) and (4). The
results obtained at low temperature by our HBA are very close to the
results obtained by the zero temperature greedy algorithms. Note
that there is no cooling scheme as occurs in SA, so step $5$ ends at
a maximum number of iterations as opposed to a unchanged best
partition that is achieved as $T\to 0$ in SA.

In the easy phase, different starting trajectories, each beginning
in the symmetric initial state, but they often lead to the same
solution. In the hard phase, changing the random seed may
significantly alter the final result of an individual trial because
the solver becomes trapped in different local minima. Thus we apply
additional trials in order to sample different regions of the energy
landscape and arrive at better solutions. In the unsolvable phase,
increasing the number of trials $s$ does not substantially change
the quality of the solutions unless one happens to sample the energy
landscape in the immediate vicinity near the optimal partition, but
the probability of doing so is small with a finite number of trials
$s$.

%\subsection{Simulated annealing comparison}

%We illustrate the effect of a simulated annealing algorithm \figref{fig:SAplot}

% --- Clique Tutte derivation ---------------------------------------------------------
% moved here for float placement
\begin{figure*}[t!]
\begin{center}
\subfigure[]{\raisebox{0.775in}{\includegraphics[width=\subfigwidthnew]{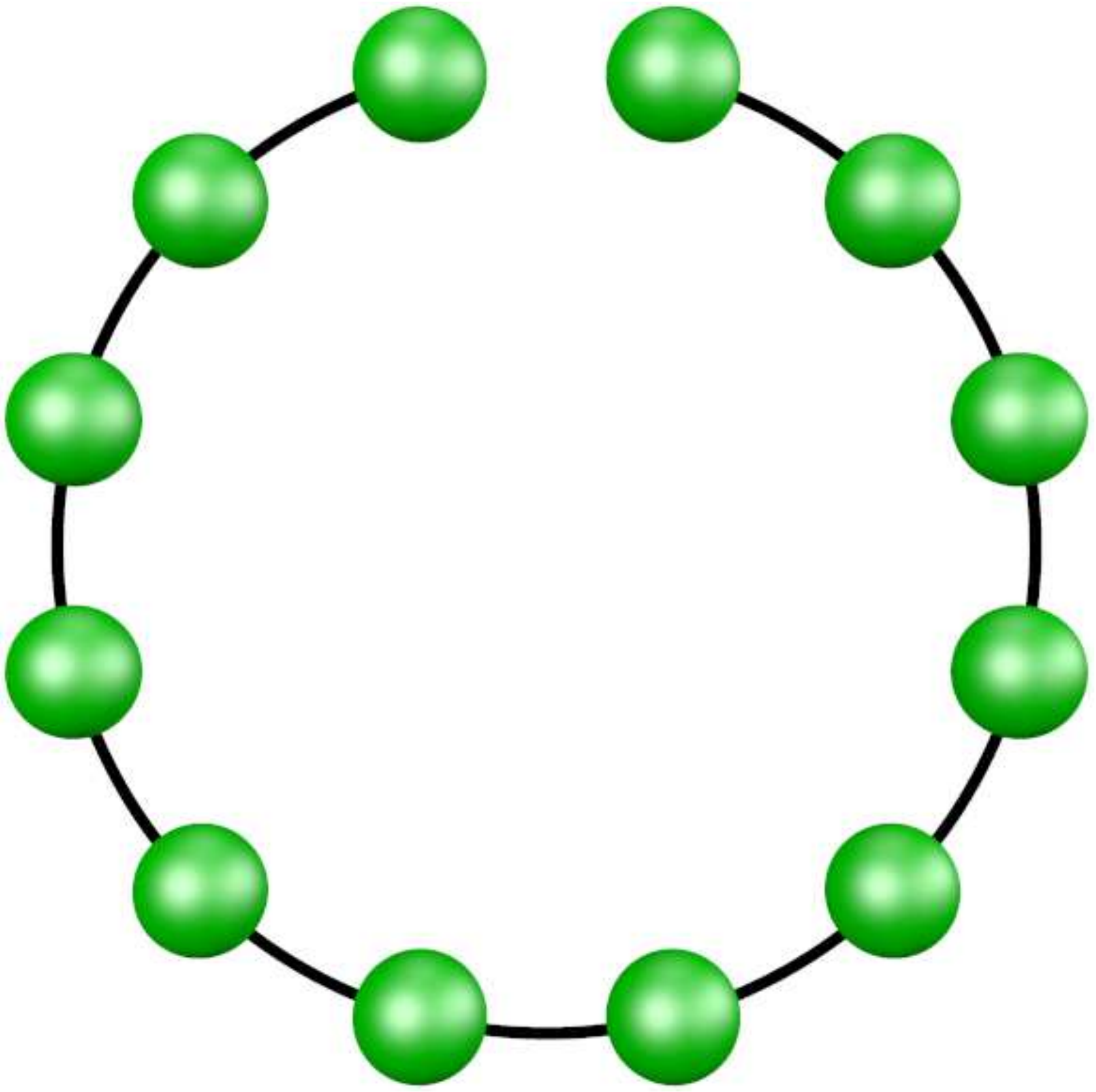}}}
\subfigure[]{\raisebox{0.025in}{\includegraphics[width=1.2\columnwidth]{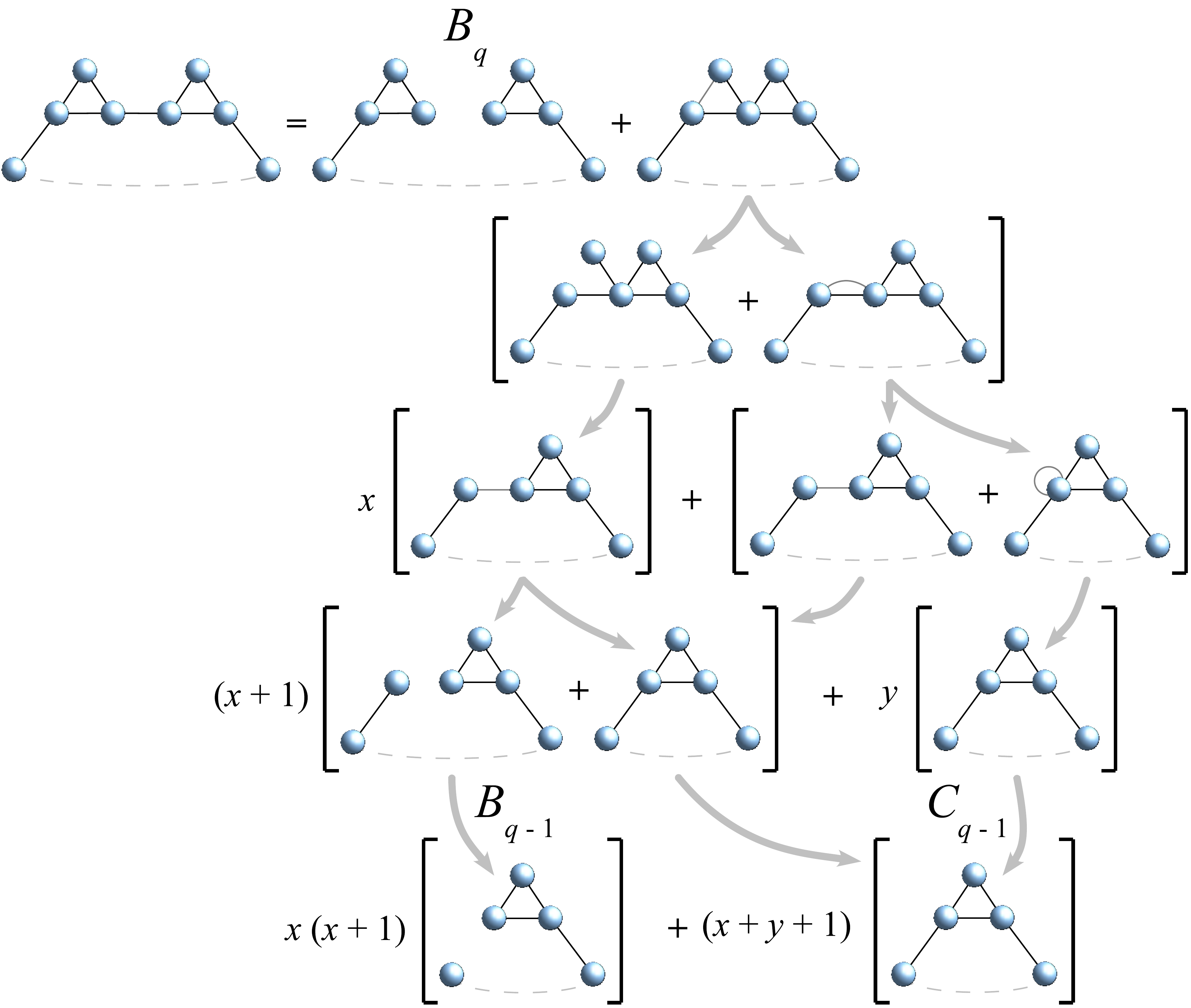}}}
\end{center}
\caption{(Color online) In panel (a), we depict a chain $B_q$ of $q$ cliques
(complete sub-graphs of maximally connected clusters) of size $l$ connected
by single edges.
The corresponding circle of $q$ cliques $C_q$ is depicted in \figref{fig:cliquecircle}.
In panel (b), we show the derivation of the Tutte polynomial in \eqnref{eq:Cq}
for size $l=3$ cliques.
We iteratively break edges and merge nodes according to the Tutte polynomial recurrence
relation \cite{ref:pottstutte} in \Appref{app:tuttepolyUW} until we arrive at configurations
that are reduced clique circle components.
For presentation purposes, gray edges are cut in the next line of the derivation.
The dashed gray line at the bottom of each sub-graph represents the remainder of the clique
circle which is not touched or affected by the operations on the displayed subgraph.}
\label{fig:CCTD}
\end{figure*}
% --- end clique Tutte derivation -----------------------------------------------------

\section{Tutte polynomials} \label{app:tuttepoly}

We give a very brief introduction to Tutte polynomials consisting of the essential
facts necessary for the derivations presented in this paper.
The notation used here is mostly standard, but the notation elsewhere in
the text deviates from standard notation in order to facilitate the partition
function derivation in \secref{sec:f-qandTrandom}.
For an undirected graph $G$, we denote  the \emph{deletion} (removal) of an edge
$e$ by $G'$ and a \emph{contraction} of the edge by $G''$ where a contraction
consists of removing the edge $e$ and merging the corresponding vertices.

\subsection{Unweighted graph $G$} \label{app:tuttepolyUW}

If $G$ has no edges, the Tutte polynomial is $t(G;x,y)=1$.
If $G$ is a disjoint graph of partitions, then $A$ and $B$ $t(G;x,y)=t(A;x,y)~\! t(B;x,y)$.
When an edge $e$ in an unweighted graph $G$ is ``\emph{cut},'' the recurrence
relations are \cite{ref:pottstutte}:
\begin{itemize}
  \item For a general edge, $t(G;x,y)=t(G_e';x,y)+t(G_e'';x,y)$
        which is the sum of two graphs where $e$ is deleted and contracted.
  \item If edge $e$ is an isthmus between two otherwise disconnected regions
        of $G$, then $t(G;x,y) = x~\! t(G_e'';x,y)$ where the edge $e$ is contracted.
  \item If edge $e$ is a loop (a vertex self-edge), then $t(G;x,y) = y~\! t(G_e';x,y)$
        where the edge $e$ is deleted.
\end{itemize}
The resulting Tutte polynomial is a function of two variables $(x,y)$, and it
is independent of the construction order.
Different graphs $G$ and $H$ may be described by the same function $t(G;x,y)=t(H;x,y)$.
A sample calculation is performed \Appref{app:cliquecircle} for a circle of complete
sub-graphs (cliques) as shown in \subfigref{fig:CCTD}{b}.

Tutte polynomials are related to the partition function of a ferromagnetic ($J>0$)
or antiferromagnetic ($J<0$) Potts model given by
\begin{equation}
  H(\{\sigma\}) = -\sum_{i\neq j} J\delta(\sigma_i,\sigma_j)
  \label{eq:tuttepottsUW}
\end{equation}
for any connected pair of nodes $i$ and $j$ with an interaction strength $J$.
The corresponding partition function is
\begin{equation}
  Z = q^{k(G)} v^{|V|-k(G)} t(G;x,y)
  \label{eq:ZTPApp}
\end{equation}
where $q$ is the number of clusters or states,
$v=\exp(\beta J)-1$, $G$ denotes the graph, $k(G)$ is the
number of connected components in $G$, $|V|$ is the number of vertices,
%$t(G;x,y)$ is the Tutte polynomial expression for graph $G$,
$x=(q+v)/v$ and $y=v+1$.

%\subsection*{Appendix F:  Tutte polynomial lemmas}

In \secref{sec:f-qandTrandom}, we use the following lemma to derive
high temperature $T$ approximation for a constructed random graph.
We denote $K_l$ as the Tutte polynomial for a complete graph, and
$K_l^{(d)}$ denotes that the graph has $d$ duplicated (possibly
redundant) edges.
%We also denote $G$ as a graph, $G'$ is the graph with a \emph{deleted}
%edge, and $G''$ is a graph with a \emph{contracted} edge.

\begin{lemma}
  For a clique $K_l^{(d)}$ of size $l$ with $d$ duplicate edges between
  any pair of nodes, the Tutte polynomial at $y=0$ is $K_l$.
  %$K_l^{(d=0)}$\equiv K_l$.
  \label{lem:lemmaone}
\end{lemma}

\begin{proof}
  Let $G$ be a complete graph with $l$ vertices and $d=1$ redundant
  edge.  If we delete and contract the duplicate edge, the Tutte
  polynomial $t(G)\equiv K_l^{(d=1)}$ is
  \begin{equation*}
    K_l^{(1)} = K_l + K_{l-1}^{(l-1)}
    %\label{eq:}
  \end{equation*}
  The contracted vertex in the second term contains $r=1$ loop.
  We cut the loop and have
  \begin{eqnarray}
    K_l^{(1)} & = & K_l + y K_{l-1}^{(l-2)} \nonumber\\
    K_l^{(1)} & = & K_l \label{eq:lemmaoneoneproof}
    %& & \qedhere\nonumber
  \end{eqnarray}
  where we used $y=0$ in the second line.

  Now, assume that we can reduce $K_l^{(d)} = K_l$.
  Let $G$ be a complete graph with $l$ vertices and $d+1$ duplicate edges.
  If we cut one duplicate edge, the resulting Tutte polynomial
  $t(G)\equiv K_l^{(d+1)}$ is
  \begin{equation*}
     K_l^{(d+1)} = K_l^{(d)} + K_{l-1}^{(d+l-1)}
    %\label{eq:}
  \end{equation*}
  %because the contracted edge merges
  The contracted vertex in the second term contains $r\ge 1$ loops.
  We cut each loop in sequence and obtain
  \begin{eqnarray}
    K_l^{(d+1)} & = & K_l^{(d)} + y^r K_{l-1}^{(d+l-r-1)}. \nonumber\\
    K_l^{(d+1)} & = & K_l^{(d)}
    \label{eq:lemmaonedproof}
  \end{eqnarray}
  Since $K_l^{(d)}=K_l$, we also equate $K_l^{(d+1)}=K_l$ by \eqnref{eq:lemmaonedproof}.
  \Eqnref{eq:lemmaoneoneproof} shows that the relation holds for $d=1$; therefore,
  by mathematical induction $K_l^{(d)}=K_l$ holds true for any integer $d\ge 1$.\qedhere
\end{proof}

\subsection{Weighted graph $G$} \label{app:tuttepolyW}

An excellent summary of multivariate Tutte polynomials (MVTP) is found
in Ref.\ \cite{ref:jacksonmvtutte}.
The MVTP allows for arbitrary weights
$\mathbf{v} = [v_e]$ for the edges $\{e\}$ of $G$.
If $G$ has no edges, the MVTP is $Z(G;q,\mathbf{v})=q$.
For an undirected graph $G$, the weighted Potts Hamiltonian is
\begin{equation}
  H(\{\sigma\}) = -\sum_{i\neq j} J_{ij}\delta(\sigma_i,\sigma_j).
  \label{eq:tuttepottsW}
\end{equation}
When an edge $e$ in $G$ is ``\emph{cut},'' the recurrence relation is
\begin{equation}
  Z(G;q,\mathbf{v}) = Z(G';q,\mathbf{v}) + v_e Z(G'';q,\mathbf{v})
  \label{eq:mvtutte}
\end{equation}
where $J_e$ corresponds to the edge weight between two nodes $i$ and $j$
and $v_e = \exp{\beta J_e}-1$.

As with the unweighted case, if $G$ is a disjoint graph of partitions $A$ and $B$,
then $Z(G;x,y)=Z(A;q,\mathbf{v})~\!Z(B;q,\mathbf{v})$.
If partitions $A$ and $B$ are joined at a single vertex, then
then $Z(G;x,y)=Z(A;q,\mathbf{v})~\!Z(B;q,\mathbf{v})/q$.
Unlike \eqnref{eq:ZTPApp} for unweighted graphs, \eqnref{eq:mvtutte} holds for loops
or bridges, but for concreteness, cutting an isthmus $e$ yields
\begin{eqnarray}
  Z(G;q,\mathbf{v}) & = & \left(1 + v_e/q\right) Z(G_e';x,y) \\
  Z(G;q,\mathbf{v}) & = & \left(q + v_e\right) Z(G_e'';x,y)
  \label{eq:mvtpbridge}
\end{eqnarray}
where $e$ is deleted or contracted, respectively.
If $e$ is a loop, then
\begin{equation}
  Z(G;q,\mathbf{v}) = (1+v_e) Z(G_e';x,y).
  \label{eq:mvtploop}
\end{equation}
Note that the MVTP \emph{is} the partition function.
That is, there are no prefactors of $q$ or $v_e$.
%We could define variables $x_e$ or $y_e$  corresponding $x$
%and $y$ in the unweighted case, but this is not necessary for
%the derivations in \secref{sec:f-wPottsqandT}.
Finally, if two parallel edges connect the same pair of nodes $i$ and $j$
with weights $J_1$ and $J_2$, then $Z_G$ is unchanged if we replace
the parallel edges by a single edge with a weight $J' = J_1 + J_2$
(this negates the need for lemma \ref{lem:lemmaone} above).

\section{Derivation of the Tutte polynimial for a circle of cliques}
\label{app:cliquecircle}

As depicted in \figref{fig:cliquecircle}, we define $\Cq{q}$ as a circle of $q$
cliques where we focus those of size $l=3$ for the current derivation.
The Tutte polynomial for a triangle is $\tri\equiv\trixy$.
For convenience, we also define, $\trip\equiv(\tri + x + 1)=[(x+1)^2 + y]$
and $y'\equiv(x + y + 1)$.

We define $\Bq{q}$ to be the Tutte polynomial for a clique chain
as depicted in \subfigref{fig:CCTD}{a}.
In this case, it is trivial to construct $\Bq{q}$ %for the clique chain
\begin{equation}
  \Bq{q} = x^{q-1} \trixy^q.
  \label{eq:Bq}
\end{equation}
With \eqnref{eq:Bq}, we construct a recurrence relation for the clique
circle configurations as shown in \subfigref{fig:CCTD}{b}
\begin{equation}
  \Cqq = \Bqq + x\xpo\Bq{q-1} + \xpoy\Cq{q-1}.
  \label{eq:Cq}
\end{equation}
From this relation, we can sum the series exactly.
%We show a few iterations explicitly,
\begin{eqnarray}
  \Cqq %& = & \Bqq + x\xpo\Bq{q-1} + \xpoy\Cq{q-1} \nonumber\\
       %& = & \Bqq + x\xpo\Bq{q-1} + \xpoy \nonumber\\
       %&   &   \times\Big[ \Bq{q-1} + x\xpo\Bq{q-2} + \xpoy\Cq{q-2} \Big] \nonumber\\
       %& = & \Bqq + (\tri + x + 1)\Bq{q-1} + \xpo\tri\Bq{q-2}  \nonumber\\
       %&   &   + \xpoyp{2}\Cq{q-2} \nonumber\\
  %\Cqq
        & = & \Bqq + \trip\Bq{q-1} + x\xpo\xpoy\Bq{q-2}  \nonumber\\
        &   & +~\xpoyp{2}\Cq{q-2}  \nonumber\\
  %     & = & \Bqq + \trip\Bq{q-1} + x\xpo\xpoy\Bq{q-2}  + \xpoyp{2}\Big[ \Bq{q-2}
  %             + \tri\Bq{q-3} + \xpoy\Cq{q-3} \Big] \nonumber\\
  %\Cqq & = & \Bqq + \trip\Bq{q-1} + \trip\xpoy\Bq{q-2} + \tri\xpoyp{2}\Bq{q-3}
  %             + \xpoyp{3}\Cq{q-3} \nonumber\\
  %\Cqq & = & \Bqq + \trip\Bq{q-1} + \trip\xpoy\Bq{q-2} + \trip\xpoyp{2}\Bq{q-3} +\tri\xpoyp{3}\Bq{q-4}
  %           + \xpoyp{4}\Cq{q-4} \nonumber\\
    \vdots &   & \phantom{\Bqq + \trip\Bq{q-1} } \vdots \nonumber\\
  \Cqq & = & \Bqq + \trip\sum_{i=0}^{q-4}\xpoyp{i}\Bq{q-i-1}  \nonumber\\
       &   &   +~\tri\xpoyp{q-3}\Bq{2} + \xpoyp{q-2}\Cq{2}.~~~~~ %spaces are because of number overlap
  \label{eq:Cqrr}
\end{eqnarray}
Note that the last $\Bq{j}$ term uses $\tri$ not $\trip$.
Also, it can be shown that $\Cq{2} = \xpop{2}\big( x^3 + \tri \big) + y\xpo\tri$.
Substituting these values into the equation, we arrive at
%\begin{eqnarray}
  %%\Cqq & = & x^{q-1}\trixy^q + \tripxy\sum_{i=0}^{q-4}\xpoyp{i}x^{q-i-2}\trixy^{q-i-1} \nonumber\\
  %%     &   & + \trixy\xpoyp{q-3}x\trixy^2 \nonumber\\
  %%     &   & + \xpoyp{q-2}\Big[ \xpop{2}\big( x^3 + \trixyn \big) + y\xpo\trixy \Big]\nonumber\\
%  \Cqq & = & x^{q-1}\trixy^q + \tripxy \nonumber\\
%       &   & \times\sum_{i=0}^{q-4}\xpoyp{i}x^{q-i-2}\trixy^{q-i-1} \nonumber\\
%       &   & + x\xpoyp{q-3}\trixy^3  \nonumber\\
%       &   & + x\xpoyp{q}\big( x^2 + x + 1 \big)  \nonumber\\
%       &   & + y\xpoyp{q-1}\big[ \xpop{2} + y \big] .
%  \label{eq:Cqexpanded}
%\end{eqnarray}
\begin{eqnarray}
  \Cqq & = & x^{q-1}\tri^q + \trip\sum_{i=0}^{q-4}{y'}^{i}x^{q-i-2}\tri^{q-i-1}
              + x{y'}^{q-3}\tri^3 \nonumber\\
       &   & +~x{y'}^{q}\big( x^2 + x + 1 \big) + y{y'}^{q-1}\trip .
  \label{eq:Cqexpanded}
\end{eqnarray}
%\begin{widetext}
%\begin{equation}
%  \Cqq  =  x^{q-1}\tri^q + \trip\sum_{i=0}^{q-4}\xpoyp{i}x^{q-i-2}\tri^{q-i-1}
%           + x\xpoyp{q-3}\tri^3 + x\xpoyp{q}\big( x^2 + x + 1 \big)
%           + y\xpoyp{q-1}\trip .
%  \label{eq:Cqexpanded}
%\end{equation}
%\end{widetext}
In the high temperature $T$ limit, $y\ll x$, so we approximate $y\simeq 0$,
and the equation simplifies to
\begin{eqnarray}
  \Cqq^{(T)} & \simeq & %x^{q-1}x^{q}\xpop{q} + \xpop{2}\sum_{0}^{q-4}\xpop{i}x^{q-i-2}x^{q-i-1}\xpop{q-i-1} \nonumber\\
       %&   & + x\xpo\xpop{q-3}xx^2\xpop{2}  + x\xpop{q}\big( x^2 + x + 1 \big) \nonumber\\
       %& = & x^{2q-1}\xpop{q}     + \xpop{q+1}\sum_{i=0}^{q-4}x^{2q-2i-3}
       %      + x^4\xpop{q}                 + x\xpop{q}\big( x^2 + x + 1 \big)   \nonumber\\
       %& = & x\xpop{q}\left[ x^{2q-2} + \xpo\sum_{i=0}^{q-4}x^{2q-2i-4} + x^3 + x^2 + x + 1 \right] \nonumber\\
       %& = & x\xpop{q}\Big[ x^{2q-2} + x^{2q-3} + x^{2q-4} + \cdots + x^5 + x^4 + x^3 + x^2 + x + 1 \Big] \nonumber\\
              x\xpop{q}\Big[ x^{2q-2} + \cdots + x^2 + x + 1 \Big] \nonumber\\
        & = & x\xpop{q}\left[ \frac{1-x^{2q-1}}{1-x} \right] ,
  \label{eq:CqexpandedhighT}
\end{eqnarray}
We make a final high $T$ approximation
\begin{equation}
  \Cqq^{(T)} \simeq \xpop{q+1} x^{2q-3}
  \label{eq:CqexpandedhighTmore}
\end{equation}
using $\left( x^{2q-1}-1\right) \simeq x^{2q-1}$ and
$\left(1-x\right)^{-1} \simeq \left(1+x\right)/x^2$.

% --- trefoil knot figures ------------------------------------------------
\begin{figure}[t]
\begin{center}%
%\subfigure[]{\raisebox{0.0in}{\includegraphics[width=\subfigwidthnew]{TrefoilKnot}}}
%\subfigure[]{\raisebox{0.0in}{\includegraphics[width=1.2\columnwidth]{TriangleClique}}}
\subfigure[]{\includegraphics[width=0.45\columnwidth]{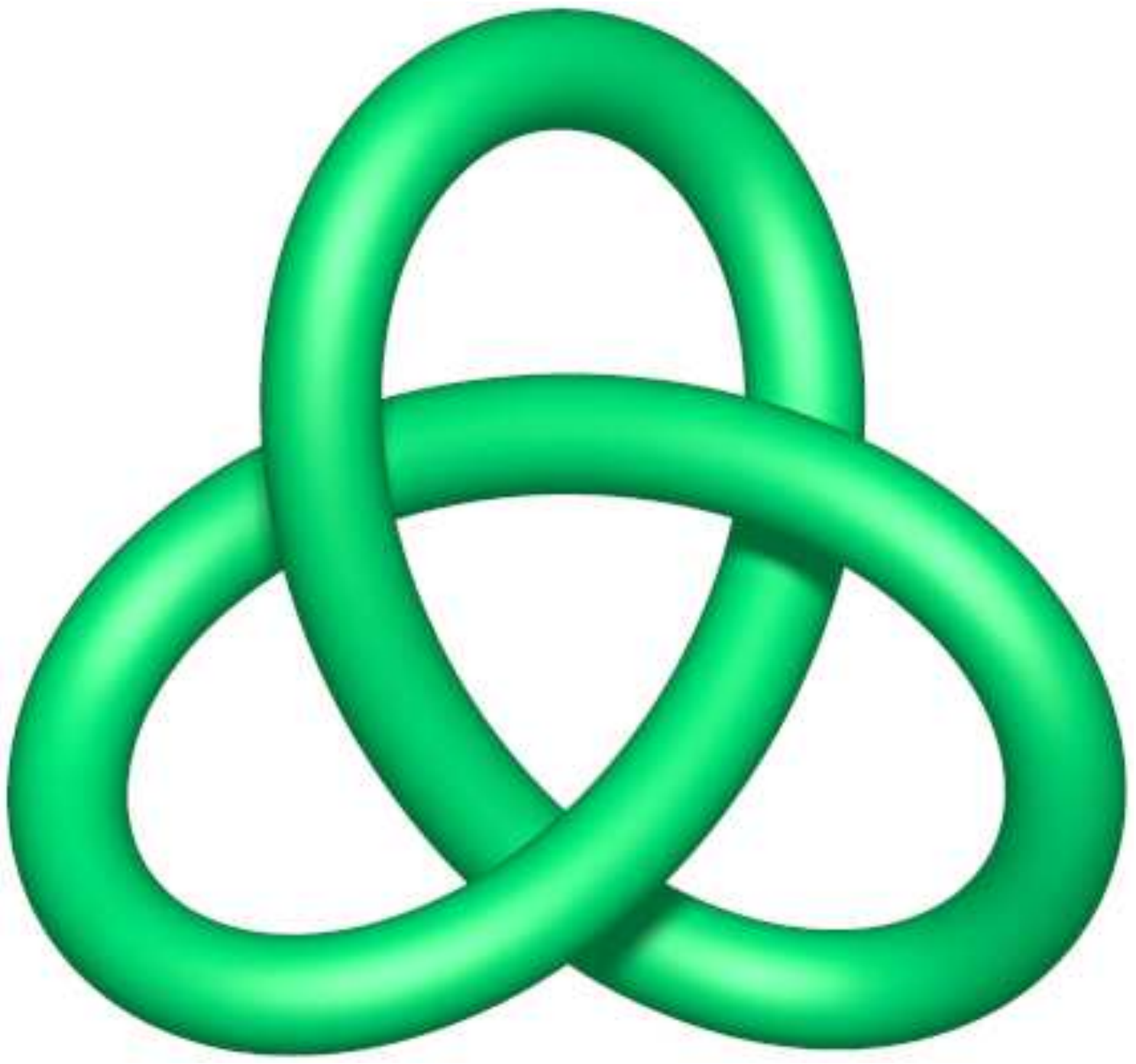}}%
\subfigure[]{\includegraphics[width=0.45\columnwidth]{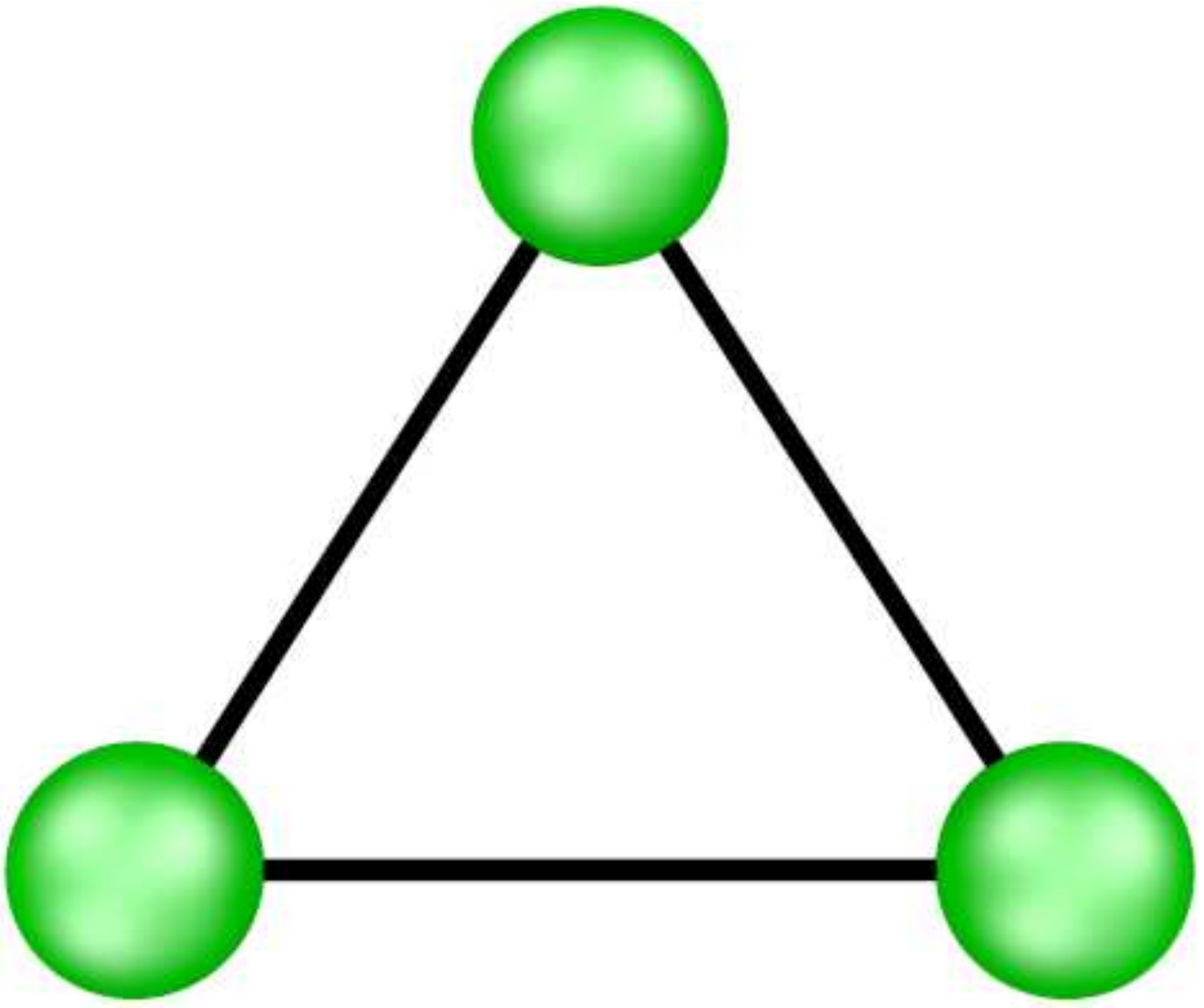}}%
\end{center}%
\caption{(Color online) Panel (a) depicts the trefoil knot, and panel (b)
shows the corresponding graph $G$ constructed from the distinct knot regions
and crossings \cite{ref:kauffmansignedG}.
That is, nodes correspond to ``checkerboard-shaded'' regions (shade the outside
lobes of the trefoil knot leaving the interior region unshaded), and edges
correspond to knot crossings.
Jones polynomials $V_J(x)$ in knot theory are related to Tutte Polynomials,
and \eqnref{eq:trefoilknot} represents the trefoil knot corresponding to the
triangle subgraph in panel (b).}%
\label{fig:trefoilknot}%
\end{figure}
% --- end trefoil knot figures ---------------------------------------------

\section{Random knot ``transitions''} \label{app:trefoilknot}

A general 3D knot may be represented as a $4$-valent planar graph
\cite{ref:kauffmansignedG} [\ie{}, corresponding to a two-dimensional (2D)
square lattice connectivity allowing self-loops].
This relation connects the Tutte polynomial to the Jones polynomial in knot
theory.
Conversely, all connected, signed planar graphs have a corresponding link
diagram representation (2D knot projection).
Alternating over-under crossings result in unsigned planar graphs
\cite{ref:kauffmansignedG} (e.g., the trefoil knot in \figref{fig:trefoilknot}).
Ref.\ \cite{ref:kauffmanknotsintro} provides an introduction to the mathematics
and physics of knot theory.
The Jones polynomial of a given knot is intimately related to quantum field theories
\cite{ref:wittenknots}, via its connection to [an SU(2) type] Wilson loop associated
the same knot.

As a concrete example, \subfigref{fig:trefoilknot}{a} depicts a simple trefoil
knot which is related to the triangle clique depicted in \subfigref{fig:trefoilknot}{b}
\cite{ref:kauffmansignedG}.
The Tutte polynomial of \subfigref{fig:trefoilknot}{b} is $K_3(G;x,y) = x^2 +x + y$.
Then we generate the Jones polynomial %using the relation $xy=1$ resulting in
\begin{equation}
  V_J(x) = x^2 + x + \frac{1}{x}
  \label{eq:trefoilknot}
\end{equation}
where we used $xy=1$ because the trefoil knot has alternating crossings
\cite{ref:thistlewaitespanning}.
While the trefoil knot is clearly not random, we conjecture that the
transitions detected in random graphs with embedded ground states
in the current work can have similar transition repercussions in random knots.

% ********************** END DOCUMENT TEXT **************************
%\bibliographystyle{model1a-num-names}
%\bibliographystyle{nature}
%\bibliography{../mybiblio}        % mybiblio.bib is my BibTeX database

\end{document}